\documentclass[12pt]{article}
\usepackage[margin=1 in]{geometry}
\usepackage{amsmath,amsthm,amssymb}
\usepackage{times}
\usepackage{graphicx}
\usepackage{enumerate}
\usepackage{setspace}
\usepackage[colorlinks,linkcolor=blue,anchorcolor=blue,citecolor=blue,hypertexnames=false]{hyperref}
\usepackage{float}
\usepackage[title]{appendix}
\usepackage{tikz}
\usetikzlibrary{arrows.meta, decorations.markings}
\usetikzlibrary{decorations.pathreplacing}
\usetikzlibrary{calc,3d,matrix}
\usepackage{pgfplots}
\usepgfplotslibrary{fillbetween}
\pgfplotsset{compat = newest}
\usetikzlibrary{patterns}
\usepackage{subcaption}
\usepackage{natbib}
\theoremstyle{plain}

\newtheorem{proposition}{Proposition}
\newtheorem{theorem}{Theorem}
\newtheorem*{theorem*}{Theorem}
\newtheorem{corollary}{Corollary}

\newtheorem{lemma}{Lemma}

\newtheorem{definition}{Definition}
\theoremstyle{remark}
\newtheorem{example}{Example}
\newtheorem{remark}{Remark}

\newcommand{\R}{\mathbb{R}}
\newcommand{\E}{\mathbb{E}}

\newcommand{\EMLRP}{\mathcal{E}^{\mathrm{MLRP}}}
\newcommand{\cx}{\succeq_{\mathrm{cx}}}
\newcommand{\Ss}{\mathbf{S}}

\makeatletter

\newcommand{\Rmnum}[1]{\expandafter\@slowromancap\romannumeral #1@}

\newcommand{\xc}[1]{{***\textcolor{red}{XC: #1}}***}

\date{{November 2026}}

\title{	\textbf{On the Monotonicity of Information Measures}%
\footnote{We thank Arjada Bardhi, Tommaso Denti, Teddy Kim, R. Vijay Krishna, Fei Li, Jeffrey Mensch, Luciano Pomatto,	Todd Sarver, Bruno Strulovici, Curtis Taylor, Jo\~{a}o Thereze, Udayan Vaidya, Kun Zhang, Mu Zhang, and seminar participants at KAEA-VSS Micro Seminar, FSU, UNC, Duke, Mannheim, ITAM, Iowa, Georgetown, UCF, UC Riverside, WashU, SEA 2024, and RUD 2026 for insightful comments. 
}
\vspace{20pt}	
}

\author{Xiaoyu Cheng\footnote{Department of Economics, Florida State University, Tallahassee, FL, USA. E-mail: \href{mailto:xcheng@fsu.edu}{xcheng@fsu.edu}.}\and Yonggyun Kim\footnote{Department of Economics, Florida State University, Tallahassee, FL, USA. E-mail: \href{mailto:ykim22@fsu.edu}{ykim22@fsu.edu}.}}
\begin{document}
\maketitle

\begin{abstract}

We study the monotonicity of information measures, functions that assign a value to each statistical experiment, with respect to the Blackwell and Lehmann orders. We show that such monotonicity can be characterized by local, first-order conditions. Under natural invariance properties, a measure is Blackwell monotone if and only if it decreases in the direction of signal replacements, and Lehmann monotone if and only if it decreases in the direction of reverse signal replacements, which replace signals with higher ones only in low states or with lower ones only in high states. We translate both conditions into the space of posterior beliefs and apply them to likelihood-separable and posterior-separable measures, showing in particular that the entropy cost is Lehmann monotone. Among posterior-separable measures, we further identify a Lehmann monotone subclass, posterior-mean-ratio measures, that characterizes the Lehmann order.

\bigskip



\end{abstract}

\newpage
\section{Introduction}\label{sec:intro}

Measuring information is a basic task across many disciplines. The object being measured is a statistical experiment, which maps states of the world to distributions over signals, and an \emph{information measure} assigns a numerical value to each experiment. In statistics and information theory, such measures quantify how much an experiment reveals, through entropy and divergences \citep{cover1999elements,polyanskiy2022information}; in economics, they represent the cost of acquiring information \citep{sims2003implications,caplin2015revealed,gentzkow2014costly,mackowiak2023rational}, its value to the decision makers who use or buy it \citep{bergemann2018design}, and preferences for early or late resolution of uncertainty \citep{kreps1978temporal,grant1998intrinsic}.

Whatever its interpretation, the most fundamental requirement on an information measure is \emph{monotonicity}: a more informative experiment should be assigned a weakly higher value.\footnote{Throughout the paper, ``monotone,'' ``increasing,'' and ``decreasing'' are understood in the weak sense. We explicitly use ``strictly'' whenever the strict notion is intended.} A measure that violates it is difficult to interpret, whether as a cost, a value, or a divergence. Informativeness, however, is not directly observable and is often defined indrectly through the decisions that information supports. The standard comparison is the order of \citet{blackwell1951comparison,blackwell1953equivalent}: one experiment is more informative than another if every decision maker, whatever her decision problem, weakly prefers it, or equivalently, if the less informative experiment can be obtained from the more informative one by adding noise.

Requiring dominance in every decision problem makes the Blackwell order demanding, leaving many pairs of experiments incomparable. When the relevant decision problems share a monotone structure, as in auctions, screening, and many comparative-statics settings, the order of \citet{lehmann1988comparing} offers a more tailored comparison. It requires dominance only across monotone decision problems, including those with single-crossing or interval-dominance preferences \citep{milgrom1994monotone,quah2009comparative}. It thus ranks more pairs than Blackwell, and these additional comparisons are relevant to the decision problems at hand.

In this paper, we characterize monotonicity under both orders by first-order conditions. In particular, we provide the first general characterization of Lehmann monotonicity. Our characterizations replace comparisons across all ranked pairs with local checks at individual experiments. Verifying monotonicity via direct comparison is difficult because both orders are partial orders on the multidimensional space of experiments, where comparable experiments may lie far apart. To illustrate our approach by analogy, consider the coordinatewise order on $\R^2$. A differentiable function is monotone with respect to this order if and only if both partial derivatives are nonnegative everywhere. Thus, global monotonicity reduces to local sign conditions along two fixed directions. For the Blackwell and Lehmann orders, however, the directions of decreasing informativeness are not fixed but depend on the experiment under consideration. We explicitly characterize a set of directions sufficient to capture global monotonicity: an information measure is monotone if and only if it decreases along these directions.

To illustrate the idea, consider an experiment with two states, $\theta_1<\theta_2$, and two signals, $s_1 < s_2$:
\[
P=\begin{bmatrix}0.8&0.2\\0.3&0.7\end{bmatrix},
\]
where the entry $P_{ij}$ in row $i$ and column $j$ is the probability of observing signal $s_j$ when the state is $\theta_i$.
One way to make $P$ less informative in the Blackwell order is to add noise by replacing $s_1$ by $s_2$ with probability $\epsilon$, regardless of the state, that is, to apply the garbling
\[
M_\epsilon=\begin{bmatrix}1-\epsilon&\epsilon\\0&1\end{bmatrix},
\qquad
PM_\epsilon=P+\epsilon P^{1\to2},
\qquad
P^{1\to2}=\begin{bmatrix}-0.8&0.8\\-0.3&0.3\end{bmatrix}.
\]
For an information measure $\mathcal I$, Blackwell monotonicity requires $\mathcal I(P)\ge\mathcal I(PM_\epsilon)$ for every $\epsilon\in[0,1]$. If $\mathcal I$ is differentiable, dividing by $\epsilon$ and letting $\epsilon\to0$ shows that the derivative of $\mathcal I$ at $P$ in the direction $P^{1\to2}$ must be nonpositive:
\[
\langle\nabla\mathcal I(P),P^{1\to2}\rangle
=0.8\left(\frac{\partial\mathcal I}{\partial P_{12}}-\frac{\partial\mathcal I}{\partial P_{11}}\right)
+0.3\left(\frac{\partial\mathcal I}{\partial P_{22}}-\frac{\partial\mathcal I}{\partial P_{21}}\right)\le0:
\]
shifting a small amount of probability from $s_1$ to $s_2$, in proportion to the likelihoods of $s_1$, decreases the measure. We call this operation a \emph{signal replacement}.

Our first result establishes that, for continuous and differentiable information measures, the first-order inequalities derived from all signal replacements, together with two invariance properties, are sufficient for Blackwell monotonicity for experiments with arbitrary finite state and signal spaces (Theorem~\ref{thm:BM-general}). These invariance properties require that the measure be unchanged by relabeling signals (permutation invariance) or splitting a signal into two copies (split invariance), since neither operation changes informativeness. We prove sufficiency by constructing, for any two Blackwell-ranked experiments, a continuous path from the more informative to the less informative experiment. At every point, the path moves in a nonnegative combination of signal-replacement directions, so the measure decreases along it.

The Lehmann order admits an analogous garbling representation. For experiments satisfying the monotone likelihood ratio property (MLRP), one experiment is Lehmann more informative than another if and only if the latter can be obtained from the former by a \emph{monotone quasi-garbling}, which may add different noise in each state, provided that the noise is more likely to produce higher signals in lower states \citep{kim2023a}. The corresponding elementary operation is a \emph{reverse signal replacement}, which replaces a low signal with a higher one only in states below a cutoff, or a high signal with a lower one only in states above a cutoff. In the example above, replacing $s_1$ by $s_2$ with probability $\epsilon$ only in the low state $\theta_1$ yields $P+\epsilon P^{1\to2}_{\le1}$, where
\[
P^{1\to2}_{\le1}=\begin{bmatrix}-0.8&0.8\\0&0\end{bmatrix}.
\]
This operation makes the high signal more likely in the low state and, for sufficiently small $\epsilon$, weakens the association between high states and high signals while preserving the MLRP. The resulting first-order inequality is
\[
\langle\nabla\mathcal I(P),P^{1\to2}_{\le1}\rangle
=0.8\left(\frac{\partial\mathcal I}{\partial P_{12}}-\frac{\partial\mathcal I}{\partial P_{11}}\right)\le0.
\]
Compared with signal replacements, which act across all states, reverse signal replacements could operate on a tail of states and thus introduce additional directions.\footnote{With two states, as in this example, the Blackwell and Lehmann orders coincide on MLRP experiments \citep{jewitt2007information}, so the inequality above can be derived from signal replacements. The two orders, and hence the two sets of first-order conditions, differ only with three or more states; the example serves only to illustrate the operation.}

Our second result shows that requiring the information measure to decrease along all reverse signal-replacement directions, together with split invariance, characterizes Lehmann monotonicity (Theorem~\ref{thm:LM-general}). A key difficulty in establishing sufficiency is that the MLRP set is nonconvex, so a path constructed by taking convex combinations of two MLRP experiments need not remain within it.\footnote{See Appendix~\ref{appendix:nonconvex-MLRP-set} for an example.} Our main technical contribution is a local-generation characterization of the Lehmann order: Lehmann comparisons can be generated by paths that remain within the MLRP set and whose local movements are generated by reverse signal replacements (Proposition~\ref{prop:LM-generation}). This path construction reduces global comparisons to local ones, with the first-order inequalities ensuring that the measure decreases along each path.

Although our characterizations are stated in terms of signal likelihoods, many applications—notably rational inattention, information design, and preferences over information—work instead with the distribution of posterior beliefs induced by an experiment under a fixed prior. We show that a direct chain-rule argument translates our first-order conditions into this posterior representation (Theorem~\ref{thm:posterior-foc-direct}).

We then apply our characterizations to classes of information measures widely used in practice. For likelihood-separable measures, Blackwell monotonicity is equivalent to sublinearity of the generating function, whereas Lehmann monotonicity is more demanding: some sublinear generators violate it, while weighted $p$-norms with $p>1$ satisfy it (Propositions~\ref{prop:likelihood-separable} and~\ref{prop:likelihood-separable-p-norm}). For posterior-separable measures, Blackwell monotonicity is equivalent to convexity of the generating function (Corollary~\ref{cor:ps-blackwell}), while Lehmann monotonicity reduces to a comparison of Bregman divergences evaluated at truncated posteriors (Proposition~\ref{prop:ps-lehmann}); in particular, the entropy cost of \citet{sims2003implications} is Lehmann monotone (Corollary~\ref{cor:mutual-info}). In posterior space, each condition compares the marginal value of a transferred block of probability at the two signals involved. Expressing these marginal values through local utilities recovers the Blackwell-monotonicity characterization of \citet{grant1998intrinsic} and yields a corresponding characterization of Lehmann monotonicity (Remark~\ref{rem:local-utility}).

Beyond characterizing the monotonicity of information measures, our results also yield a characterization of the Lehmann order itself. Parallel to the classical result that the Blackwell order coincides with the convex order on posterior beliefs, we show that one experiment is Lehmann more informative than another if and only if it is ranked higher by every measure in a subclass of posterior-separable measures that we call \emph{posterior-mean-ratio measures} (Theorem~\ref{thm:lehmann-cx}). The proof relies on Theorem~\ref{thm:LM-general}: establishing Lehmann monotonicity of posterior-mean-ratio measures reduces to checking the first-order conditions at each experiment. This class includes the variance of posterior predictions of increasing state payoffs and every $f$-divergence between the signal distributions of two states, including the log-likelihood ratio costs of \citet{pomatto2023a}. Theorem~\ref{thm:lehmann-cx} also recovers the characterizations of Lehmann order by \citet{kim1995efficiency}, \citet{Chi2015}, and \citet{asseyer2025information}.

\paragraph{Related Literature}
The comparison of statistical experiments has a long history in statistics, beginning with \citet{blackwell1951comparison,blackwell1953equivalent} and developed by \citet{lecam1964sufficiency} and \citet{torgersen1991comparison}, who also study functionals of experiments and their connection to majorization. This connection gives our question a classical counterpart. Majorization is the Blackwell order restricted to two-state experiments in which one state's signal distribution is uniform, and for majorization, the Schur--Ostrowski criterion characterizes the monotone (Schur-convex) functions by symmetry together with a first-order condition \citep[Theorem~3.A.4]{marshall2011inequalities}. Theorem~\ref{thm:BM-general} extends this type of characterization to experiments with any finite number of states and signals, and Theorem~\ref{thm:LM-general} provides its counterpart for the Lehmann order.

In the literature on stochastic orders, orders are often characterized by the classes of functions whose expectations they preserve \citep{shaked2007stochastic}; the classical instance relevant here is the equivalence between the Blackwell order and the convex order on posterior beliefs \citep{blackwell1953equivalent,strassen1965existence,cartier1964comparaison}. Section~\ref{sec:cx-lehmann} provides an analogous characterization of the Lehmann order.

In economics, Blackwell monotonicity is the standard consistency requirement for information measures, but it is mostly used as an ingredient rather than studied in isolation. In the axiomatic literature on the cost of information, it is combined with other properties to characterize specific functional forms \citep{mensch2018cardinal,hebert2021neighborhood,pomatto2023a,baker2023,bloedel2021cost,bordoli2025convex}. In the decision-theoretic literature, it emerges as a necessary consequence of rational, cost-minimizing behavior when information costs are identified from choice data \citep{deOliveira2017rationally,chambers2020costly}.\footnote{This revealed preference framework has been broadly applied to derive testable implications of stochastic choice models and to characterize the behavioral foundations of various cost structures (e.g., \citealp{caplin2015revealed, denti2022posterior, lipnowski2022predicting}).} Several first-order results are closely related to Theorem~\ref{thm:BM-general}. \citet{grant1998intrinsic} show that, for preferences over information with a G\^ateaux-differentiable representation, Blackwell monotonicity is equivalent to convexity of the local utilities; we recover this criterion in Remark~\ref{rem:local-utility}. Claim~1 of \citet{ravid2022learning} likewise relies on directional derivatives, but for functionals of the distribution of posterior means, which disregard any information not captured in the first moment of posterior beliefs, so the two characterizations are not directly translatable. For likelihood- and posterior-separable measures, Blackwell monotonicity is known to be equivalent to sublinearity and convexity of the generating function, respectively \citep{caplin2015revealed,denti2022posterior,denti2022experimentalorder,baker2023}. Our contribution for the Blackwell order is therefore a unified first-order characterization that applies to continuous and differentiable measures and serves as the benchmark for the Lehmann case.

The order of \citet{lehmann1988comparing} was introduced for monotone decision problems in the sense of \citet{karlin1956theory}, in which each state has an ideal action, increasing in the state, and payoffs decrease as the action moves away from it. The comparison extends to problems in which payoffs satisfy single-crossing differences \citep{milgrom1994monotone} or the interval dominance order \citep{quah2009comparative}, and \citet{kim2023a} identifies a general class of monotone decision problems for which it is the relevant criterion.\footnote{\citet{athey2018value} take the converse approach: they fix classes of monotone decision problems and characterize the comparisons of information that every decision maker in each class values.} The order has been applied in auctions \citep{persico2000information,bobkova2024information,frick2026auctions}, finance \citep{bond2023ordering}, screening \citep{kim2023a,asseyer2025information}, and grantmaking \citep{adda2024grantmaking}, and extends to decision makers who are averse to ambiguity \citep{li2020information}.

Existing characterizations of the Lehmann order describe which experiments are ranked rather than which measures respect the ranking. \citet{jewitt2007information} characterizes the order through Blackwell comparisons on pairs of states. In moral hazard, \citet{kim1995efficiency} ranks information systems by a mean-preserving spread of the likelihood ratio, a criterion that, under the MLRP, is equivalent to the Lehmann order \citep{jewitt2007information,kim2023a}. \citet{Chi2015} and \citet{asseyer2025information} characterize the order through convex-order comparisons of posterior predictions and of conditional inverse hazard rates, respectively. Verifying that a given measure is Lehmann monotone through these results still requires comparing every pair of Lehmann-ranked experiments. Theorem~\ref{thm:LM-general} instead provides a general criterion based on first-order conditions at each experiment. It shows, for instance, that sublinearity of the generating function for likelihood-separable measures and convexity for posterior-separable measures, which ensure Blackwell monotonicity, do not ensure Lehmann monotonicity (Propositions~\ref{prop:likelihood-separable} and~\ref{prop:ps-lehmann}).

\section{Preliminaries}\label{sec:preliminaries}
Let $\Theta = \{\theta_{1}, \ldots, \theta_{n}\}$ be a finite set of states with $n \geq 2$, and for each $m \geq 2$, let $\mathcal{S}_{m} = \{s_{1}, \ldots, s_{m}\}$ be a finite set of signals. A \emph{(statistical) experiment} with $m$ signals is a mapping $P: \Theta \rightarrow \Delta(\mathcal{S}_{m})$, which we identify with an $n \times m$ matrix
\begin{equation*}
P = \begin{bmatrix}
P_{11} & \cdots & P_{1m} \\
\vdots & \ddots & \vdots \\
P_{n1} & \cdots & P_{nm}
\end{bmatrix},
\end{equation*}
where $P_{ij} = P(s_{j}|\theta_{i})$ is the probability of observing signal $s_{j}$ in state $\theta_{i}$ with $\sum_{j=1}^{m} P_{ij} = 1$ for all $i$. We write $P_{i} = (P_{i1}, \ldots, P_{im})$ for the $i$-th row of $P$, the distribution of signals conditional on state $\theta_{i}$, and $P^{j} = (P_{1j}, \ldots, P_{nj})^{\intercal}$ for the $j$-th column of $P$, the vector of likelihoods of signal $s_{j}$ across states. Let $\mathcal{E}_m \subseteq \R^{n\times m}$ denote the set of all experiments with signal space $\mathcal{S}_{m}$ and let $\mathcal{E} = \bigcup_{m \ge 2} \mathcal{E}_m$ denote the set of all finite experiments.

\subsection{Information Orders}

We compare experiments using two standard information orders: the Blackwell order \citep{blackwell1951comparison,blackwell1953equivalent} and the Lehmann order \citep{lehmann1988comparing}. The two orders are tailored to different classes of decision problems: the Blackwell order to all decision problems, and the Lehmann order to \emph{monotone} ones, in which states and actions are ordered and higher actions are more valuable in higher states. Both can be defined through garblings. A \emph{garbling matrix} is a stochastic matrix $M$, i.e., $M_{jk} \ge 0$ and $\sum_{k} M_{jk} = 1$ for all $j$; applying an $m \times m'$ garbling matrix $M$ to $P \in \mathcal{E}_m$ in state $\theta_i$ yields the signal distribution $P_i M$.

Since monotone decision problems presuppose an ordering of states and signals, the Lehmann order is defined for experiments satisfying the \emph{monotone likelihood ratio property} (MLRP). For nonnegative vectors $x,y\in\R^n_+$, write $y\ge_{\mathrm{MLR}}x$ if
\begin{equation}
x_iy_{i'}\ge x_{i'}y_i\qquad\text{for all } i' > i. \tag{MLR} \label{eq:MLR-scales}
\end{equation}
When $x$ is positive, this means that $y_i/x_i$ is increasing in $i$. With states and signals ordered by their indices, an experiment $P\in\mathcal E_m$ satisfies the MLRP if
\begin{equation}
P^{j'}\ge_{\mathrm{MLR}}P^{j}\qquad\text{for all } j' > j   . \tag{MLRP} \label{eq:MLRP}
\end{equation}
Let $\EMLRP_m \subseteq \mathcal{E}_{m}$ denote the set of MLRP experiments with $m$ signals and $\EMLRP = \bigcup_{m \ge 2} \EMLRP_m$.

\begin{definition}[Blackwell and Lehmann orders]\label{def:orders}
\leavevmode
\begin{enumerate}[(a)]
\item For $P \in \mathcal{E}_m$ and $Q \in \mathcal{E}_{m'}$, $P$ is \textbf{Blackwell more informative} than $Q$, denoted $P \succeq_{B} Q$, if there exists an $m \times m'$ garbling matrix $M$ such that
\[
Q_i = P_i M \quad \text{for all } i, \quad \text{i.e., } Q = PM.
\]
\item For $P \in \EMLRP_m$ and $Q \in \EMLRP_{m'}$, $P$ is \textbf{Lehmann more informative} than $Q$, denoted $P \succeq_{L} Q$, if there exist $m \times m'$ garbling matrices $M(1), \ldots, M(n)$ such that
\[
Q_i = P_i M(i) \quad \text{for all } i,
\]
and $M(1), \ldots, M(n)$ are \emph{reversely monotone}: $\sum_{k \le \ell} M(i)_{jk}$ is increasing in $i$ for all $j$ and $\ell$. In this case, $Q$ is called a \emph{monotone quasi-garbling} of $P$.\footnote{\citet{lehmann1988comparing} originally defines the order through a quantile condition on continuous signal spaces. \citet[Theorem~1]{kim2023a} shows that, for MLRP experiments, it is equivalent to this definition.}
\end{enumerate}
We say that $P$ and $Q$ are \textbf{Blackwell equivalent} (resp.\ \textbf{Lehmann equivalent}), denoted $P \simeq_{B} Q$ (resp.\ $P \simeq_{L} Q$), if both $P \succeq_{B} Q$ and $Q \succeq_{B} P$ (resp.\ $P \succeq_{L} Q$ and $Q \succeq_{L} P$) hold.
\end{definition}

The two orders differ only in whether the noise may depend on the state. A Blackwell garbling adds the same noise $M$ in every state. A monotone quasi-garbling, introduced by \citet{kim2023a}, allows a different garbling matrix in each state, provided that, for each signal $s_j$, the $j$-th rows of $M(i)$ decrease in first-order stochastic dominance (FOSD) as $i$ increases. The added noise is then more likely to produce higher signals in lower states, so it can only weaken the association between high states and high signals. Setting $M(1) = \cdots = M(n) = M$ satisfies reverse monotonicity trivially, so for MLRP experiments $P \succeq_{B} Q$ implies $P \succeq_{L} Q$; the converse fails because the state-dependent noise allowed in (b) cannot always be replicated by a single garbling.

Both orders admit value characterizations: $P \succeq_{B} Q$ if and only if $P$ yields a weakly higher expected payoff than $Q$ in every decision problem \citep{blackwell1951comparison,blackwell1953equivalent}, and $P \succeq_{L} Q$ if and only if it does so in every monotone decision problem \citep{lehmann1988comparing}. Since monotone decision problems form a subclass of all decision problems, the Lehmann order ranks more pairs.

\subsection{Information Measures}\label{sec:measures}

The central object of our paper is an \emph{information measure}, a function that assigns a numerical value to each experiment.
Formally, an information measure $\mathcal I:\mathcal E\to\R$ is defined by a collection $(\mathcal I_m)_{m\ge2}$ of functions $\mathcal I_m:\mathcal E_m\to\R$, with $\mathcal I(P)=\mathcal I_m(P)$ for $P\in\mathcal E_m$.

We define monotonicity of the information measure according to both information orders.

\begin{definition}\label{def:blackwell_lehmann_monotonicity}
An information measure $\mathcal{I}$ is \textbf{Blackwell monotone} if $\mathcal{I}(P) \geq \mathcal{I}(Q)$ for all $P, Q \in \mathcal{E}$ with $P \succeq_{B} Q$, and \textbf{Lehmann monotone} if $\mathcal{I}(P) \geq \mathcal{I}(Q)$ for all $P,Q \in \EMLRP$ with $P \succeq_{L} Q$.
\end{definition}

Since $\mathcal I$ is defined through the collection $(\mathcal I_m)_{m\ge2}$, we impose regularity assumptions separately for each $m$. We say that $\mathcal I$ is continuous if every $\mathcal I_m$ is continuous on $\mathcal E_m$, and differentiable if every $\mathcal I_m$ is differentiable on the relative interior of $\mathcal E_m$ in the following sense. The relative interior of $\mathcal E_m$ is the set of all \emph{positive experiments}:
\begin{align*}
\mathcal E_{m,+}:=\operatorname{relint}(\mathcal E_m)=\{P\in\mathcal E_m:P_{ij}>0\text{ for all }i,j\}.
\end{align*}
These are the experiments at which every perturbation preserving the row sums is feasible: letting
\begin{align*}
	T_m=\left\{H\in\R^{n\times m}:\sum_{j=1}^m H_{ij}=0\text{ for every }i\right\}
\end{align*}
denote the space of perturbation directions, for any $P\in\mathcal E_{m,+}$ and $H\in T_m$, $P+tH$ remains an experiment for all sufficiently small $|t|$. We say $\mathcal I_m$ is differentiable if it is differentiable on $\mathcal E_{m,+}$ relative to the affine hull of $\mathcal E_m$. Let $\nabla\mathcal I_m(P)\in\R^{n\times m}$ be a gradient, with $(i,j)$-entry denoted by $\partial\mathcal I_m(P)/\partial P_{ij}$, so that the directional derivative along any $H\in T_m$ is $\langle\nabla\mathcal I_m(P),H\rangle$.\footnote{While we consider $\mathcal I_m$ on $\mathcal E_m$, its functional form may naturally be defined on a larger domain in $\R^{n\times m}$. Any two such differentiable functions that agree on $\mathcal E_m$ may have different gradients, but they give the same inner products with every $H\in T_m$. Hence, their ordinary partial derivatives may be used directly to calculate directional derivatives.} Henceforth, we write $\nabla\mathcal I$ with the understanding that it is defined for each $\mathcal E_m$ through $\mathcal I_m$.

Finally, let $\EMLRP_{m,+}=\EMLRP_m\cap\mathcal E_{m,+}$ denote the set of positive MLRP experiments with $m$ signals, and let $\EMLRP_{+}=\bigcup_{m\ge2}\EMLRP_{m,+}$.
\section{Blackwell Monotonicity}\label{sec:blackwell}

This section characterizes Blackwell monotonicity through two invariance requirements and a local condition on signal replacements. We first introduce the invariance requirements, which follow from operations that preserve informativeness.

\paragraph{Permutation Invariance.} Relabeling signals does not change the informativeness of an experiment. Formally, for $P \in \mathcal{E}_{m}$, a permutation matrix $\Pi$ is an $m \times m$ garbling matrix with exactly one nonzero entry in each row and each column. Its inverse, $\Pi^{-1}$, is also a permutation matrix and we have $P \succeq_{B} P\, \Pi \succeq_B P\, \Pi\, \Pi^{-1} = P$. Thus, $P \simeq_{B} P\, \Pi$. 

Say that an information measure $\mathcal{I}$ is \emph{permutation invariant} if for all $P \in \mathcal{E}$ and all permutation matrices $\Pi$, $\mathcal{I}(P) = \mathcal{I}(P\, \Pi)$.

\paragraph{Split Invariance.} Consider a garbling operation that, upon observing signal $s_j$, preserves the signal with probability $1-\lambda$ and relabels it as a new, distinct signal with probability $\lambda$. This new experiment is Blackwell equivalent to the original one, since a garbling that merges the two signals together recovers the original experiment. In matrix terms, for $P \in \mathcal{E}_{m}$, this corresponds to splitting the column $P^j$ into two columns, $(1-\lambda) P^j$ and $\lambda P^j$, resulting in the following experiment
\begin{align}
	P' = [P^1, \cdots, (1-\lambda) P^j, \lambda P^j, \cdots, P^m] \in \mathcal{E}_{m+1}. \label{eq:split}
\end{align}

Say that an information measure $\mathcal{I}$ is \emph{split invariant} if $\mathcal{I}(P) = \mathcal{I}(P')$ for all $P$ and $P'$ as in \eqref{eq:split} with $\lambda \in [0,1]$. Note that when $\lambda \in \{0,1\}$, the split operation simply adds a column $ \mathbf{0} $ to $P$, effectively embedding $P$ into $\mathcal{E}_{m+1}$. Thus, split invariance also ensures that the information measure remains unchanged when an unobserved signal is added to an experiment.

We next introduce an operation that reduces Blackwell informativeness: a garbling that replaces a signal $s_{j}$ with another signal $s_{k}$ with probability $\epsilon$, leaving all other signals unchanged.

\begin{definition}\label{def:signal-replacement}
For $1\le j\neq k\le m$ and $\epsilon\in[0,1]$, the \textbf{signal replacement} matrix $M^{j \to k}_\epsilon$ is the $m\times m$ garbling matrix whose rows coincide with those of $I_m$ except the $j$-th, which is
\begin{equation*}
\bigl(0,\ldots,0,\underbrace{1-\epsilon}_{j\text{-th}},0,\ldots,0,\underbrace{\epsilon}_{k\text{-th}},0,\ldots,0\bigr).
\end{equation*}
For $P\in\mathcal E_m$, applying it yields $P M^{j\to k}_\epsilon = P + \epsilon P^{j\to k}$, where
\begin{equation*}
P^{j \rightarrow k } :=
\begin{bmatrix}
0 & \cdots & \smash{\underbrace{-P^{j}}_{j\text{-th}}} & \cdots &0 & \cdots & \smash{\underbrace{P^{j}}_{k\text{-th}}} & \cdots & 0
\end{bmatrix} \in T_{m}.
\end{equation*}
\end{definition}

If the information measure is Blackwell monotone, its value must decrease after a signal replacement: $\mathcal{I}(P + \epsilon P^{j\to k}) \le \mathcal{I}(P)$ for all $\epsilon \in [0,1]$. When the information measure is differentiable, letting $\epsilon \downarrow 0$ yields
\begin{align*}
	\lim_{\epsilon \downarrow 0} \frac{\mathcal{I}(P + \epsilon P^{j\to k}) - \mathcal{I}(P)}{\epsilon} = \langle \nabla \mathcal{I}(P), P^{j \rightarrow k }\rangle \leq 0.
\end{align*}
This motivates the following definition.

\begin{definition}\label{def:decreasing_signal_replacement_general}
A differentiable information measure $\mathcal{I}:\mathcal{E} \rightarrow \R$ is said to be \textbf{decreasing in signal replacement} if, for all $m \geq 2$, $P\in\mathcal E_{m,+}$, and $1\le j\neq k\le m$,   
\begin{equation}\label{ineq:decreasing-signal-replacement}
\langle \nabla \mathcal{I}(P), P^{j \rightarrow k} \rangle  = \sum_{i=1}^{n} \left ( -\frac{\partial \mathcal{I}(P)}{\partial P_{ij}} + \frac{\partial \mathcal{I}(P)}{\partial P_{ik}} \right ) P_{ij} \leq 0.
\end{equation}
\end{definition}

The next result establishes that together with the two invariances, this condition is also sufficient for Blackwell monotonicity. 

\begin{theorem}\label{thm:BM-general}
A continuous and differentiable information measure $\mathcal{I}: \mathcal{E} \rightarrow \R$ is Blackwell monotone if and only if it is permutation invariant, split invariant, and decreasing in signal replacement. 
\end{theorem}

The essence of the proof of Theorem \ref{thm:BM-general} is to show how the local condition of decreasing in signal replacement can be used to establish a comparison of the values assigned to any two Blackwell-comparable experiments. Specifically, let $P \in \mathcal{E}_{m,+}$ and $Q = P M \in \mathcal{E}_{m',+}$ for some garbling matrix $M$, thus $P \succeq_{B} Q$. By split and permutation invariance, we can embed both experiments into $\mathcal{E}_{m + m'}$ with disjoint sets of signals, $\hat{P} = [P, \mathbf{0}]$ and $\hat{Q} = [\mathbf{0}, Q]$, without changing their values under $\mathcal{I}$. Consider the straight path $\phi(t) = (1 -t) \hat{P} + t \hat{Q}$ for $t \in (0,1)$. Its direction is 
\begin{align*}
	\hat{Q} - \hat{P} = [-P, P M] = \sum_{j = 1}^{m}\sum_{k=1}^{m'}\frac{M_{jk}}{1-t}\phi(t)^{j \rightarrow m + k}, 
\end{align*}
a nonnegative combination of signal replacements. Since $\phi(t) \in \mathcal{E}_{m + m', +}$ for all $t \in (0,1)$, decreasing in signal replacement implies
\begin{align*}
	\frac{d}{dt}\mathcal{I}(\phi(t)) = \sum_{j = 1}^{m}\sum_{k=1}^{m'}\frac{M_{jk}}{1-t}\langle \nabla \mathcal{I}(\phi(t)), \phi(t)^{j \rightarrow m + k}\rangle \leq 0.
\end{align*}
The mean value theorem and the continuity of $\mathcal{I}$ then yield $\mathcal{I}(P) \geq \mathcal{I}(Q)$. The conclusion extends to all $P \in \mathcal{E}_{m}$ and $Q \in \mathcal{E}_{m'}$ by continuity.

\begin{remark}
Theorem \ref{thm:BM-general} provides a characterization of Blackwell monotonicity when the information measure $\mathcal{I}$ is defined over all finite experiments. In some applications, however, it may be desirable to consider information measures defined only on experiments with a fixed number of signals, i.e., a measure $\mathcal{I}_{m}$ defined only over $\mathcal{E}_{m}$. In this case, Theorem \ref{thm:BM-general} no longer applies because its proof relies on embedding experiments into higher-dimensional spaces where information measures must also be defined. In Section \ref{sec:quasi-convex}, we show that a signal-replacement characterization of Blackwell monotonicity can still be obtained when $\mathcal{I}_{m}$ satisfies quasi-convexity. 
\end{remark}

\section{Lehmann Monotonicity}\label{sec:lehmann}

The Blackwell analysis illustrates how local restrictions on an information measure can characterize global monotonicity. In particular, any Blackwell comparison decomposes into a nonnegative combination of signal replacements along a straight path (Theorem~\ref{thm:BM-general}). This section establishes a parallel result for the Lehmann order. Because the Lehmann order is defined only on the MLRP set, which is nonconvex, a straight path connecting two Lehmann-ranked experiments need not lie entirely within this set. Instead, we show that Lehmann comparisons can be generated by local operations along paths that preserve the MLRP throughout (Proposition~\ref{prop:LM-generation}), which yields our first-order characterization of Lehmann monotonicity.

\subsection{Characterization by Reverse Signal Replacements}

For Blackwell, signal replacements form a class of garblings that is rich enough to generate every Blackwell comparison, yet simple enough to yield a local condition. For Lehmann, we look for the analogous class among monotone quasi-garblings (Definition~\ref{def:orders}(b)). Since a garbling is the special case $M(1) = \cdots = M(n)$, the natural candidates are signal replacements whose garbling matrices are allowed to differ across states in a reversely monotone way.

Consider an operation that replaces a signal $s_{j}$ with a higher adjacent signal $s_{j+1}$ with probability $\epsilon$, but only in states below a cutoff state $\theta_\ell$. We call this and its counterpart that replaces a signal with a lower adjacent signal in states above a cutoff \emph{reverse signal replacements}. The following definition formalizes these operations as monotone quasi-garblings.

For $P \in \mathcal{E}_m$ and a cutoff state $\theta_\ell$, write
\begin{equation*}
	P^j_{\le \ell} = (P_{1j}, \ldots, P_{\ell j}, 0, \ldots, 0)^\intercal \qquad \text{and} \qquad P^j_{\ge \ell} = (0, \ldots, 0, P_{\ell j}, \ldots, P_{nj})^\intercal.
\end{equation*}

\begin{definition}\label{def:reverse-signal-replacement}
Fix a cutoff $1 \le \ell \le n$ and $\epsilon \in [0,1]$.
\begin{enumerate}[(a)]
\item For $1 \le j \le m-1$, the \textbf{reverse signal replacement} of $s_j$ with $s_{j+1}$ below $\theta_\ell$ is the monotone quasi-garbling with
\[
M(i) = M_{\epsilon}^{j \rightarrow j+1} ~~ \text{ for } i \leq \ell, \qquad M(i) = I_{m} ~~~ \text{ for }~ i > \ell.
\]
For $P \in \mathcal{E}_m$, it yields $P + \epsilon P^{j \rightarrow j+1}_{\leq \ell}$, where $P^{j \rightarrow j+1}_{\leq \ell} := [\mathbf{0}, \ldots, \mathbf{0}, -P^j_{\le \ell}, P^j_{\le \ell}, \mathbf{0}, \ldots, \mathbf{0}]$, with $-P^j_{\le \ell}$ in the $j$-th column.
\item For $2 \le j \le m$, the \textbf{reverse signal replacement} of $s_j$ with $s_{j-1}$ above $\theta_\ell$ is the monotone quasi-garbling with
\[
M(i) = I_{m} ~~ \text{ for } i < \ell, \qquad M(i) = M_{\epsilon}^{j \rightarrow j-1} ~~~ \text{ for }~ i \geq \ell.
\]
For $P \in \mathcal{E}_m$, it yields $P + \epsilon P^{j \rightarrow j-1}_{\geq \ell}$, where $P^{j \rightarrow j-1}_{\geq \ell} := [\mathbf{0}, \ldots, \mathbf{0}, P^j_{\ge \ell}, -P^j_{\ge \ell}, \mathbf{0}, \ldots, \mathbf{0}]$, with $-P^j_{\ge \ell}$ in the $j$-th column.
\end{enumerate}
\end{definition}

To see they are indeed monotone quasi-garblings, note every row other than the $j$-th is the same across states, and in the $j$-th row the cumulative probability $\sum_{q \le r} M(i)_{jq}$ is increasing in $i$: in (a), it equals $1-\epsilon$ for $i \le \ell$ and $1$ for $i > \ell$ when $r = j$; in (b), it equals $0$ for $i < \ell$ and $\epsilon$ for $i \ge \ell$ when $r = j-1$; for all other $r$, it does not depend on $i$.

Reverse signal replacements reduce to signal replacements when $\ell = n$ in (a) or $\ell = 1$ in (b), but they exclude replacements between non-adjacent signals. This restriction helps preserve the MLRP: monotone quasi-garblings, including ordinary garblings, need not preserve the MLRP. Therefore, to derive a local necessary condition for Lehmann monotonicity, we need replacements that remain in this set. When the replacement applies to all states, no condition is needed: signal replacements between adjacent signals always preserve the MLRP. With a cutoff, the following lemma shows that, at a positive experiment, strict likelihood-ratio ordering across the cutoff ensures that all sufficiently small replacements preserve the MLRP and hence reduce Lehmann informativeness. 

\begin{lemma}\label{lem:Lehmann-marginal-general}
Suppose that $P \in \EMLRP_{m,+}$.
For any $1 \le \ell \le n$ and $1\le j \le m-1$, if $P_{\ell j}P_{\ell+1,j+1} > P_{\ell+1,j}P_{\ell,j+1}$,
then there exists $\epsilon' \in (0,1]$ such that $P \succeq_L P + \epsilon P^{j\rightarrow j+1}_{\le \ell} \in \EMLRP_m$ for all $\epsilon \in [0,\epsilon']$.
Similarly, for any $1 \le \ell \le n$ and $2 \le j \le m$, if $P_{\ell-1,j-1}P_{\ell j} > P_{\ell,j-1}P_{\ell-1,j}$,
then there exists $\epsilon'' \in (0,1]$ such that $P \succeq_L P + \epsilon P^{j\rightarrow j-1}_{\ge \ell} \in \EMLRP_m$ for all $\epsilon \in [0,\epsilon'']$.\footnote{In this lemma and Definition~\ref{def:decreasing_reverse_signal_general}, for the relevant adjacent pair, set $(P_{n+1,j},P_{n+1,j+1})=(0,1)$ at the lower-tail cutoff $\ell=n$, and $(P_{0,j-1},P_{0,j})=(1,0)$ at the upper-tail cutoff $\ell=1$. With these conventions, one may take $\epsilon'=\frac{P_{\ell j}P_{\ell+1,j+1}-P_{\ell,j+1}P_{\ell+1,j}}{P_{\ell j}(P_{\ell+1,j}+P_{\ell+1,j+1})}$. This equals $1$ when $\ell=n$; for $\ell<n$, the likelihood-ratio inequality across the cutoff becomes an equality at this bound. The upper-tail bound is analogous.}
\end{lemma}

When the strict likelihood-product inequality in Lemma~\ref{lem:Lehmann-marginal-general} holds, the replacement reduces Lehmann informativeness for all sufficiently small $\epsilon$, and letting $\epsilon\downarrow 0$ yields the following necessary first-order condition for Lehmann monotonicity. 

\begin{definition}\label{def:decreasing_reverse_signal_general}
A differentiable information measure $\mathcal{I}:\mathcal E\to\R$ is said to be \textbf{decreasing in reverse signal replacement} if, for all $m \ge 2$, $P \in \EMLRP_{m,+}$, and $1 \le \ell \le n$,
\begin{enumerate}[(a)]
\item if $1 \le j \le m-1$ and $P_{\ell j}P_{\ell+1,j+1} > P_{\ell+1,j}P_{\ell,j+1}$,
\begin{equation}
\langle \nabla \mathcal{I}(P), P^{j\rightarrow j+1}_{\le \ell} \rangle = \sum_{i=1}^{\ell} \left ( -\frac{\partial \mathcal{I}(P)}{\partial P_{ij}} + \frac{\partial \mathcal{I}(P)}{\partial P_{i,j+1}} \right ) P_{ij} \leq 0; \label{ineq:Lehmann-finite-1}
\end{equation}
\item if $2 \le j \le m$ and $P_{\ell-1,j-1}P_{\ell j} > P_{\ell,j-1}P_{\ell-1,j}$,
\begin{equation}
\langle \nabla \mathcal{I}(P), P^{j\rightarrow j-1}_{\ge \ell} \rangle = \sum_{i=\ell}^{n} \left ( -\frac{\partial \mathcal{I}(P)}{\partial P_{ij}} + \frac{\partial \mathcal{I}(P)}{\partial P_{i,j-1}} \right ) P_{ij} \leq 0.
\label{ineq:Lehmann-finite-2}
\end{equation}
\end{enumerate}
\end{definition}

The next theorem establishes that this condition, together with split invariance on $\EMLRP$, is both necessary and sufficient for Lehmann monotonicity. 

\begin{theorem}\label{thm:LM-general}
A continuous and differentiable information measure $\mathcal{I}: \mathcal{E} \rightarrow \R$ is Lehmann monotone if and only if it is split invariant on $\EMLRP$ and decreasing in reverse signal replacement.
\end{theorem}

The discussion above establishes the necessity of decreasing in reverse signal replacement. For split invariance, notice that for any $P \in \EMLRP_{m}$, the experiment $P'$ as defined in \eqref{eq:split} also satisfies the MLRP for all $\lambda \in [0,1]$. Since $P$ and $P'$ are Blackwell equivalent and the Blackwell order implies the Lehmann order on MLRP experiments, they are also Lehmann equivalent. Thus, split invariance over MLRP experiments is necessary. On the other hand, since the MLRP is generally violated under permutations, permutation invariance is not relevant here. 

It remains to show that the local reverse-signal replacement conditions, together with split invariance, imply the comparison of values for every pair of Lehmann-ranked experiments. We provide a detailed proof sketch in the following subsection. 

\subsection{Generating Lehmann Comparisons by Reverse Signal Replacements}\label{sec:LM-general-sketch}

For proving sufficiency, the core argument is the following local-generation characterization of Lehmann order: every Lehmann comparison between positive experiments is generated, up to splitting and merging signals, by reverse signal replacements along MLRP-preserving paths. This characterization concerns the Lehmann order itself and is independent of any information measure.

\begin{proposition}[Local generation of the Lehmann order]\label{prop:LM-generation}
Let $P,Q\in\EMLRP_+$ with $P\succeq_L Q$. Then there exist $P=R_0,R_1,\ldots,R_K=Q$ in $\EMLRP$ such that, for each $1\le k\le K$, either
\begin{enumerate}[(a)]
\item one of $R_{k-1}$ and $R_k$ is obtained from the other by splitting a signal as in~\eqref{eq:split}; or
\item $R_{k-1},R_k\in\EMLRP_m$ for some $m$, and there is a continuous path $\phi:[0,1]\to\EMLRP_m$ from $R_{k-1}$ to $R_k$ such that, for all but finitely many $t\in[0,1]$, $\phi(t)\in\EMLRP_{m,+}$ and $\phi$ is differentiable at $t$ with
\[
\phi'(t)=\sum_{j,\ell}\Bigl(\alpha_{j\ell}(t)\,\phi(t)^{j\to j+1}_{\le\ell}+\beta_{j\ell}(t)\,\phi(t)^{j\to j-1}_{\ge\ell}\Bigr)
\quad\text{for some }\alpha_{j\ell}(t),\beta_{j\ell}(t)\ge0.
\]
\end{enumerate}
\end{proposition}

Proposition~\ref{prop:LM-generation} shows that every Lehmann comparison between positive experiments can be generated by a finite sequence of experiments through two operations: (a) splitting or merging signals, which preserves the MLRP and Lehmann order, and (b) moving along a path in the MLRP set whose velocity is a nonnegative combination of reverse signal replacements except at finitely many times. In this sense, Proposition~\ref{prop:LM-generation} is the Lehmann counterpart of the straight-path decomposition used for Theorem~\ref{thm:BM-general}. Here, a straight path is not feasible because the MLRP set is nonconvex, so the path construction must be more intricate. 

Once Proposition~\ref{prop:LM-generation} is established, the proof of Theorem~\ref{thm:LM-general} follows the same logic as that of Theorem~\ref{thm:BM-general}. Decreasing in reverse signal replacement ensures that $\mathcal I$ decreases along each path in Proposition~\ref{prop:LM-generation}(b),\footnote{The construction in Proposition~\ref{prop:LM-generation} may involve reverse signal-replacement directions at cutoffs where the relevant likelihood products are equal. Lemma~\ref{lem:binding-extension} shows that the first-order inequalities remain valid in this case.} while split invariance leaves its value unchanged under the operations in Proposition~\ref{prop:LM-generation}(a). Connecting any $P,Q\in\EMLRP_+$ with $P\succeq_L Q$ through these steps therefore yields $\mathcal I(P)\ge\mathcal I(Q)$. Continuity and split invariance extend the comparison to all experiments in $\EMLRP$. The following subsections present how the paths in Proposition~\ref{prop:LM-generation} are constructed.

\subsubsection{The Lehmann Order and PP Plots}
We first revisit a geometric characterization of MLRP experiments and the Lehmann order introduced by \cite{jewitt2007information}. For an experiment $P\in\EMLRP_m$ and each adjacent state pair, the probability-probability (PP) plot joins the cumulative points $F_i^j=(F_P(j|\theta_i),F_P(j|\theta_{i+1}))$ in signal order, where $F_P(j|\theta_i):= \sum_{k \le j} P_{ik} $ is the probability of a signal no greater than $s_j$ in state $\theta_i$. The horizontal and vertical coordinates record the probability of a signal no greater than $j$ in states $i$ and $i+1$, respectively. For a second experiment $Q$, write $G_i^j$ for the corresponding points. When the state pair is fixed, we suppress $i$ and write $F^j$ and $G^j$. Let $\Ss_i(P)$ denote the region above this curve within the unit square (see Figure~\ref{figure:PPplot}). Formally, 
\begin{equation}\label{def:Ssf}
\Ss_i(P):=\operatorname{conv}\{(0,1),F_i^0,\ldots,F_i^{m}\},
\qquad 1\le i\le n-1.
\end{equation}
When $P_{ij} > 0$, the segment corresponding to signal $j$ has slope $P_{i+1,j}/P_{ij}$, and the MLRP requires these slopes to be increasing. In other words, an experiment satisfies the MLRP if and only if every adjacent-state PP plot is convex. The following characterization translates the Lehmann order into PP comparisons.

\begin{lemma}[\cite{jewitt2007information}]\label{lem:Lehmann-geometry}
For any $P,Q\in\EMLRP$,
\[
P\succeq_L Q
\quad\Longleftrightarrow\quad
\Ss_i(Q)\subseteq\Ss_i(P)\quad\text{for every }1\le i\le n-1.
\]
\end{lemma}

Geometrically, a less informative experiment has a higher PP plot and hence a smaller upper region, as illustrated in Figure~\ref{figure:PPplot}. 

The two existing characterizations of the Lehmann order play complementary roles in proving Proposition~\ref{prop:LM-generation}. The monotone quasi-garbling representation \citep{kim2023a} describes how an experiment can be made less informative, but not how to reach a less informative one through local changes while remaining in the MLRP set; the PP-plot characterization \citep{jewitt2007information} identifies the target, but not a path to it. We connect these two characterizations through the \emph{chord replacement} construction developed next, showing how local operations generate the Lehmann order.

\subsubsection{Chord Replacement}

Fix two positive experiments $P\in\EMLRP_{m,+}$ and $Q\in\EMLRP_{m',+}$ with $P\succeq_L Q$. To prove Proposition~\ref{prop:LM-generation}, we move $P$ toward an experiment with the same PP plots as $Q$ through MLRP-preserving paths, splitting signals where needed. Geometrically, this amounts to continuously shrinking $\Ss_i(P)$ toward $\Ss_i(Q)$ for every $1\le i<n$ while preserving the convexity of \emph{all} PP plots. There are two main challenges in this construction. First, the path cannot simply be a straight line, since convex combinations of MLRP experiments can violate the MLRP. Second, the PP plots are interdependent: changing signal probabilities to move one plot may also change the others.

To address these challenges, we introduce chord replacement, which continuously shrinks the upper region of one selected PP plot, preserves the MLRP, and leaves every other adjacent-state PP plot unchanged. These properties allow us to shrink the upper regions toward their targets one at a time through successive chord replacements, without undoing earlier steps. Along each such motion, the velocity is a nonnegative combination of reverse signal-replacement directions, as required in Proposition~\ref{prop:LM-generation}(b).

\begin{figure}
\centering
\begin{subfigure}[b]{0.46\textwidth}
\begin{tikzpicture}[scale=5]
	
	\draw[->] (0,0) -- (1.05,0) node[below] {\small $i$};
	\draw[->] (0,0) -- (0,1.05) node[left]  {\small $i+1$};
	
	\fill[red!20,opacity=0.2] (0,1) -- (0,0) -- (0.3,.1)--(0.6,0.25) -- (0.85,0.6) -- (1,1) -- (1,1) -- cycle;
	\draw[thick, red] 		  (0,0) -- (0.3,0.1) -- (0.6,.25) -- (0.85,0.6) -- (1,1) ;
	\draw [->] (0.6,0.1)--(0.48,0.18) node [pos=0,right]  {\footnotesize slope: $\frac{P_{i+1,j+1}}{P_{i,j+1}}$};
	\draw [->] (0.82,0.35)--(0.73,0.43) node [pos=0,right]  {\footnotesize slope: $\frac{P_{i+1,j+2}}{P_{i,j+2}}$};
	
	\fill  (0.3, 0.1) circle[radius=0.015]  ;
	\node [anchor = west] at (.3,.07) {\footnotesize $ F^{j}  $} ;
	\fill  (0.6, 0.25) circle[radius=0.015] node [anchor = west] {\footnotesize $ F^{j+1}  $} ;
	\fill  (0.85, 0.6) circle[radius=0.015] node [anchor = west] {\footnotesize $ F^{j+2} $} ;
	
	\node [red, anchor = west] at (0.9, 0.75) {$\Ss_i (P)$} ;
	
	\fill[blue!20,opacity=0.2] (0,1) -- (0,0) -- (0.2,0.1) -- (0.45,.25) -- (0.82, 0.65) -- (1,1) -- (1,1) -- cycle;
	\draw[thick, blue] 		  (0,0) -- (0.2,0.1) -- (0.45,.25) -- (0.82,0.65) -- (1,1) ;				
	
	\node [blue] at (0.4, 0.6) {$\Ss_i (Q)$} ;
	
\end{tikzpicture}			
\caption{PP plot and 
	the Lehmann order
}
\label{figure:PPplot}
\end{subfigure}
\hspace*{25pt}
\begin{subfigure}[b]{0.46\textwidth}
\centering
\begin{tikzpicture}[scale=5]
	\tikzset{arrow line/.style={
			postaction={decorate},
			decoration={
				markings,
				mark=between positions 0.2 and 3 step 0.2 with {\arrow[scale=0.55]{Latex}}
			}
	}}
	
	\def \xa {.2};
	\def \ya {.1};
	\def \xb {.6};
	\def \yb {.25};
	\def \xc {.85};
	\def \yc {.75};
	\def \l {((\yb-\ya)/(\xb-\xa) +(\yc-\ya)/(\xc-\xa))/2};
	
	\draw[->] (0,0) -- (1.05,0) node[below] {\small $i$};
	\draw[->] (0,0) -- (0,1.05) node[left]  {\small $i+1$};
	
	\draw[thick] 		  (0,.05) -- (\xa,\ya) -- (\xb,\yb) -- (\xc,\yc) -- (.9,1);
	
	\node at (0.3, 0.7) {$\Ss_i (P)$} ;

	\fill  (\xa, \ya) circle[radius=0.01]  ;
	\node [anchor = west] at (\xa,.07) {\footnotesize $ F^{j} 
		$} ;
	\fill  (\xb, \yb) circle[radius=0.01] node [anchor = west] {\footnotesize $ F^{j+1} 
		$} ;
	\fill  (\xc, \yc) circle[radius=0.01] node [anchor = west] {\footnotesize $ F^{j+2}
		$} ;
	\draw[blue] (\xa,\ya) -- (\xc, \yc) ;
	\draw[arrow line, dotted] (\xb, \yb)--( {\xa + (\xc-\xa)*((\yb-\ya)/(\yc-\ya))} , \yb );
	\fill[blue]  ( {\xa + (\xc-\xa)*((\yb-\ya)/(\yc-\ya))} , \yb ) circle[radius=0.01]  
	;
	\draw[arrow line, dotted] (\xb , \yb )--( \xb , {\ya + (\yc-\ya)*( (\xb-\xa)/(\xc-\xa) )  } );
	\fill[blue]  (\xb , \yb )--( \xb , {\ya + (\yc-\ya)*( (\xb-\xa)/(\xc-\xa) )  } ) circle[radius=0.01] 
	;
	
	\draw [dotted] (\xa,\ya) --(.8, {\ya+ \l*(.8-\xa )} ) ;
	
	\draw [red] (\xa,\ya) -- ( {\xb} , {\ya+\l*(\xb-\xa)} ) -- (\xc,\yc);
	
	\fill[red]  ( {\xa + (\yb-\ya)/(\l)} , \yb ) circle[radius=0.012] node[above left=2pt,inner sep=1pt,font=\footnotesize] {$B^-$}; 
	\fill[red]  ( {\xb} , {\ya+\l*(\xb-\xa)} ) circle[radius=0.012] node[above right=2pt,inner sep=1pt,font=\footnotesize] {$B^+$};
	
\end{tikzpicture}			
\caption{Chord replacement}
	\label{figure:lem-shrink}
\end{subfigure}
\caption{Geometry of MLRP experiments and the Lehmann order}
\label{figure:Ssi}
\end{figure}

\paragraph{Two-segment case.} Fix an adjacent state pair $i,i+1$. Consider the two segments joining $F^j$, $F^{j+1}$, and $F^{j+2}$ in Figure~\ref{figure:lem-shrink}. The goal is to replace them by the blue chord joining their endpoints, thereby shrinking the upper region $\Ss_i(P)$ by removing the region between the original plot and that chord. First insert a zero column between the original signals $j+1$ and $j+2$. This duplicates the middle cumulative point $F^{j+1}$ without changing any PP plot; this is a splitting step as in Proposition~\ref{prop:LM-generation}(a). The two copies can then move in different directions. We coordinate their movements through two reverse signal replacements that leave every other adjacent-state PP plot unchanged.

The first replacement transfers a proportion $\epsilon_1$ of the probability of the original signal $j+1$ into the new signal in every state below $i$. The second transfers a proportion $\epsilon_2$ of the probability of the original signal $j+2$ into the new signal in every state above $i+1$. Suppressing the unchanged columns and omitting the zeros in the restricted vector notation, we have
\[
\begin{tikzpicture}[baseline=(transfer.center)]
\matrix (transfer) [matrix of math nodes,
left delimiter={[}, right delimiter={]},
nodes={inner sep=2pt}, column sep=10pt, row sep=5pt] {
P^{j+1}_{\leq i} & \mathbf0_i & P^{j+2}_{\leq i}\\
P^{j+1}_{\geq i+1} & \mathbf0_{n-i} & P^{j+2}_{\geq i+1}\\
};
\draw[-{Latex[length=1.5mm]},semithick]
(transfer-1-1.north) to[out=65,in=115,looseness=1.25]
node[above,font=\footnotesize] {$\epsilon_1$} (transfer-1-2.north);
\draw[-{Latex[length=1.5mm]},semithick]
(transfer-2-3.south) to[out=245,in=295,looseness=1.25]
node[below,font=\footnotesize] {$\epsilon_2$} (transfer-2-2.south);
\end{tikzpicture}
\quad\longrightarrow\quad
\begin{bmatrix}
(1-\epsilon_1)P^{j+1}_{\leq i} & \epsilon_1 P^{j+1}_{\leq i} & P^{j+2}_{\leq i}\\
P^{j+1}_{\geq i+1} & \epsilon_2 P^{j+2}_{\geq i+1} & (1-\epsilon_2)P^{j+2}_{\geq i+1}
\end{bmatrix}.
\]
In the selected PP plot, the two copies become the red points in Figure~\ref{figure:lem-shrink}, with coordinates
\[
\begin{aligned}
B^-&=\bigl(F_P(j+1|\theta_i)-\epsilon_1 P_{i,j+1},\ F_P(j+1|\theta_{i+1})\bigr),\\
B^+&=\bigl(F_P(j+1|\theta_i),\ F_P(j+1|\theta_{i+1})+\epsilon_2 P_{i+1,j+2}\bigr).
\end{aligned}
\]
Thus, the first copy moves horizontally to the left and the second vertically upward. Both exterior endpoints stay fixed.

For any other adjacent state pair, both states lie on the same side of the cutoff. If $i'<i$, the lower-tail transfer moves the first copy backward along the original incoming segment, while the second copy stays at the original vertex. If $i'>i$, the upper-tail transfer moves the second copy forward along the original outgoing segment, while the first stays fixed. In each case the transfer merely splits one original segment into proportional parts. Thus, the marked cumulative points can change while the entire PP curve, and hence $\Ss_{i'}(P)$, remains unchanged throughout the operation.

For the selected PP plot, we coordinate the increases in $\epsilon_1$ and $\epsilon_2$ so that $F^j$, $B^-$, and $B^+$ remain collinear. The first two red segments therefore rotate upward together around the fixed left endpoint toward the blue chord. The rotating segments become steeper while the remaining segment leading to $F^{j+2}$ becomes flatter, remaining at least as steep as the rotating segments. Their slopes remain ordered and between the two original slopes until all three segments coincide with the blue chord. The coordination is essential for preserving convexity. Increasing $\epsilon_1$ alone moves $B^-$ left while keeping $B^+$ fixed: the first segment becomes steeper than the second, violating the MLRP. Increasing $\epsilon_1$ and $\epsilon_2$ simultaneously at the prescribed rates instead keeps the two segments aligned and preserves convexity. 

This coordinated rotation also satisfies the velocity condition in Proposition~\ref{prop:LM-generation}(b): at each interior point, its velocity is a nonnegative linear combination of the two reverse signal-replacement directions at the current experiment. Lemma~\ref{lem:removal} extends this construction to the chord joining any two cumulative points. Longer chords across multiple signals introduce additional difficulties, as illustrated in Figure~\ref{figure:lem-removal}; see Appendix~\ref{sec:chord-replacement-proof} for the proof.

\subsubsection{From Chord Replacement to a Global Path}
\begin{figure}
\centering
\begin{subfigure}[b]{0.46\textwidth}
\begin{tikzpicture}[scale=5]

	\def \xa {.35};
	\def \ya {.1};
	\def \xb {.75};
	\def \yb {.3};
	\def \xc {.96};
	\def \yc {.85};

	\def \xk {.45};
	\def \yk {.25};
	\def \xl {.82};
	\def \yl {.65};

	\def \lk {\yk/\xk};
	\def \xki { ( \xk*( \xc * \yb -\xb*\yc ) )/( \xk *(\yb-\yc) + (-\xb +\xc)*\yk )} ;
	
	\draw[->] (0,0) -- (1.05,0) node[below] {\small $i$};
	\draw[->] (0,0) -- (0,1.05) node[left]  {\small $i+1$};
	
	\fill[red!20,opacity=0.2] (0,1) -- (0,0) -- (\xa,\ya)--(\xb,\yb) -- (\xc,\yc) -- (1,1) -- (1,1) -- cycle;
	\draw[thick, red] 		  (0,0) -- (\xa,\ya) -- (\xb,\yb) -- (\xc,\yc) -- (1,1) ;			
	
	\fill  (\xk, \yk) circle[radius=0.01]  ;
	\node [anchor = east] at (\xk, \yk ) {\footnotesize $ G^{1}  $} ;
	\fill  (\xl, \yl ) circle[radius=0.01] node [anchor = east] {\footnotesize $ G^{2}  $} ;
	
	\node[red,anchor=north,font=\small] at (.66,.22) {$p$};
	
	\fill[blue!20,opacity=0.2] (0,1) -- (0,0) --  (\xk,\yk) -- (\xl,\yl ) -- (1,1) -- (1,1) -- cycle;
	\draw[thick, blue] 		  (0,0) -- (\xk,\yk) -- (\xl,\yl) -- (1,1) ;				

	\draw [dotted,thick] (0,0) -- (.94, {\lk*.94}) ;
	\node[anchor=west,font=\small] at (.94,{\lk*.94}) {$h_{1}$}; 

	\fill [red] ( {\xki}, {\lk * \xki} ) circle[radius=0.014];
	
	\draw[blue,thin] (.43,.62) -- (.62,.434);
	\node[blue,anchor=south,font=\small] at (.40,.65) {$q$};
	
\end{tikzpicture}			
\caption{Step 1
}
\label{figure:Ssi-operation-chain-1}
\end{subfigure}
\hspace*{15pt}
\begin{subfigure}[b]{0.46\textwidth}
\centering
\begin{tikzpicture}[scale=5]
	\tikzset{arrow line/.style={
			postaction={decorate},
			decoration={
				markings,
				mark=between positions 0.2 and 3 step 0.2 with {\arrow[scale=0.55]{Latex}}
			}
	}}
	
	\def \xa {.35};
	\def \ya {.1};
	\def \xb {.75};
	\def \yb {.3};
	\def \xc {.96};
	\def \yc {.85};

	\def \xk {.45};
	\def \yk {.25};
	\def \xl {.82};
	\def \yl {.65};

	\def \lk {\yk/\xk};
	\def \xki { ( \xk*( \xc * \yb -\xb*\yc ) )/( \xk *(\yb-\yc) + (-\xb +\xc)*\yk )} ;

	\def \yki {\lk*\xki };

	\def \ll {(\yl-\yk)/(\xl-\xk)};

	\def \xli {.928368};
	\def \yli {.767155 };

	\draw[->] (0,0) -- (1.05,0) node[below] {\small $i$};
	\draw[->] (0,0) -- (0,1.05) node[left]  {\small $i+1$};
	
	\fill[red!20,opacity=0.2] (0,1) -- (0,0) -- ({\xki},{\yki}) -- (\xc,\yc) -- (1,1) -- (1,1) -- cycle;
	\draw[thick, red] 		  (0,0) -- ({\xki},{\yki})  -- (\xc,\yc) -- (1,1) ;			
	
	\fill  (\xk, \yk) circle[radius=0.01]  ;
	\node [anchor = east] at (\xk, \yk ) {\footnotesize $ G^{1}  $} ;
	\fill  (\xl, \yl ) circle[radius=0.01] node [anchor = east] {\footnotesize $ G^{2}  $} ;
	
	\node[red,anchor=north,font=\small] at (.70,.35) {$p$};
	
	\fill[blue!20,opacity=0.2] (0,1) -- (0,0) --  (\xk,\yk) -- (\xl,\yl ) -- (1,1) -- (1,1) -- cycle;
	\draw[thick, blue] 		  (0,0) -- (\xk,\yk) -- (\xl,\yl) -- (1,1) ;				

	\draw [dotted,thick] (\xk,\yk) -- (1, {\yl+\ll*(1-\xl)}) ;
	\node[anchor=west,font=\small] at (1,{\yl+\ll*(1-\xl)}) {$h_{2}$};

	\fill [red] ( {\xli}, {\yli} ) circle[radius=0.014];
	
	\draw[blue,thin] (.43,.62) -- (.62,.434);
	\node[blue,anchor=south,font=\small] at (.40,.65) {$q$};
	
\end{tikzpicture}			
\caption{Step 2}
	\label{figure:Ssi-operation-chain-2}
\end{subfigure}
\caption{Matching the target PP plot.}
\label{figure:Ssi-operation-chain}
\end{figure}

We next show how repeated chord replacements match each target PP plot, thereby shrinking the corresponding upper region to $\Ss_i(Q)$. The procedure is illustrated in Figure~\ref{figure:Ssi-operation-chain}. In each panel, $p$ is the current PP curve in red and $q$ is the target curve in blue, with $p\le q$. Let $h_1,\ldots,h_d$ be the straight lines extending the segments of $q$. For each $h_r$, the set where $p<h_r$ is either empty or a single interval with $p=h_r$ at both endpoints. In the latter case, split the corresponding signals into positive proportional parts if needed to make these endpoints cumulative points. Chord replacement then implements $p\mapsto\max\{p,h_r\}$, raising $p$ to the line on this interval and leaving the rest unchanged. The red dots mark the right endpoints.

Convexity gives $q=\max_{1\le r\le d}h_r$, so processing all $d$ supporting lines yields $p=q$ without overshooting the target or disturbing previously matched segments. Throughout, every other adjacent-state PP plot remains unchanged. Repeating the procedure for each $1\le i<n$ therefore yields an experiment $\bar Q$ with the same PP plots as $Q$. Finally, $\bar Q$ and $Q$ have a common refinement obtained from each by splitting signals (Lemma~\ref{lem:PP-reconstruction}), which completes the sequence in Proposition~\ref{prop:LM-generation}. The formal proof is in Appendix~\ref{sec:LM-general-proof}.

\section{Posterior-Space Characterization of Monotonicity}\label{sec:posterior}
Sections~\ref{sec:blackwell} and~\ref{sec:lehmann} characterize Blackwell and Lehmann monotonicity for information measures defined on the space of experiments. When a prior $\mu\in\Delta(\Theta)$ is fixed, many applications, notably rational inattention, information design, and preferences over information, describe an experiment by the distribution of posterior beliefs it induces and define information measures on the space of random posteriors. This section translates our first-order conditions into posterior space, yielding the corresponding characterization of monotonicity for these measures.

\subsection{Experiments and Random Posteriors}\label{sec:random-posteriors}

Fix a full-support prior $\mu\in\Delta(\Theta)$ throughout this section. Given an experiment $P\in\mathcal E_m$, let $\tau^{j}$ denote the unconditional probability of signal $s_j$ and, whenever $\tau^{j} > 0$, let $q^{j} \in \Delta(\Theta)$ denote the posterior conditional on signal $s_j$, i.e.,
\begin{align}\label{eq:bayes}
    \tau^{j}=\sum_{i=1}^n\mu_iP_{ij}, \quad \text{ and } \quad q^j_i=\frac{\mu_iP_{ij}}{\tau^j},\qquad i=1,\ldots,n.
\end{align}
The \emph{random posterior} induced by $P$ is the distribution
\begin{equation*}
X_P=\sum_{j:\tau^j>0}\tau^j\,\delta_{q^j}\in\Delta(\Delta(\Theta)).
\end{equation*}
Since $P$ has finitely many signals, $X_P$ has finite support and satisfies Bayes plausibility: $\sum_j\tau^jq^j=\mu$. Let
\begin{equation*}
\mathcal X=\Big\{X\in\Delta(\Delta(\Theta)):X\text{ has finite support and }\textstyle\int q\,X(dq)=\mu\Big\}
\end{equation*}
denote the set of finite-support Bayes-plausible random posteriors. Conversely, every $X=\sum_{j=1}^m\tau^j\delta_{q^j}\in\mathcal X$ is induced by the experiment $P\in\mathcal E_m$ with $P_{ij}=\tau^jq^j_i/\mu_i$. Hence the \emph{Bayes map} $P\mapsto X_P$ at prior $\mu$ maps $\mathcal E$ onto $\mathcal X$ \citep{denti2022experimentalorder}. Moreover, $P\in\mathcal E_{m,+}$ if and only if every signal has positive probability and every posterior $q^j$ has full support.

An information measure on posterior space is a function $\mathcal{W}:\mathcal X\to\R$, which assigns a value to each Bayes-plausible distribution of posterior beliefs. Through the Bayes map, $\mathcal{W}$ induces an information measure on experiments given by
\begin{equation}
\mathcal{I}(P)=\mathcal{W}(X_P). \label{eq:posterior-rep}
\end{equation}
The induced measure $\mathcal{I}$ is automatically permutation and split invariant, since neither operation changes the induced random posterior $X_P$.

We say that $\mathcal W:\mathcal X\to\R$ is \emph{Blackwell monotone} if $\mathcal W(X_P)\ge\mathcal W(X_Q)$ whenever $P\succeq_BQ$, and \emph{Lehmann monotone} if the same inequality holds whenever $P,Q\in\EMLRP$ and $P\succeq_LQ$. Notice that for Lehmann monotonicity, we restrict attention to random posteriors induced by experiments satisfying the MLRP. For signals $j<k$ with positive probability and states $i<i'$, \eqref{eq:bayes} gives
\begin{equation}
P_{ij}P_{i'k}-P_{i'j}P_{ik}=\frac{\tau^j\tau^k}{\mu_i\mu_{i'}}\left(q^j_iq^k_{i'}-q^j_{i'}q^k_i\right). \label{eq:cross-product}
\end{equation}
Thus, $P$ satisfies the MLRP if and only if its posteriors are MLR-ordered. Let $\mathcal X^{\text{MLR}}\subseteq\mathcal X$ denote the set of random posteriors whose supports are totally ordered by $\ge_{\mathrm{MLR}}$. Accordingly, the posterior-space conditions for Lehmann monotonicity will be evaluated at those in $\mathcal X^{\text{MLR}}$.

\subsection{Local Conditions in Posterior Space}\label{sec:posterior-foc}
To express our first-order conditions in posterior space, we write random posteriors in terms of their probability weights and posterior beliefs. For each $m\ge2$, define
\begin{equation*}
\mathcal W_m(\tau^1,\ldots,\tau^m;q^1,\ldots,q^m)
=\mathcal W\Big(\sum_{j=1}^m\tau^j\delta_{q^j}\Big)
\end{equation*}
on the domain
\begin{equation*}
\mathcal X_m
=\left\{(\tau^1,\ldots,\tau^m;q^1,\ldots,q^m)
\in\Delta_m\times\Delta(\Theta)^m:
\sum_{j=1}^m\tau^jq^j=\mu\right\}
\subseteq\R^{m+mn},
\end{equation*}
where $\Delta_m$ denotes the simplex of $m$ probability weights.

Let $\mathcal X_{m,+}$ be the subset of $\mathcal X_m$ in which all weights are positive and all posterior beliefs have full support, and let $\mathcal X_+=\bigcup_{m\ge2}\mathcal X_{m,+}$. As for $\mathcal I$, we impose regularity separately for each $m$: we call $\mathcal W$ continuous if every $\mathcal W_m$ is continuous on $\mathcal X_m$, and differentiable if every $\mathcal W_m$ has a gradient $\nabla\mathcal W_m(z)\in\R^{m+mn}$ at each $z=(\tau,q)\in\mathcal X_{m,+}$ satisfying
\begin{equation*}
\lim_{\substack{z'\to z\\z'\in\mathcal X_{m,+},\ z'\ne z}}
\frac{\left|\mathcal W_m(z')-\mathcal W_m(z)-\langle\nabla\mathcal W_m(z),z'-z\rangle\right|}
{\|z'-z\|}=0.
\end{equation*}
We denote the components of $\nabla\mathcal W_m(z)$ corresponding to $\tau^j$ and $q_i^j$ by $\partial\mathcal W_m/\partial\tau^j$ and $\partial\mathcal W_m/\partial q_i^j$, respectively.\footnote{As in Section~\ref{sec:measures}, the gradient need not be unique. Any two differentiable functions that agree with $\mathcal W_m$ on $\mathcal X_m$ give the same derivatives along feasible paths, so their ordinary partial derivatives may be used directly.} The following lemma shows that these definitions agree with the regularity requirements for the induced information measure $\mathcal I$.

\begin{lemma}\label{lem:posterior-regularity}
Fix a full-support prior $\mu$, let $\mathcal W:\mathcal X \rightarrow \R$, and define $\mathcal I(P)=\mathcal W(X_P)$. $\mathcal I$ is continuous and differentiable if and only if $\mathcal W$ is continuous and differentiable.
\end{lemma}

Let $\mathcal W$ be differentiable, and let $(\tau,q)\in\mathcal X_{m,+}$ be the posterior representation of $P\in\mathcal E_{m,+}$. For $j\ne k$, consider replacing signal $s_j$ with $s_k$ with probability $\epsilon$. This reduces the probability of $s_j$ by $\epsilon\tau^j$ and increases the probability of $s_k$ by the same amount. The posterior $q^j$ is unchanged, while the receiving posterior moves toward $q^j$, with derivative $\tau^j(q^j-q^k)/\tau^k$ at $\epsilon=0$. Since $\mathcal I_m(P)=\mathcal W_m(\tau,q)$, the chain rule gives
\begin{equation}\label{eq:marginal-value}
\begin{aligned}
\langle\nabla\mathcal I_m(P),P^{j\to k}\rangle
=\tau^j\Bigg[&-\frac{\partial\mathcal W_m}{\partial\tau^j}
+\frac{\partial\mathcal W_m}{\partial\tau^k} +\frac1{\tau^k}\sum_{i=1}^n
\frac{\partial\mathcal W_m}{\partial q_i^k}(q_i^j-q_i^k)\Bigg].
\end{aligned}
\end{equation}
The first two terms capture the change in signal probabilities; the last captures the change in the receiving signal's posterior. It is convenient to collect these effects, for a given signal, into a single affine function of a belief. For each signal $s_j$, define
\begin{equation}\label{eq:marginal-plane}
L_j(r)=\frac{\partial\mathcal W_m}{\partial\tau^j}+\frac{1}{\tau^j}\sum_{i=1}^n\frac{\partial\mathcal W_m}{\partial q_i^j}\left(r_i-q_i^j\right),\qquad r\in\Delta(\Theta).
\end{equation}
Since $L_j(q^j)=\partial\mathcal W_m/\partial\tau^j$, the expression in~\eqref{eq:marginal-value} equals $\tau^j[L_k(q^j)-L_j(q^j)]$. More generally, transferring joint probability $\epsilon$ with state composition $r$ from signal $s_j$ to signal $s_k$ has first-order effect $\epsilon[L_k(r)-L_j(r)]$: the probability of $s_j$ falls and its posterior moves away from $r$, while the probability of $s_k$ rises and its posterior moves toward $r$. Applying~\eqref{eq:marginal-value} to signal replacements and the same calculation to reverse signal replacements yields the following characterization in posterior space.

\begin{theorem}\label{thm:posterior-foc-direct}
Let $\mathcal W:\mathcal X\to\R$ be continuous and differentiable. Then:
\begin{enumerate}[(a)]
\item $\mathcal W$ is Blackwell monotone if and only if, for every $m\ge2$, every $(\tau,q)\in\mathcal X_{m,+}$, and every $j\ne k$,
\begin{equation}\label{ineq:posterior-blackwell-direct}
L_k(q^j)\le L_j(q^j).
\end{equation}
\item $\mathcal W$ is Lehmann monotone if and only if, for every $m\ge2$, every $(\tau,q)\in\mathcal X_{m,+}$ with $q^m\ge_{\mathrm{MLR}}\cdots\ge_{\mathrm{MLR}}q^1$, and every $1\le\ell\le n$,
\begin{align}
L_{j+1}(\underline q^j_\ell)&\le L_j(\underline q^j_\ell)\quad\text{whenever }1\le j\le m-1\text{ and }q_\ell^jq_{\ell+1}^{j+1}>q_{\ell+1}^jq_\ell^{j+1}, \label{ineq:posterior-lehmann-lower-direct}\\
L_{j-1}(\overline q^j_\ell)&\le L_j(\overline q^j_\ell)\quad\text{whenever }2\le j\le m\text{ and }q_{\ell-1}^{j-1}q_\ell^j>q_\ell^{j-1}q_{\ell-1}^j. \label{ineq:posterior-lehmann-upper-direct}
\end{align}
Here $\underline q_\ell^j=q^j(\cdot|\theta\le\theta_\ell)$ and $\overline q_\ell^j=q^j(\cdot|\theta\ge\theta_\ell)$ are the lower and upper truncations.\footnote{Let $(q_{n+1}^j,q_{n+1}^{j+1})=(0,1)$ and $(q_0^{j-1},q_0^j)=(1,0)$ for the respective cutoff conditions.}
\end{enumerate}
\end{theorem}

Note \eqref{ineq:posterior-blackwell-direct} requires that no other signal assigns a higher marginal value to the composition $q^j$ than $s_j$ itself. It is precisely the decreasing signal replacement condition of Theorem~\ref{thm:BM-general} in posterior coordinates. Conditions \eqref{ineq:posterior-lehmann-lower-direct} and \eqref{ineq:posterior-lehmann-upper-direct} similarly express the decreasing reverse signal replacement conditions of Theorem~\ref{thm:LM-general}. The lower-tail condition requires that $s_{j+1}$ assign no higher marginal value to $\underline q^j_\ell$ than $s_j$ itself. The upper-tail condition imposes the analogous requirement on $s_{j-1}$ at $\overline q^j_\ell$. Unlike a signal replacement, transferring a proper tail moves both posteriors: the destination moves toward the transferred composition, and the source moves away from it.

\section{Applications and Extensions}\label{sec:applications}

This section presents three applications and extensions of our characterizations. We first apply the first-order conditions to two widely used classes of separable information measures and relate the posterior-space conditions to local-utility characterizations (Section~\ref{sec:separable}). Building on the posterior-separable case, we then characterize the Lehmann order itself through posterior-mean-ratio measures (Section~\ref{sec:cx-lehmann}). Finally, we obtain a complementary characterization of Blackwell monotonicity for measures defined on a fixed signal space under quasi-convexity (Section~\ref{sec:quasi-convex}).

\subsection{Likelihood- and Posterior-Separable Measures}\label{sec:separable}

Two classes of information measures are widely used in applications: likelihood-separable measures, defined on experiments, and posterior-separable measures, defined on random posteriors. Both are additively separable across signals: each integrates a fixed generating function, over likelihood vectors or over posteriors. The local conditions of Sections~\ref{sec:blackwell}--\ref{sec:posterior} therefore reduce to conditions on that function. 

\subsubsection{Likelihood-separable measures}\label{sec:likelihood-separable}
An information measure $\mathcal I:\mathcal E\to\R_+$ is \emph{likelihood-separable}, a term introduced by \citet{denti2022experimentalorder} in the context of information costs, if there exists a function $\psi:[0,1]^n\to\R_+$ such that, for every $m$ and every $P\in\mathcal E_m$,
\begin{equation*}
\mathcal I(P)=\sum_{j=1}^m \psi(P^j)-\psi(\mathbf 1).
\end{equation*}

A likelihood-separable measure $\mathcal{I}$ is continuous and differentiable if $\psi$ is continuous on $[0,1]^n$ and differentiable on $(0,1)^n$. Because these measures are additively separable across signal likelihoods, decreasing in signal replacement simplifies to 
\begin{equation*}
\langle \nabla\mathcal I(P),P^{j\rightarrow k}\rangle
=\sum_{i=1}^n \left( -\frac{\partial\psi(P^j)}{\partial P_i}+\frac{\partial\psi(P^k)}{\partial P_i}\right)P_{ij} \leq 0.
\end{equation*}
The first-order inequality therefore depends only on the two likelihood vectors $P^j$ and $P^k$. For Blackwell monotonicity, the only feasibility requirement is $P^j+P^k\leq\mathbf 1$. For Lehmann monotonicity, the two vectors must be adjacent columns of an MLRP experiment, which requires more than the pairwise MLR order in \eqref{eq:MLR-scales}. The following lemma characterizes this additional feasibility requirement.

\begin{lemma} \label{lem:nMLR}
Let $x,y \in (0,1)^n$ satisfy $y \geq_{\mathrm{MLR}} x$.
The following statements are equivalent:
\begin{enumerate}[(a)]
\item They are neighbor monotone-likelihood-ratio ordered, written 
$y \geq_{\mathrm{nMLR}} x$: that is, there exist $\alpha,\beta>0$ and an experiment 
$[\cdots \ \alpha x \ \beta y \ \cdots]\in\EMLRP$ in which $\alpha x$ and $\beta y$ are adjacent.

\item There exists $h\in(0,1)^n$ such that, for all $1\leq i\leq n-1$,
\begin{equation}
\frac{h_{i+1}}{h_i}
\leq
\frac{x_{i+1}}{x_i}
\leq
\frac{y_{i+1}}{y_i}
\leq
\frac{1-h_{i+1}}{1-h_i}.
\label{eq:nMLR}
\end{equation}

\item The recursively defined sequence $(\overline h_i)_{i=1}^n$ satisfies 
$\overline h_i>0$ for all $1\leq i\leq n$, where
\[
\overline h_1:=1,
\qquad
\overline h_{i+1}
:=
\min\left\{
1-\frac{y_{i+1}}{y_i}(1-\overline h_i),
\frac{x_{i+1}}{x_i}\overline h_i
\right\}.
\]
\end{enumerate}
\end{lemma}

These observations lead to the following unified necessary and sufficient conditions for the Blackwell and Lehmann monotonicity of likelihood-separable measures.
\begin{proposition}\label{prop:likelihood-separable}
Suppose a likelihood-separable measure is continuous and differentiable. The following statements hold:
\begin{enumerate}[(a)]
\item A likelihood-separable measure is Blackwell monotone if and only if $\psi$ is sublinear.\footnote{That is, $\psi(\gamma \cdot h) = \gamma \cdot \psi(h)$ for all $h \in [0,1]^n$ and $\gamma \geq 0$ such that $\gamma \cdot h \in [0,1]^n$; and $\psi(h + h') \leq \psi(h) + \psi(h')$ for all $h, h' \in [0,1]^n$ such that $h + h' \in [0,1]^n$.}

\item A likelihood-separable measure with a sublinear generator $\psi$ is Lehmann monotone if and only if for every $1\le\ell\le n$ and every $x,y\in(0,1)^n$ such that $y\geq_{\mathrm{nMLR}}x$, 
\begin{equation}\label{ineq:sublinear-LM}
\sum_{i=1}^{\ell}\left(-\frac{\partial\psi(x)}{\partial P_i}+\frac{\partial\psi(y)}{\partial P_i}\right)x_i\leq0,
\qquad\text{and}\qquad
\sum_{i=\ell}^n\left(-\frac{\partial\psi(y)}{\partial P_i}+\frac{\partial\psi(x)}{\partial P_i}\right)y_i\leq0.
\end{equation}
\end{enumerate}  
\end{proposition}

Results similar to statement (a) appear in \citet{denti2022experimentalorder} and \citet{baker2023}; our characterization gives a short proof. Specifically, to show that sublinear $\psi$ implies decreasing in signal replacement, observe that
\begin{align*}
\mathcal I_m(P+\epsilon P^{j\rightarrow k})-\mathcal I_m(P) & = \psi(P^k+\epsilon P^j) - \psi(P^k) + \psi((1-\epsilon)P^j) - \psi(P^j) \\
& \leq \epsilon \psi(P^{j}) - \epsilon \psi(P^j) = 0,
\end{align*}
where the inequality follows from sublinearity of $\psi$. Theorem \ref{thm:BM-general} therefore establishes the result.\footnote{Permutation invariance and split invariance follow directly from the definition and the sublinearity of $\psi$.}

The new result is statement (b), which characterizes Lehmann monotonicity for likelihood-separable measures with sublinear $\psi$. The two inequalities in~\eqref{ineq:sublinear-LM} express the reverse-signal-replacement conditions in terms of $\psi$. By Lemma~\ref{lem:nMLR}, two positive likelihood vectors are nMLR-ordered if and only if they can appear, after positive rescaling, as adjacent columns of an MLRP experiment. Thus, Theorem~\ref{thm:LM-general} implies that Lehmann monotonicity is equivalent to~\eqref{ineq:sublinear-LM} holding for all nMLR-ordered pairs. A convenient sufficient condition is to verify these inequalities for all MLR-ordered pairs, without imposing the additional feasibility requirement in Lemma~\ref{lem:nMLR}.

For likelihood-separable measures, any sublinear function $\psi$ guarantees Blackwell monotonicity, covering many familiar examples such as norms and seminorms. With two states, the Blackwell and Lehmann orders coincide on MLRP experiments, so sublinearity also ensures Lehmann monotonicity. With three or more states, the next result shows that it does not, but that weighted $p$-norms with $p>1$ do.

\begin{proposition}\label{prop:likelihood-separable-p-norm}
The following statements hold: 
\begin{enumerate}[(a)]
\item There exist likelihood-separable measures with sublinear $\psi$ that are not Lehmann monotone.
				
\item If $\psi(\cdot)$ is a weighted $p$-norm with $p>1$, i.e., there exists $w_{i} > 0$ for all $i$ such that
\begin{equation*}
	\psi(h) = \left(\sum_{i=1}^{n}w_{i}h_{i}^{p}\right)^{1/p}, 
\end{equation*}
then the likelihood-separable measure is Lehmann monotone. 
\end{enumerate}
\end{proposition}

\subsubsection{Posterior-separable measures}\label{sec:posterior-separable}

An information measure on posterior space, $\mathcal W:\mathcal X\to\R$, is \emph{posterior-separable} if, for some continuous function $v:\Delta(\Theta)\to\R$, it takes the form
\begin{equation}
\mathcal W_v(X)=\int v(q)\,X(dq)-v(\mu),\qquad X\in\mathcal X.
\label{eq:ps-def}
\end{equation}
The corresponding measure on experiments is $\mathcal I_v(P)=\mathcal W_v(X_P)$. Such measures are widely used as information costs in rational inattention \citep{sims2003implications,caplin2015revealed}. 

Furthermore, $\mathcal W_v$ is differentiable if $v$ is differentiable on the relative interior of $\Delta(\Theta)$. In this case, the coordinate representation $\mathcal W_v(\tau,q)=\sum_j\tau^jv(q^j)-v(\mu)$ has partial derivatives
\begin{equation*}
\frac{\partial\mathcal W_v}{\partial\tau^j}=v(q^j),
\qquad
\frac{\partial\mathcal W_v}{\partial q_i^j}=\tau^j\frac{\partial v(q^j)}{\partial q_i}.
\end{equation*}
Substituting these expressions into~\eqref{eq:marginal-plane} gives $L_j(r)=v(q^j)+\nabla v(q^j)\cdot(r-q^j)$: the marginal value of a composition $r$ at signal $s_j$ is the tangent plane of $v$ at $q^j$. To state the resulting conditions, define
\begin{equation}
D_v(r\,\|\,q)=v(r)-v(q)-\nabla v(q)\cdot(r-q),
\label{eq:bregman}
\end{equation}
for $r\in\Delta(\Theta)$ and $q$ in its relative interior. This expression is the gap between $v(r)$ and the affine approximation to $v$ at $q$; when $v$ is convex, it is the Bregman divergence generated by $v$ and, in the terminology of \citet{frankel2019quantifying}, measures the amount of information in news that moves beliefs from $q$ to $r$.

For the Blackwell order, $L_k(q^j)-L_j(q^j)$ in condition~\eqref{ineq:posterior-blackwell-direct} becomes
\begin{equation*}
v(q^k)-v(q^j)+\nabla v(q^k)\cdot(q^j-q^k)
=-D_v(q^j\,\|\,q^k)\le0.
\end{equation*}
This yields the familiar convexity characterization.

\begin{corollary}\label{cor:ps-blackwell}
A posterior-separable measure $\mathcal W_v$ is Blackwell monotone if and only if $v$ is convex on $\Delta(\Theta)$. If $v$ is differentiable on the relative interior of $\Delta(\Theta)$, these conditions are also equivalent to $D_v(p\,\|\,r)\ge0$ for all full-support beliefs $p,r$.
\end{corollary}

Thus, writing $H=-v$, the Blackwell monotone posterior-separable measures are exactly the expected reductions in a concave measure of uncertainty, which \citet[footnote~14]{frankel2019quantifying} identify as the expected values of information in decision problems.

For the Lehmann order, the same substitution turns $L_{j+1}(\underline q^j_\ell)-L_j(\underline q^j_\ell)$ in~\eqref{ineq:posterior-lehmann-lower-direct} into
\begin{equation*}
\begin{aligned}
&v(q^{j+1})-v(q^j)
+\nabla v(q^{j+1})\cdot(\underline q^j_\ell-q^{j+1})
-\nabla v(q^j)\cdot(\underline q^j_\ell-q^j) =D_v(\underline q^j_\ell\,\|\,q^j)
-D_v(\underline q^j_\ell\,\|\,q^{j+1}).
\end{aligned}
\end{equation*}
The upper-tail condition~\eqref{ineq:posterior-lehmann-upper-direct} has the same form, with $\overline q^j_\ell$ in place of $\underline q^j_\ell$ and $j-1$ in place of $j+1$. These calculations and Theorem~\ref{thm:posterior-foc-direct}(b) yield the following characterization.

\begin{proposition}\label{prop:ps-lehmann}
A posterior-separable and differentiable measure $\mathcal W_v$ is Lehmann monotone for every full-support prior if and only if, for every pair of full-support beliefs $q'\ge_{\mathrm{MLR}}q$ and every $1\le\ell\le n$,
\begin{equation}\label{ineq:ps-lehmann}
D_v(\underline q_\ell\,\|\,q')\ge D_v(\underline q_\ell\,\|\,q),
\qquad
D_v(\overline q'_\ell\,\|\,q)\ge D_v(\overline q'_\ell\,\|\,q').
\end{equation}
\end{proposition}

Since $q'\ge_{\mathrm{MLR}}q\ge_{\mathrm{MLR}}\underline q_\ell$ and $\overline q'_\ell\ge_{\mathrm{MLR}}q'\ge_{\mathrm{MLR}}q$, condition~\eqref{ineq:ps-lehmann} requires that news leading to a truncated posterior carry more information when it starts from a belief farther away in the MLR order. At the extreme cutoffs, the truncation is the entire posterior: $\underline q_n=q$ and $\overline q'_1=q'$, so the two inequalities reduce to $D_v(q\,\|\,q')\ge0$ and $D_v(q'\,\|\,q)\ge0$, which convexity guarantees by Corollary~\ref{cor:ps-blackwell}. The intermediate cutoffs impose additional restrictions.

Convexity alone does not deliver them. For example, residual variance, $H(q)=\sum_iq_i(1-q_i)$, is a standard measure of uncertainty, generated by $v(q)=\|q\|^2-1$ with $D_v(r\,\|\,q)=\|r-q\|^2$, but it violates~\eqref{ineq:ps-lehmann}. Let $n=3$, $q=(\tfrac{9}{20},\tfrac12,\tfrac1{20})$, and $q'=(\tfrac1{20},\tfrac{17}{20},\tfrac1{10})$, so that $q'\ge_{\mathrm{MLR}}q$. At the cutoff $\ell=3$, $\overline q'_3=\delta_{\theta_3}$, and $D_v(\delta_{\theta_3}\,\|\,q)=\tfrac{271}{200}<\tfrac{307}{200}=D_v(\delta_{\theta_3}\,\|\,q')$. More generally, the two classes in this section are linked. Let $\psi$ be continuous and sublinear and fix a full-support prior $\mu$. By Bayes' rule, $P_{ij}=\tau^jq^j_i/\mu_i$, so positive homogeneity of $\psi$ gives
\begin{equation*}
\sum_{j=1}^m\psi(P^j)-\psi(\mathbf 1)
=\sum_{j=1}^m\tau^jv_\mu(q^j)-v_\mu(\mu),
\qquad
v_\mu(q)=\frac1c\,\psi\!\left(\frac{cq_1}{\mu_1},\ldots,\frac{cq_n}{\mu_n}\right),
\end{equation*}
for any $0<c\le\min_i\mu_i$. At each prior, a likelihood-separable measure is thus the posterior-separable measure generated by $v_\mu$, which is convex because $\psi$ is. The likelihood-separable measures in Proposition~\ref{prop:likelihood-separable-p-norm}(a) are therefore posterior-separable measures with convex generators that are not Lehmann monotone, and their generators violate~\eqref{ineq:ps-lehmann}.

A leading example in which the condition does hold is the \emph{entropy cost} of \citet{sims2003implications}, which measures the expected reduction in Shannon entropy. Writing $H(q)=-\sum_iq_i\log q_i$, with $0\log0=0$, it is the posterior-separable measure generated by $v=-H$:
\begin{equation*}
\mathcal W_v(X)=H(\mu)-\int H(q)\,X(dq)
=\int D_{\mathrm{KL}}(q\,\|\,\mu)\,X(dq).
\end{equation*}
The induced measure $\mathcal I_v(P)=\mathcal W_v(X_P)$ is the mutual information between the state and the signal. For $v=-H$, $D_v$ is the Kullback--Leibler divergence, which is monotone in the MLR order along every MLR-ordered triple, not only along truncations. If $q'\ge_{\mathrm{MLR}}q\ge_{\mathrm{MLR}}p$ with $q,q'$ of full support, then $\log(q_i/q'_i)$ is decreasing in $i$ and $p$ is first-order stochastically dominated by $q$, so
\begin{equation*}
D_{\mathrm{KL}}(p\,\|\,q')-D_{\mathrm{KL}}(p\,\|\,q)
=\sum_{i=1}^np_i\log\frac{q_i}{q'_i}
\ge\sum_{i=1}^nq_i\log\frac{q_i}{q'_i}
=D_{\mathrm{KL}}(q\,\|\,q')\ge0.
\end{equation*}
Taking $p=\underline q_\ell$ gives the first inequality in~\eqref{ineq:ps-lehmann}, and the second follows symmetrically. Proposition~\ref{prop:ps-lehmann} therefore yields the following.

\begin{corollary}\label{cor:mutual-info}
The entropy cost is Lehmann monotone for every full-support prior.
\end{corollary}

\begin{remark}[Local utilities and the connection to \citet{grant1998intrinsic}]\label{rem:local-utility}
Fix a full-support prior $\mu$. For a posterior-separable measure, the marginal value of a posterior is given by the same function $v$ at every random posterior $X$. For a general $\mathcal W$, the corresponding object is a \emph{local utility} $v(\cdot;X)$, defined by
\begin{equation*}
\lim_{\alpha\downarrow0}\frac{\mathcal W((1-\alpha)X+\alpha Y)-\mathcal W(X)}{\alpha}
=\int v(q;X)\,(Y-X)(dq)\qquad\text{for all }Y.
\end{equation*}
Suppose $\mathcal W$ is continuous and admits differentiable local utilities at all positive random posteriors, with posterior gradient jointly continuous in $(q,X)$. The coordinate gradient of $\mathcal W_m$ can then be chosen so that $L_j$ is the tangent plane of $v(\cdot;X)$ at $q^j$. Our first-order condition for Blackwell monotonicity thus becomes $D_{v(\cdot;X)}(q^j\,\|\,q^k)\ge0$ for every pair of posteriors of every such $X$, recovering the convexity criterion of \citet{grant1998intrinsic} at the fixed prior. For Lehmann monotonicity, our first-order conditions compare the tangent planes at adjacent MLR-ordered posteriors by evaluating both at the lower or upper truncation of the source posterior.

An information measure can be differentiable in our coordinates yet fail to admit a local utility at some random posterior. Conversely, it can admit local utilities at every random posterior even when its coordinate representation fails to be differentiable. Online Appendix~\ref{OA:posterior-local-utility} formalizes both monotonicity characterizations and illustrates the distinction between the two regularity requirements.
\end{remark}

\subsection{A Convex-Order Characterization of the Lehmann Order}\label{sec:cx-lehmann}

So far, we have taken an information measure as given and asked when it is monotone. We now reverse this perspective and ask whether a family of measures can characterize an information order itself. For the Blackwell order, the answer is classical. Corollary~\ref{cor:ps-blackwell} shows that every posterior-separable measure with a convex generator is Blackwell monotone; conversely, these measures jointly characterize the order.

Formally, fix a full-support prior $\mu$ and restrict attention to experiments $P\in\mathcal E_{m,+}$, with induced random posterior $X_P=\sum_j\tau^j\delta_{q^j}\in\mathcal X_+$. For finite-support probability distributions $G$ and $H$, write $G\cx H$ if $\int f\,dG\ge\int f\,dH$ for every continuous convex function $f$ on the relevant domain.

\begin{proposition}[\citealp{blackwell1953equivalent}]\label{prop:blackwell-cx}
For any full-support prior $\mu$, $P\succeq_BQ$ if and only if $X_P\cx X_Q$.
\end{proposition}

Lehmann dominance need not imply Blackwell dominance, so requiring every posterior-separable measure with a convex generator to rank experiments in the same direction is too strong. We therefore identify a subclass whose comparisons characterize the Lehmann order on $\EMLRP_+$.

For nonzero $a,b\in\R^n_+$ and a full-support posterior $q$, define the \emph{ratio of posterior means}
\begin{equation*}
\rho_{a,b}(q)=\frac{\langle a,q\rangle}{\langle b,q\rangle}
=\frac{\E_q[a(\theta)]}{\E_q[b(\theta)]}.
\end{equation*}
For nonzero $a,b\in\R^n_+$ with $a\ge_{\mathrm{MLR}}b$, define the \emph{ratio distribution}
\begin{equation*}
Z^P_{a,b}=\sum_j\tau^j_b\,\delta_{\rho_{a,b}(q^j)},
\qquad
\tau^j_b=\frac{\tau^j\langle b,q^j\rangle}{\langle b,\mu\rangle},
\end{equation*}
which reweights the probability of each signal by the posterior mean $\langle b,q^j\rangle$ of $b$ and normalizes by its prior mean $\langle b,\mu\rangle$. Bayes plausibility gives
\begin{equation*}
\sum_j\tau^j_b=1,
\qquad
\sum_j\tau^j_b\,\rho_{a,b}(q^j)=\rho_{a,b}(\mu).
\end{equation*}
For a finite convex function $f:(0,\infty)\to\R$, the \emph{posterior-mean-ratio (PMR) measure}\label{sec:pmr} is the Jensen gap of $f$ under the ratio distribution, scaled by $\langle b,\mu\rangle$:
\begin{equation}\label{eq:pmr-def}
\mathcal I_{a,b,f}(P)
=\langle b,\mu\rangle\Bigl[\sum_j\tau^j_b\,f\bigl(\rho_{a,b}(q^j)\bigr)-f\bigl(\rho_{a,b}(\mu)\bigr)\Bigr].\tag{PMR}
\end{equation}
Since the subtracted term is the same for every experiment, comparing PMR measures for all convex $f$ is equivalent to comparing expectations of all convex functions under the ratio distributions. The weights $\tau^j_b$ equal the original signal probabilities $\tau^j$ when $b=\mathbf 1$ and the conditional probabilities $\tau^jq_i^j/\mu_i=P_{ij}$ given $\theta_i$ when $b=e_i$. Hence, $Z^P_{e_k,e_i}$ is the distribution of the posterior odds $q_k/q_i$ conditional on $\theta_i$.

PMR measures are posterior separable. By the definition of $Z^P_{a,b}$,
\begin{equation*}
\mathcal I_{a,b,f}(P)=\sum_j\tau^jv_{a,b,f}(q^j)-v_{a,b,f}(\mu),
\end{equation*}
with generator $v_{a,b,f}(q)=\langle b,q\rangle\,f(\rho_{a,b}(q))$. Its convexity follows from that of $f$ by Jensen's inequality. Since the domain is $\mathcal X_+$, the generator does not have to be defined on the boundary of the simplex or admit a finite continuous extension there. The following theorem gives the corresponding characterization of the Lehmann order.

\begin{samepage}
\begin{theorem}\label{thm:lehmann-cx}
Fix any full-support prior $\mu$, and let $P,Q\in\EMLRP_+$. The following are equivalent:
\begin{enumerate}[(a)]
\item $P\succeq_LQ$;
\item for all nonzero $a,b\in\R^n_+$ with $a\ge_{\mathrm{MLR}}b$, $Z^P_{a,b}\cx Z^Q_{a,b}$, that is, $\mathcal I_{a,b,f}(P)\ge\mathcal I_{a,b,f}(Q)$ for every finite convex $f:(0,\infty)\to\R$;
\item for every $i<k$, the posterior odds conditional on $\theta_i$ are ordered in the convex order: $Z^P_{e_k,e_i}\cx Z^Q_{e_k,e_i}.$
\end{enumerate}
\end{theorem}
\end{samepage}

As in Proposition~\ref{prop:blackwell-cx}, a single full-support prior suffices. The two formulations in~(b) are equivalent by~\eqref{eq:pmr-def}. The implication~(a)$\Rightarrow$(b) is where our characterization enters. Since PMR measures are posterior separable, Theorem~\ref{thm:LM-general}, through Proposition~\ref{prop:ps-lehmann}, reduces their Lehmann monotonicity to first-order conditions at each experiment. For differentiable $f$, the MLR ordering of $a,b$ and convexity of $f$ yield the required Bregman inequalities, and approximation extends the comparison to every finite convex $f$. Likewise, convexity of the generator $v_{a,b,f}$ and Proposition~\ref{prop:blackwell-cx} imply that every PMR measure is Blackwell monotone.

Condition~(c) is the special case $a=e_k$, $b=e_i$ of~(b). It is the criterion of \citet{kim1995efficiency}, written in terms of posterior odds.\footnote{\citet{kim1995efficiency} states the condition for likelihood ratios: for every $i<k$, the distribution of $P_{kj}/P_{ij}$ under $\theta_i$ is a mean-preserving spread of the corresponding distribution under $Q$. Since $q_k^j/q_i^j=(\mu_k/\mu_i)\,P_{kj}/P_{ij}$ and both distributions assign probability $P_{ij}$ to signal $j$, $Z^P_{e_k,e_i}$ is the distribution of the likelihood ratio rescaled by the prior odds $\mu_k/\mu_i$. Rescaling by a positive constant preserves the convex order, so the two conditions coincide and neither depends on the prior.} For~(c)$\Rightarrow$(a), adjacent states suffice: adjusting the probability weights converts each comparison with $k=i+1$ into a convex-order comparison of the posteriors in the corresponding two-state restriction. Proposition~\ref{prop:blackwell-cx} then gives Blackwell dominance on every pair of adjacent states, and Lemma~\ref{lem:Lehmann-geometry} gives $P\succeq_LQ$, recovering the pairwise Blackwell characterization of \citet{jewitt2007information}.

\subsubsection{Examples}

The following two families are PMR measures and are therefore Lehmann monotone by Theorem~\ref{thm:lehmann-cx}.

\begin{example}[Dispersion of predictions]\label{ex:dispersion}
Let $b=\mathbf 1$, $a=u$, and $f(z)=z^2$, where $u$ is nonzero, nonnegative, and increasing, so $u\ge_{\mathrm{MLR}}\mathbf 1$. Then
\begin{equation*}
\mathcal I_{a,b,f}(P)
=\sum_j\tau^j\bigl(\langle u,q^j\rangle-\langle u,\mu\rangle\bigr)^2,
\end{equation*}
which is the variance of the posterior prediction of $u(\theta)$. Adding a constant to $u$ does not change this variance, so the result applies to every increasing state payoff. Thus, Lehmann dominance implies a weakly higher variance of posterior predictions of increasing state payoffs.
\end{example}

\begin{example}[Pairwise divergences]\label{ex:divergence}
Fix states $i<k$ and a finite convex function $f:(0,\infty)\to\R$ with $f(1)=0$. Set $a=e_k/\mu_k$ and $b=e_i/\mu_i$. Then $\langle b,\mu\rangle=\rho_{a,b}(\mu)=1$, $\rho_{a,b}(q^j)=P_{kj}/P_{ij}$, and $\tau^j_b=P_{ij}$, so
\begin{equation*}
\mathcal I_{a,b,f}(P)
=\sum_jP_{ij}\,f\!\left(\frac{P_{kj}}{P_{ij}}\right)
=D_f(P_k\,\|\,P_i),
\end{equation*}
the $f$-divergence between the signal distributions in states $\theta_k$ and $\theta_i$ \citep{polyanskiy2022information}.\footnote{\citet{bloedel2025modeling} propose information costs based on multivariate $f$-divergences, which compare the signal distributions in all states with a common reference distribution chosen to minimize the divergence, and derive their implications for optimal information acquisition and discrete choice. The divergences here instead compare the signal distributions of two fixed states.} Thus, every such pairwise divergence and every nonnegative combination of these divergences is Lehmann monotone. Familiar examples include Kullback--Leibler divergence, with $f(t)=t\log t$; Pearson $\chi^2$ divergence, with $f(t)=(t-1)^2$; squared Hellinger distance, up to scale, with $f(t)=(\sqrt t-1)^2$; and total variation distance, with $f(t)=\tfrac12|t-1|$.

The generators $f(t)=t\log t$ and $f(t)=-\log t$ yield $D_{\mathrm{KL}}(P_k\,\|\,P_i)$ and $D_{\mathrm{KL}}(P_i\,\|\,P_k)$, respectively. The result therefore also establishes Lehmann monotonicity of the log-likelihood ratio costs of \citet{pomatto2023a},
\begin{equation*}
\sum_{i\ne k}\beta_{ik}D_{\mathrm{KL}}(P_i\,\|\,P_k),
\qquad \beta_{ik}\ge0.
\end{equation*}
\end{example}

\subsubsection{Relation to Existing Characterizations}\label{sec:cx-literature}

Theorem~\ref{thm:lehmann-cx} also recovers two characterizations of the Lehmann order based on other statistics of the posterior: \citet{Chi2015} compares posterior predictions $\langle a,q\rangle$ ($b=\mathbf 1$, $a$ increasing), and \citet{asseyer2025information} compares posterior inverse hazard rates $\sum_{k>i}q_k/q_i$ ($a=\sum_{k>i}e_k$, $b=e_i$).

\begin{proposition}\label{prop:chi-asseyer}
Let $P,Q\in\EMLRP_+$. Each of the following is equivalent to $P\succeq_LQ$:
\begin{enumerate}[(a)]
\item \citep{Chi2015} For every full-support prior and every increasing $u\in\R^n$, the distribution of the posterior prediction $\langle u,q\rangle$ under $X_P$ dominates the corresponding distribution under $X_Q$ in the convex order.
\item \citep{asseyer2025information} For every full-support prior and every $i<n$, the distribution of the posterior inverse hazard rate $h_i(q)=\sum_{k>i}q_k/q_i$ conditional on $\theta_i$ under $P$ dominates the corresponding distribution under $Q$ in the convex order.
\end{enumerate}
\end{proposition}

Varying the prior in~(a) and~(b) can equivalently be viewed as enlarging the family of PMR measures evaluated at a fixed full-support benchmark prior: as the proof shows, each test at another prior becomes a test with rescaled vectors at the benchmark prior. Conditions~(a) and~(b) and condition~(c) of Theorem~\ref{thm:lehmann-cx} therefore describe different families of measures that characterize the same Lehmann order. Posterior odds conditional on $\theta_i$ depend on the prior only through the constant $\mu_k/\mu_i$, so condition~(c) requires a single prior, whereas predictions and inverse hazard rates depend on the prior, and~(a) and~(b) allow it to vary.

\subsection{Quasi-Convex Information Measures}\label{sec:quasi-convex}

Theorems \ref{thm:BM-general} and \ref{thm:LM-general} provide characterizations of Blackwell and Lehmann monotonicity for information measures defined over all finite experiments. Here, we show how quasi-convexity allows for a characterization of Blackwell monotonicity within a fixed signal space. 

In this subsection only, we restrict attention to $\mathcal{E}_{m}$ and let an information measure be defined as $\mathcal{I}_{m}: \mathcal{E}_{m} \rightarrow \R$. It is quasi-convex if for any $P, Q \in \mathcal{E}_{m}$, and for all $\lambda \in [0,1]$, 
\begin{align*}
	\mathcal{I}_{m}(\lambda P + (1-\lambda)Q) \leq \max\{\mathcal{I}_{m}(P), \mathcal{I}_{m}(Q)\}.
\end{align*}
The convex combination $\lambda P + (1-\lambda)Q$ represents an experiment equivalent to an ex-ante randomization between $P$ and $Q$ without revealing the actual experiment. Quasi-convexity therefore means that such ex-ante randomization cannot create an experiment assigned a strictly higher value than both $P$ and $Q$. This is weaker than convexity, a common assumption on information measures in applications. 

\begin{theorem}\label{thm:blackwell-monotone-characterization}
A continuous, differentiable, and quasi-convex information measure $\mathcal{I}_{m}: \mathcal{E}_{m} \rightarrow \R$ is Blackwell monotone if and only if it is permutation invariant and decreasing in signal replacement.
\end{theorem}

The proof of Theorem \ref{thm:blackwell-monotone-characterization} uses the convex structure of garbling matrices. The extreme points of the set $\mathcal M_m$ of $m\times m$ stochastic matrices have exactly one nonzero entry in each row (see, e.g., \citet{Cao2022}). For any such matrix that is not a permutation matrix, moving one entry equal to one from a column containing at least two ones to a zero column increases its rank by one. A single signal replacement reverses this move. Repeating the construction reaches a permutation matrix, so every extreme-point garbling of $P\in\mathcal E_m$ can be obtained from a permutation of $P$ through successive signal replacements. Permutation invariance and decreasing in signal replacement, together with continuity, imply that each such garbling has a value no greater than $\mathcal I_m(P)$. Finally, every garbling of $P$ is a convex combination of extreme-point garblings, so quasi-convexity extends the comparison to all garblings. 

Because the MLRP set is nonconvex, extending the quasi-convexity argument used for Blackwell monotonicity to Lehmann monotonicity is not straightforward. The proof of Theorem~\ref{thm:LM-general} addresses nonconvexity by constructing MLRP-preserving paths through experiments with additional signals. It is not immediate whether quasi-convexity can provide an alternative argument. We leave the possibility of a general fixed-signal-space characterization of Lehmann monotonicity under quasi-convexity for future research.

\section{Conclusion}\label{sec:conclusion}

This paper establishes first-order conditions for the monotonicity of information measures with respect to the Blackwell and Lehmann orders. Monotonicity is a global property, but our characterizations reduce it to a local check, much as a nonnegative derivative certifies that a function on the real line is increasing. For the Blackwell order, it suffices to verify that the measure decreases under small signal replacements; for the Lehmann order, under small reverse signal replacements, which shift signals upward only in low states and downward only in high states. In each case, the check involves a single experiment and a single elementary operation, rather than every pair of ranked experiments. For the Lehmann order, the key step is a local-generation characterization: every Lehmann comparison is generated by reverse signal replacements along paths that remain within the MLRP set.

The characterizations apply to continuous and differentiable information measures, whether defined on experiments or on distributions of posterior beliefs, and we illustrate their use in several directions. Applied to likelihood-separable and posterior-separable measures, the classes most commonly used in applications, they identify which members are Lehmann monotone; in particular, the entropy cost is. Applied to posterior-mean-ratio measures, a subclass of posterior-separable measures, they yield a characterization of the Lehmann order itself, a counterpart to the classical convex-order characterization of the Blackwell order.

We view these results as a step toward comparative statics with respect to information. In many economic problems, the object of interest is how more or less information affects an objective, such as a decision maker's expected payoff, a seller's profit, or welfare in a market or a game. Such an objective, viewed as a function of the experiment, is itself an information measure, and whether it is monotone in the Blackwell or Lehmann order is the question of whether more information always helps or always hurts. Establishing this by comparing every pair of ranked experiments is often intractable. Our results allow one instead to examine the local effect of an elementary loss of information, a signal replacement or a reverse signal replacement, and to conclude from its sign what holds globally.

Several directions remain open. Our characterization of Lehmann monotonicity is developed for finite signal spaces, and its proof relies on constructing paths through finitely many signals; extending it to continuous signal spaces is a natural next step. A related problem is optimization over MLRP experiments under a Lehmann monotone measure, which our analysis does not address and which is complicated by the nonconvexity of the MLRP set. Finally, our approach raises the question of whether other information orders, such as the linear-Blackwell order \citep{chen2026ranking}, can be generated by local operations that yield analogous monotonicity characterizations.

\appendix
\part*{Appendix}
\numberwithin{equation}{section}
\numberwithin{lemma}{section}
\numberwithin{proposition}{section}
\numberwithin{theorem}{section}
\numberwithin{definition}{section}
\numberwithin{figure}{section}

\section{Proofs for Section \ref{sec:blackwell}}

\begin{proof}[Proof of Theorem \ref{thm:BM-general}]
Necessity is established in the main text. For sufficiency, suppose that $\mathcal I$ is permutation invariant, split invariant, and decreasing in signal replacement. First take $P\in\mathcal E_{m,+}$ and $Q=PM\in\mathcal E_{m',+}$, and let $\hat P$, $\hat Q$, and $\phi(t)=(1-t)\hat P+t\hat Q$ be as in Section~\ref{sec:blackwell}. Since $\phi(t)\in\mathcal E_{m+m',+}$ for every $t\in(0,1)$, the argument there shows that $\frac{d}{dt}\mathcal I(\phi(t))\le0$ on $(0,1)$. For any $0<s<t<1$, the mean value theorem gives $\mathcal I(\phi(s))\ge\mathcal I(\phi(t))$. Letting $s\downarrow0$ and $t\uparrow1$, continuity yields $\mathcal I(P)=\mathcal I(\hat P)\ge\mathcal I(\hat Q)=\mathcal I(Q)$.

For arbitrary $P\in\mathcal E_m$ and $Q=PM\in\mathcal E_{m'}$, choose positive stochastic matrices $P_r\to P$ and $M_r\to M$ of the same respective dimensions. Then $Q_r=P_rM_r$ is positive and converges to $Q$. The positive case gives $\mathcal I(P_r)\ge\mathcal I(Q_r)$, and continuity yields $\mathcal I(P)\ge\mathcal I(Q)$.
\end{proof}

\section{Proofs for Section \ref{sec:lehmann}}
\subsection{Nonconvexity of the MLRP Set}\label{appendix:nonconvex-MLRP-set}

We show that $\EMLRP_3$ is not convex. Consider
\[  P = \begin{bmatrix} 
	0.04 & 0.36 & 0.60 \\ 
	0.02 & 0.18 & 0.80 \\
	0.02 & 0.18 & 0.80 
\end{bmatrix} \quad \text{ and } \quad 
Q = \begin{bmatrix} 
	0.60 & 0.04 & 0.36 \\ 
	0.40 & 0.06 & 0.54 \\
	0.40 & 0.06 & 0.54
\end{bmatrix}    .
\]
The likelihood ratios between $\theta_2$ and $\theta_1$ are $(1/2,1/2,4/3)$ under $P$ and $(2/3,3/2,3/2)$ under $Q$, both increasing in the signal. Since the last two rows of each matrix coincide, $P,Q\in\EMLRP_3$. Their average is
\[
R = \frac{1}{2}P+\frac{1}{2}Q = 
\begin{bmatrix} 
0.32 & 0.20 & 0.48 \\ 
0.21 & 0.12 & 0.67 \\ 
0.21 & 0.12 & 0.67 
\end{bmatrix}.
\]
This experiment violates the MLRP between $\theta_1,\theta_2$ and $s_1,s_2$, since
\[
\frac{R_{21}}{R_{11}}=\frac{0.21}{0.32}>\frac{0.12}{0.20}=\frac{R_{22}}{R_{12}}.
\]
Thus $R\notin\EMLRP_3$.

\subsection{Proof of Lemma~\ref{lem:Lehmann-marginal-general}}
\begin{proof}[Proof of Lemma~\ref{lem:Lehmann-marginal-general}]
Consider the lower-tail replacement $Q=P+\epsilon P^{j\to j+1}_{\le\ell}$. By Definition~\ref{def:reverse-signal-replacement}, $Q$ is a monotone quasi-garbling of $P$ for every $\epsilon\in[0,1]$. Thus Definition~\ref{def:orders}(b) gives $P\succeq_L Q$ once we verify that $Q$ satisfies the MLRP.

For $\epsilon<1$, $Q$ is positive, so it suffices to check that $Q_{i+1,k}/Q_{ik}$ is increasing in $k$ for every adjacent state pair. Multiplying these ratios gives the same property for every other state pair. If $i<\ell$, the ratio in column $j$ is unchanged, while that in column $j+1$ is a weighted average of the original ratios in columns $j$ and $j+1$:
\[
\frac{P_{i+1,j}}{P_{ij}}
\le \frac{Q_{i+1,j+1}}{Q_{i,j+1}}
=\frac{P_{i+1,j+1}+\epsilon P_{i+1,j}}
{P_{i,j+1}+\epsilon P_{ij}}
\le \frac{P_{i+1,j+1}}{P_{i,j+1}}.
\]
All other ratios are unchanged, so their ordering is preserved. If $i>\ell$, neither row changes. At $i=\ell<n$, the ratio in column $j$ increases and that in column $j+1$ decreases; hence only their mutual ordering can fail. That ordering requires
\[
(1-\epsilon)P_{\ell j}P_{\ell+1,j+1}
\ge (P_{\ell,j+1}+\epsilon P_{\ell j})P_{\ell+1,j},
\]
or equivalently,
\[
\epsilon\le\epsilon_*(P):=
\frac{P_{\ell j}P_{\ell+1,j+1}-P_{\ell,j+1}P_{\ell+1,j}}
{P_{\ell j}(P_{\ell+1,j}+P_{\ell+1,j+1})}\in(0,1).
\]
We may therefore take $\epsilon'=\epsilon_*(P)$, with a strict cutoff inequality whenever $\epsilon<\epsilon_*(P)$. If $\ell=n$, set $\epsilon_*(P)=\epsilon'=1$. The upper-tail case follows symmetrically with the corresponding bound given by 
\[
\epsilon''=\epsilon_U(P):=
\frac{P_{\ell-1,j-1}P_{\ell j}-P_{\ell,j-1}P_{\ell-1,j}}
{P_{\ell j}(P_{\ell-1,j-1}+P_{\ell-1,j})}
\quad\text{if }\ell>1,
\]
with $\epsilon_U(P)=1$ if $\ell=1$. 
\end{proof}

\subsection{Proof of Theorem~\ref{thm:LM-general}}\label{sec:LM-thm-proof}

\begin{proof}[Proof of Theorem~\ref{thm:LM-general}]
Necessity is established in the main text. For sufficiency, suppose that $\mathcal I$ is split invariant on $\EMLRP$ and decreasing in reverse signal replacement. The following lemma reduces the comparison to positive experiments; its proof is in Online Appendix~\ref{OA:Lehmann-regularity}.

\begin{lemma}\label{lem:positive-approximation}
Suppose $\mathcal I$ is continuous and split invariant on $\EMLRP$. If $\mathcal I(P)\ge \mathcal I(Q)$ for all $P,Q\in\EMLRP_+$ with $P\succeq_L Q$, then $\mathcal I(P)\ge \mathcal I(Q)$ for all $P,Q\in\EMLRP$ with $P\succeq_L Q$.
\end{lemma}

By Lemma~\ref{lem:positive-approximation}, it suffices to show $\mathcal I(P)\ge\mathcal I(Q)$ for $P,Q\in\EMLRP_+$ with $P\succeq_LQ$. Take the sequence $P=R_0,\ldots,R_K=Q$ in Proposition~\ref{prop:LM-generation}. In case~(a), $\mathcal I(R_{k-1})=\mathcal I(R_k)$ by split invariance. 

For case~(b), let $\phi$ be the corresponding path. The construction in the proof of Proposition~\ref{prop:LM-generation} may involve reverse signal-replacement directions at cutoffs where the relevant likelihood products are equal. The following lemma extends the first-order inequalities to these cutoffs; its proof is in Online Appendix~\ref{OA:Lehmann-regularity}.

\begin{lemma}\label{lem:binding-extension}
Suppose $\mathcal I$ is differentiable and decreasing in reverse signal replacement. On $\EMLRP_+$, inequalities \eqref{ineq:Lehmann-finite-1} and \eqref{ineq:Lehmann-finite-2} remain valid when the corresponding likelihood products are equal.
\end{lemma}

Partition $[0,1]$ at the finitely many exceptional times in Proposition~\ref{prop:LM-generation}(b). On the interior of each resulting interval, the chain rule gives
\[
\frac{d}{dt}\mathcal I(\phi(t))=\sum_{j,\ell}\Bigl(\alpha_{j\ell}(t)\bigl\langle\nabla\mathcal I(\phi(t)),\phi(t)^{j\to j+1}_{\le\ell}\bigr\rangle+\beta_{j\ell}(t)\bigl\langle\nabla\mathcal I(\phi(t)),\phi(t)^{j\to j-1}_{\ge\ell}\bigr\rangle\Bigr)\le0,
\]
where each inner product is nonpositive by the inequalities in Definition~\ref{def:decreasing_reverse_signal_general} and their extension in Lemma~\ref{lem:binding-extension}. Since $\mathcal I\circ\phi$ is continuous on each closed interval, the mean value theorem implies that it is nonincreasing there. Chaining the endpoint inequalities gives $\mathcal I(R_{k-1})\ge\mathcal I(R_k)$. Aggregating these comparisons over all steps gives $\mathcal I(P)\ge\mathcal I(Q)$.
\end{proof}

\subsection{Proof of Proposition~\ref{prop:LM-generation}}\label{sec:LM-general-proof}
Following the main-text sketch, we first prove Proposition~\ref{prop:LM-generation} by matching the target PP plots through repeated chord replacements, using Lemma~\ref{lem:removal}. We then prove the two-column and three-column cases of this lemma and explain the general induction, which is formalized in Online Appendix~\ref{OA:Lehmann-regularity}. Say that $P'$ is \emph{reachable} from $P$ if there exist $P=R_0,\ldots,R_K=P'$ such that each consecutive pair satisfies~(a) or~(b) of Proposition~\ref{prop:LM-generation}. Reachability is transitive, and Proposition~\ref{prop:LM-generation} states that $Q$ is reachable from $P$ whenever $P\succeq_LQ$. For a fixed adjacent state pair $i,i+1$, we write $F^j$ in place of $F_i^j$. 

\begin{lemma}[Chord replacement]\label{lem:removal}
For any $P\in\EMLRP_{m,+}$, $1\le i<n$, and integers $0\le j<j'\le m$, there is $\tilde P\in\EMLRP_+$, reachable from $P$, such that every adjacent PP plot other than the $i$-th is unchanged and the $i$-th PP plot is obtained by replacing its portion from $F^j$ to $F^{j'}$ by the chord joining these points. The remainder of the $i$-th plot is unchanged.
\end{lemma}

\begin{proof}[Proof of Proposition~\ref{prop:LM-generation}]
Fix $P,Q\in\EMLRP_+$ with $P\succeq_L Q$. Lemma~\ref{lem:Lehmann-geometry} gives $\Ss_i(Q)\subseteq\Ss_i(P)$ for every $1\le i<n$. Fix an adjacent state pair $i,i+1$, and denote the current and target PP plots by $p\le q$. Both are continuous convex piecewise-linear functions on $[0,1]$, with endpoints $(0,0)$ and $(1,1)$. Let $h_1,\ldots, h_d$ be the affine extensions of the segments of $q$. By convexity,
\[
q=\max_{1\le r\le d}h_r.
\]

Process these lines successively, using $p$ for the current plot. Since $p-h_r$ is convex and nonnegative at $0$ and $1$, the set where $p<h_r$ is either empty or an interval with $p=h_r$ at both endpoints. In the latter case, split the corresponding signals into positive proportional parts, if necessary, to make these endpoints cumulative points, and apply Lemma~\ref{lem:removal}. Since splittings are steps of type~(a), this implements
\[
p\longmapsto\max\{p,h_r\}
\]
through steps of types~(a) and~(b), while preserving positivity and every other adjacent-state PP plot. Each update preserves $p\le q$. After all lines have been processed, $p \ge\max_r h_r=q$, so $p=q$.

Repeating this finite procedure for every adjacent state pair, while preserving plots already matched, yields a positive experiment $\bar Q$ that is reachable from $P$ and has the same adjacent-state PP plots as $Q$. The following lemma shows that such experiments can be connected by signal splittings and mergings; its proof is in Online Appendix~\ref{OA:Lehmann-regularity}.

\begin{lemma}\label{lem:PP-reconstruction}
Any two MLRP experiments with the same adjacent-state PP plots can be made identical by splitting signals in each experiment.
\end{lemma}

By Lemma~\ref{lem:PP-reconstruction}, $\bar Q$ and $Q$ admit a common refinement $R$ obtained from each by adjacent column splittings. Thus steps of type~(a) take $\bar Q$ to $R$ by splitting and $R$ to $Q$ by merging. Since $\bar Q$ is reachable from $P$, transitivity implies that $Q$ is reachable from $P$.
\end{proof}

\subsubsection{Proof of the chord-replacement lemma}\label{sec:chord-replacement-proof}

\begin{proof}[Proof of Lemma~\ref{lem:removal}]

Fix the original experiment $P$, the state pair $i,i+1$, and the block of columns $j+1,\ldots,j'$. Set
\[
X_r=\sum_{s=1}^r P_{i,j+s},\qquad
Y_r=\sum_{s=1}^r P_{i+1,j+s},\qquad
\lambda_r=\frac{Y_r}{X_r}\quad(1\le r\le j'-j),
\]
with $X_0=Y_0=0$. Measuring coordinates relative to $F^j$, the original vertices are $(X_r,Y_r)$ and $\lambda_r$ is the slope of the chord from $(0,0)$ to $(X_r,Y_r)$. The fixed state pair and starting index are suppressed from $X_r,Y_r,\lambda_r$; these quantities and the original columns $P^{j+r}$ remain fixed throughout the construction. The original segment slopes
\[
d_r=\frac{Y_r-Y_{r-1}}{X_r-X_{r-1}}
=\frac{P_{i+1,j+r}}{P_{i,j+r}},\qquad r=1,\ldots,j'-j,
\]
are positive and increasing by the MLRP. Each $\lambda_r$ is their weighted average over the first $r$ columns, so $0<\lambda_1\le\cdots\le\lambda_{j'-j}$.

For $k=1,\ldots,j'-j$, we construct an experiment $P^{(k)}\in\EMLRP_+$ that replaces the original portion from $(0,0)$ to $(X_k,Y_k)$ by its chord, leaves the rest of that plot and every other adjacent PP plot unchanged, and is reachable from $P$. For $k=1$, take $P^{(1)}=P$.

\paragraph{The two-column operation ($k=2$).}
Assume $j'-j\ge2$. We formalize the rotation illustrated in Figure~\ref{figure:lem-shrink}. If $\lambda_1=\lambda_2$, take $P^{(2)}=P$. Otherwise, for $\lambda\in[\lambda_1,\lambda_2]$, let $\phi=\phi(\lambda)\in\mathcal E_{m+1}$ be obtained by inserting a zero column between $P^{j+1}$ and $P^{j+2}$ and applying the two coordinated replacements in the main text, with $\epsilon_1=1-Y_1/(\lambda X_1)$ and $\epsilon_2=(\lambda X_1-Y_1)/(Y_2-Y_1)$. Suppressing the unchanged columns and omitting the zeros in the restricted vector notation, the resulting three columns are as follows. We use this convention in all block matrices below, and write $\phi$ for the current experiment $\phi(\lambda)$. 

\begin{center}
\begin{minipage}{\linewidth}
\centering
\begin{tikzpicture}
\matrix (twocolumn) [matrix of math nodes,
left delimiter={[}, right delimiter={]},
nodes={font=\small,inner xsep=1pt,inner ysep=3pt},
column sep=18pt,row sep=9pt] {
\dfrac{Y_1}{\lambda X_1}P^{j+1}_{\le i} & \dfrac{\lambda X_1-Y_1}{\lambda X_1}P^{j+1}_{\le i} & P^{j+2}_{\le i}\\
P^{j+1}_{\ge i+1} & \dfrac{\lambda X_1-Y_1}{Y_2-Y_1}P^{j+2}_{\ge i+1} & \dfrac{Y_2-\lambda X_1}{Y_2-Y_1}P^{j+2}_{\ge i+1}\\
};
\node[font=\footnotesize,anchor=south] at ($(twocolumn.north -| twocolumn-1-1.center)+(0,19pt)$) {$j+1$};
\node[font=\footnotesize,anchor=south] at ($(twocolumn.north -| twocolumn-1-2.center)+(0,19pt)$) {$j+2$};
\node[font=\footnotesize,anchor=south] at ($(twocolumn.north -| twocolumn-1-3.center)+(0,19pt)$) {$j+3$};
\draw[-{Latex[length=1.6mm]},semithick] ($(twocolumn-1-1.north)+(2pt,1pt)$) to[out=60,in=120,looseness=1.1] ($(twocolumn-1-2.north)+(-2pt,1pt)$);
\draw[-{Latex[length=1.6mm]},semithick] ($(twocolumn-2-3.south)+(-2pt,-1pt)$) to[out=240,in=300,looseness=1.1] ($(twocolumn-2-2.south)+(2pt,-1pt)$);
\end{tikzpicture}
\end{minipage}
\end{center}

The two interior cumulative points $B^-=(Y_1/\lambda,Y_1)$ and $B^+=(X_1,\lambda X_1)$, expressed relative to $F^j$, move left and up, respectively, as $\lambda$ increases. These movements correspond to the lower-state transfer $j+1\to j+2$ and the upper-state transfer $j+3\to j+2$ shown by the arrows. The first two slopes equal $\lambda$, while the last is $(Y_2-\lambda X_1)/(X_2-X_1)\ge\lambda$. The changing slopes stay between the two original slopes, so the selected PP plot remains convex, including at its joins to the exterior portions. Within each state group, the operation only subdivides original segments into proportional parts, so every other adjacent PP plot is unchanged and $\phi(\lambda)\in\EMLRP_{m+1,+}$ for $\lambda>\lambda_1$.

On $(\lambda_1,\lambda_2)$, differentiating the displayed coefficients gives
\[
\phi'(\lambda)=\frac1\lambda\,\phi^{j+1\to j+2}_{\le i}+\frac{X_1}{Y_2-\lambda X_1}\,\phi^{j+3\to j+2}_{\ge i+1},
\]
a combination of two reverse signal-replacement directions at the current experiment with positive coefficients, so $\phi$ satisfies condition~(b) of Proposition~\ref{prop:LM-generation} on $[\lambda_1,\lambda_2]$. At $\lambda_1$, the middle column is a zero column, so $\phi(\lambda_1)$ is obtained from $P$ by a splitting. Hence $\phi(\lambda_2)$ is reachable from $P$. At $\lambda_2$, all three segments lie on the chord from $(0,0)$ to $(X_2,Y_2)$. Set $P^{(2)}=\phi(\lambda_2)$.

\paragraph{The three-column operation ($k=3$).} Assume $j'-j\ge3$. We start from the experiment $P^{(2)}$ just constructed, whose first two original segments have been replaced by the chord of slope $\lambda_2$ from $(0,0)$ to $(X_2,Y_2)$. The next original segment, ending at $(X_3,Y_3)$, is still unchanged. To obtain $P^{(3)}$, we continue rotating the aligned segments on the chord until they align with this next segment, as illustrated in Figure~\ref{figure:lem-removal}. If $\lambda_2=\lambda_3$, they are already aligned, so take $P^{(3)}=P^{(2)}$. Otherwise, increase $\lambda$ from $\lambda_2$ to $\lambda_3$. The selected PP plot then follows the line of slope $\lambda$ from $(0,0)$ to $(X_2,\lambda X_2)$, followed by the segment joining $(X_2,\lambda X_2)$ to the fixed endpoint $(X_3,Y_3)$. The exterior portions remain unchanged, and the final segment has slope satisfying
\[
\lambda\le\frac{Y_3-\lambda X_2}{X_3-X_2}\le d_3.
\]
Together with $\lambda\ge\lambda_2\ge d_1$, these inequalities ensure convexity, including at the joins to the exterior portions. This portion of the PP plot increases pointwise with $\lambda$, from the plot at stage $2$ to the chord through $(X_3,Y_3)$.

The coordinate guides in Figure~\ref{figure:lem-removal} measure displacements from $F^j$. As illustrated there, we use the same moving pair $(Y_1/\lambda,Y_1)$ and $(X_1,\lambda X_1)$ as in the two-column case and duplicate the former endpoint $(X_2,Y_2)$ to add the pair $(Y_2/\lambda,Y_2)$ and $(X_2,\lambda X_2)$. At $\lambda_2$, the new pair coincides; it then separates leftward and upward in the same way as the first pair. Thus the points $(Y_s/\lambda,Y_s)$ move left and $(X_r,\lambda X_r)$ move up, for $r,s=1,2$, subdividing this portion of the PP plot into five segments. For $\lambda$ just above $\lambda_2$, the horizontally and vertically moving points alternate along the plot, with
\[
Y_1/\lambda<X_1<Y_2/\lambda<X_2.
\]

\begin{figure}[htbp]
\centering
\begin{subfigure}[b]{0.46\textwidth}
\centering
\begin{tikzpicture}[scale=5]

\tikzset{arrow line/.style={
postaction={decorate},
decoration={
markings,
mark=between positions 0.4 and 3 step 0.25 with {\arrow[scale=0.55]{Latex}}
}
}}

\tikzset{arrow linem/.style={
postaction={decorate},
decoration={
markings,
mark=between positions 0.8 and .9 step 0.25 with {\arrow[scale=0.55]{Latex}}
}
}}

\def \xa {0};
\def \ya {0};
\def \xb {.45};
\def \yb {.15};
\def \xc {.70};
\def \yc {.50};
\def \xd {.80};
\def \yd {.80};

\draw[->] (0,0) -- (.95,0) node[below] {\small $i$};
\draw[->] (0,0) -- (0,1.05) node[left]  {\small $i+1$};

\draw[black] (\xa,\ya) -- (\xb,\yb) -- (\xc,\yc) -- (\xd,\yd) -- (.83,.90);

\node at (0.18, 0.70) {$\Ss_i (P)$};
\fill  (\xa,\ya) circle[radius=0.01]  ;
\node[below left=2pt,inner sep=1pt] at (0,0) {\footnotesize $F^j$};
\fill  (\xb,\yb) circle[radius=0.01] node [anchor = west] {\footnotesize $ F^{j+1}  $} ;
\fill  (\xc,\yc) circle[radius=0.01] node [anchor = west] {\footnotesize $ F^{j+2} $} ;
\fill  (\xd,\yd) circle[radius=0.01] node [anchor = west] {\footnotesize $ F^{j+3} $} ;
\draw[densely dotted,gray] (\xb,0) -- (\xb,\yb);
\draw[densely dotted,gray] (\xc,0) -- (\xc,\yc);
\draw[densely dotted,gray] (\xd,0) -- (\xd,\yd);
\draw[densely dotted,gray] (0,\yb) -- (\xb,\yb);
\draw[densely dotted,gray] (0,\yc) -- (\xc,\yc);
\draw[densely dotted,gray] (0,\yd) -- (\xd,\yd);
\node[below=2pt,text=black,inner sep=1pt] at (\xb,0) {\footnotesize $X_1$};
\node[below=2pt,text=black,inner sep=1pt] at (\xc,0) {\footnotesize $X_2$};
\node[below=2pt,text=black,inner sep=1pt] at (\xd,0) {\footnotesize $X_3$};
\node[left=2pt,text=black,inner sep=1pt] at (0,\yb) {\footnotesize $Y_1$};
\node[left=2pt,text=black,inner sep=1pt] at (0,\yc) {\footnotesize $Y_2$};
\node[left=2pt,text=black,inner sep=1pt] at (0,\yd) {\footnotesize $Y_3$};
\draw[blue] (\xa,\ya) -- (\xd, \yd) ;
\draw[arrow line,thin ] (\xb,\yb)--( {\xa + (\xc-\xa)*((\yb-\ya)/(\yc-\ya))} , \yb );
\draw[arrow line, thin] (\xb,\yb)--( \xb , {\ya+ (\yc-\ya)*( (\xb-\xa) /(\xc-\xa))} );

\draw[dotted] (\xa,\ya)--(\xc,\yc);

\draw[arrow linem, thin,red] ( {\xa + (\xc-\xa)*((\yb-\ya)/(\yc-\ya))} , \yb )--( {\xa + (\xd-\xa)*((\yb-\ya)/(\yd-\ya))} , \yb );
\draw[arrow line, thin,red] ( \xb , {\ya+ (\yc-\ya)*( (\xb-\xa) /(\xc-\xa))} )--( \xb , {\ya+ (\yd-\ya)*( (\xb-\xa) /(\xd-\xa))} );

\draw[arrow line, thin,red] ( \xc , \yc )--( {\xa + (\xd-\xa)*((\yc-\ya)/(\yd-\ya))} , \yc );
\draw[arrow line, thin,red]  ( \xc , \yc )--( \xc , {\ya+ (\yd-\ya)*( (\xc-\xa) /(\xd-\xa))} );

\fill  ( {\xa + (\xc-\xa)*((\yb-\ya)/(\yc-\ya))} , \yb ) circle[radius=0.008]  ;
\fill  ( \xb , {\ya+ (\yc-\ya)*( (\xb-\xa) /(\xc-\xa))} ) circle[radius=0.008]  ;

\fill [red] ( {\xa + (\xd-\xa)*((\yb-\ya)/(\yd-\ya))} , \yb ) circle[radius=0.008]  ;
\fill [red] ( \xb , {\ya+ (\yd-\ya)*( (\xb-\xa) /(\xd-\xa))} ) circle[radius=0.008]  ;
\fill [red] ( {\xa + (\xd-\xa)*((\yc-\ya)/(\yd-\ya))} , \yc ) circle[radius=0.008]  ;
\fill [red] ( \xc , {\ya+ (\yd-\ya)*( (\xc-\xa) /(\xd-\xa))} ) circle[radius=0.008]  ;

\draw[red,fill=white] (\xb,\yc) circle[radius=0.010];
\node[above left,text=black,inner sep=2pt] at (\xb,\yc)
{\footnotesize $C_{12}$};
\end{tikzpicture}
\caption{{$\dfrac{Y_2}{X_1}\ge \lambda_3$}}
\label{figure:lem-removal-a}
\end{subfigure}	
\hspace*{15pt}
\begin{subfigure}[b]{0.46\textwidth}
\centering
\begin{tikzpicture}[scale=5]

\tikzset{arrow line/.style={
postaction={decorate},
decoration={
markings,
mark=between positions 0.4 and 3 step 0.25 with {\arrow[scale=0.55]{Latex}}
}
}}

\tikzset{arrow linem/.style={
postaction={decorate},
decoration={
markings,
mark=between positions 0.5 and .8 step 0.25 with {\arrow[scale=0.55]{Latex}}
}
}}

\def \xa {0};
\def \ya {0};
\def \xb {.50};
\def \yb {.25};
\def \xc {.60};
\def \yc {.40};
\def \xd {.70};
\def \yd {.85};

\draw[->] (0,0) -- (.95,0) node[below] {\small $i$};
\draw[->] (0,0) -- (0,1.05) node[left]  {\small $i+1$};

\draw[black] (\xa,\ya) -- (\xb,\yb) -- (\xc,\yc) -- (\xd,\yd) -- (.71,.95);

\node at (0.18, 0.70) {$\Ss_i (P)$};
\fill  (\xa,\ya) circle[radius=0.01]  ;
\node[below left=2pt,inner sep=1pt] at (0,0) {\footnotesize $F^j$};
\fill  (\xb, \yb) circle[radius=0.01] node [anchor = west] {\footnotesize $ F^{j+1}  $} ;
\fill  (\xc, \yc) circle[radius=0.01] node [anchor = west] {\footnotesize $ F^{j+2} $} ;
\fill  (\xd, \yd) circle[radius=0.01] node [anchor = west] {\footnotesize $ F^{j+3} $} ;
\draw[densely dotted,gray] (\xb,0) -- (\xb,\yb);
\draw[densely dotted,gray] (\xc,0) -- (\xc,\yc);
\draw[densely dotted,gray] (\xd,0) -- (\xd,\yd);
\draw[densely dotted,gray] (0,\yb) -- (\xb,\yb);
\draw[densely dotted,gray] (0,\yc) -- (\xc,\yc);
\draw[densely dotted,gray] (0,\yd) -- (\xd,\yd);
\node[below=2pt,text=black,inner sep=1pt] at (\xb,0) {\footnotesize $X_1$};
\node[below=2pt,text=black,inner sep=1pt] at (\xc,0) {\footnotesize $X_2$};
\node[below=2pt,text=black,inner sep=1pt] at (\xd,0) {\footnotesize $X_3$};
\node[left=2pt,text=black,inner sep=1pt] at (0,\yb) {\footnotesize $Y_1$};
\node[left=2pt,text=black,inner sep=1pt] at (0,\yc) {\footnotesize $Y_2$};
\node[left=2pt,text=black,inner sep=1pt] at (0,\yd) {\footnotesize $Y_3$};

\draw[arrow line,thin ] (\xb,\yb)--( {\xa + (\xc-\xa)*((\yb-\ya)/(\yc-\ya))} , \yb );
\draw[arrow line, thin] (\xb,\yb)--( \xb , {\ya+ (\yc-\ya)*( (\xb-\xa) /(\xc-\xa))} );
\draw[dotted] (\xa,\ya) -- (\xc, \yc);

\draw[dotted] (\xa,\ya) -- (\xd, { \ya + (\yc-\ya)*( (\xd-\xa)/(\xb-\xa) ) } );

\draw[arrow linem, thin,red] ( {\xa + (\xc-\xa)*((\yb-\ya)/(\yc-\ya))} , \yb )--( {\xa + (\xb-\xa)*((\yb-\ya)/(\yc-\ya))} , \yb );
\draw[arrow linem, thin,red] ( \xb , {\ya+ (\yc-\ya)*( (\xb-\xa) /(\xc-\xa))} )--( \xb , {\ya+ (\yc-\ya)*( (\xb-\xa) /(\xb-\xa))} );

\draw[arrow linem, thin,red] ( \xc , \yc )--( {\xa + (\xb-\xa)*((\yc-\ya)/(\yc-\ya))} , \yc );
\draw[arrow linem, thin,red]  ( \xc , \yc )--( \xc , {\ya+ (\yc-\ya)*( (\xc-\xa) /(\xb-\xa))} );

\draw[arrow line, thin,blue] ( {\xa + (\xb-\xa)*((\yb-\ya)/(\yc-\ya))} , \yb )  -- ( {\xa + (\xd-\xa)*((\yb-\ya)/(\yd-\ya))} , \yb ) ;
\draw[arrow line, thin,blue] ( \xb , {\ya+ (\yc-\ya)*( (\xb-\xa) /(\xb-\xa))} ) -- ( \xb , {\ya+ (\yd-\ya)*( (\xb-\xa) /(\xd-\xa))} );

\draw[arrow line, thin,blue] ( {\xa + (\xb-\xa)*((\yc-\ya)/(\yc-\ya))} , \yc )  -- ( {\xa + (\xd-\xa)*((\yc-\ya)/(\yd-\ya))} , \yc );
\draw[arrow line, thin,blue] ( \xc , {\ya+ (\yc-\ya)*( (\xc-\xa) /(\xb-\xa))} ) -- ( \xc , {\ya+ (\yd-\ya)*( (\xc-\xa) /(\xd-\xa))} );

\draw[blue] (\xa,\ya) -- (\xd, \yd) ;

\fill  ( {\xa + (\xc-\xa)*((\yb-\ya)/(\yc-\ya))} , \yb ) circle[radius=0.008]  ;
\fill  ( \xb , {\ya+ (\yc-\ya)*( (\xb-\xa) /(\xc-\xa))} ) circle[radius=0.008]  ;

\fill [red] ( {\xa + (\xb-\xa)*((\yb-\ya)/(\yc-\ya))} , \yb ) circle[radius=0.008]  ;
\fill [red] ( \xb , {\ya+ (\yc-\ya)*( (\xb-\xa) /(\xb-\xa))} ) circle[radius=0.008]  ;
\fill [red] ( {\xa + (\xb-\xa)*((\yc-\ya)/(\yc-\ya))} , \yc ) circle[radius=0.008]  ;
\fill [red] ( \xc , {\ya+ (\yc-\ya)*( (\xc-\xa) /(\xb-\xa))} ) circle[radius=0.008]  ;

\fill [blue] ( {\xa + (\xd-\xa)*((\yb-\ya)/(\yd-\ya))} , \yb ) circle[radius=0.008]  ;
\fill [blue] ( \xb , {\ya+ (\yd-\ya)*( (\xb-\xa) /(\xd-\xa))} ) circle[radius=0.008]  ;
\fill [blue] ( {\xa + (\xd-\xa)*((\yc-\ya)/(\yd-\ya))} , \yc ) circle[radius=0.008]  ;
\fill [blue] ( \xc , {\ya+ (\yd-\ya)*( (\xc-\xa) /(\xd-\xa))} ) circle[radius=0.008]  ;

\fill  (\xc, \yc) circle[radius=0.01] ;

\fill[red] (\xb,\yc) circle[radius=0.010];
\node[below left,text=black,fill=white,inner sep=1pt] at (\xb,\yc)
{\footnotesize $C_{12}$};
\end{tikzpicture}
\caption{{$\dfrac{Y_2}{X_1}<\lambda_3$}}
\label{figure:lem-removal-b}
\end{subfigure}

\caption{Chord replacement across three signals.}
\label{figure:lem-removal}
\end{figure}

These movements are implemented by replacing $P^{j+1},P^{j+2},P^{j+3}$ by the following five columns, retaining all exterior columns. 

\begin{center}
\begin{minipage}{\linewidth}
\centering
\begin{tikzpicture}
\matrix (beforecollision) [matrix of math nodes,
left delimiter={[}, right delimiter={]},
nodes={font=\small,inner xsep=1pt,inner ysep=3pt},
column sep=6pt,row sep=9pt] {
\dfrac{Y_1}{\lambda X_1}P^{j+1}_{\le i} & \dfrac{\lambda X_1-Y_1}{\lambda X_1}P^{j+1}_{\le i} & \dfrac{Y_2-\lambda X_1}{\lambda(X_2-X_1)}P^{j+2}_{\le i} & \dfrac{\lambda X_2-Y_2}{\lambda(X_2-X_1)}P^{j+2}_{\le i} & P^{j+3}_{\le i}\\
P^{j+1}_{\ge i+1} & \dfrac{\lambda X_1-Y_1}{Y_2-Y_1}P^{j+2}_{\ge i+1} & \dfrac{Y_2-\lambda X_1}{Y_2-Y_1}P^{j+2}_{\ge i+1} & \dfrac{\lambda X_2-Y_2}{Y_3-Y_2}P^{j+3}_{\ge i+1} & \dfrac{Y_3-\lambda X_2}{Y_3-Y_2}P^{j+3}_{\ge i+1}\\
};
\node[font=\footnotesize,anchor=south] at ($(beforecollision.north -| beforecollision-1-1.center)+(0,19pt)$) {$j+1$};
\node[font=\footnotesize,anchor=south] at ($(beforecollision.north -| beforecollision-1-2.center)+(0,19pt)$) {$j+2$};
\node[font=\footnotesize,anchor=south] at ($(beforecollision.north -| beforecollision-1-3.center)+(0,19pt)$) {$j+3$};
\node[font=\footnotesize,anchor=south] at ($(beforecollision.north -| beforecollision-1-4.center)+(0,19pt)$) {$j+4$};
\node[font=\footnotesize,anchor=south] at ($(beforecollision.north -| beforecollision-1-5.center)+(0,19pt)$) {$j+5$};
\draw[-{Latex[length=1.6mm]},semithick] ($(beforecollision-1-1.north)+(2pt,1pt)$) to[out=60,in=120,looseness=1.1] ($(beforecollision-1-2.north)+(-2pt,1pt)$);
\draw[-{Latex[length=1.6mm]},semithick] ($(beforecollision-1-3.north)+(2pt,1pt)$) to[out=60,in=120,looseness=1.1] ($(beforecollision-1-4.north)+(-2pt,1pt)$);
\draw[-{Latex[length=1.6mm]},semithick] ($(beforecollision-2-3.south)+(-2pt,-1pt)$) to[out=240,in=300,looseness=1.1] ($(beforecollision-2-2.south)+(2pt,-1pt)$);
\draw[-{Latex[length=1.6mm]},semithick] ($(beforecollision-2-5.south)+(-2pt,-1pt)$) to[out=240,in=300,looseness=1.1] ($(beforecollision-2-4.south)+(2pt,-1pt)$);
\end{tikzpicture}
\end{minipage}
\end{center}

Each leftward movement transfers lower-state mass from the column immediately before the cumulative point to the column immediately after it; each upward movement transfers upper-state mass in the opposite direction. The transfers $j+1\to j+2$ in the lower-state group and $j+3\to j+2$ in the upper-state group continue those of the two-column case; the new moving pair adds the transfers $j+3\to j+4$ and $j+5\to j+4$, respectively. All four replacements are applied simultaneously to preserve the MLRP.

As $\lambda$ increases, the point $(X_1,\lambda X_1)$ moves up while $(Y_2/\lambda,Y_2)$ moves left. They may therefore meet at $C_{12}=(X_1,Y_2)$, as shown in Figure~\ref{figure:lem-removal}. We call this meeting a \emph{collision}. It occurs at $\lambda_{1,2}=Y_2/X_1>\lambda_2$, if this slope is reached. In the above matrix, the factor $Y_2-\lambda X_1$ then vanishes in both row groups of column $j+3$, so the entire column becomes a zero column.

If $\lambda_{1,2}\ge\lambda_3$, as in Figure~\ref{figure:lem-removal-a}, the first matrix governs the entire rotation to the target chord; equality means that the points meet only at the target. If $\lambda_{1,2}<\lambda_3$, as in Figure~\ref{figure:lem-removal-b}, the points meet before the target is reached. The rotation then continues with the following five columns:

\begin{center}
\begin{minipage}{\linewidth}
\centering
\begin{tikzpicture}
\matrix (aftercollision) [matrix of math nodes,
left delimiter={[}, right delimiter={]},
nodes={font=\small,inner xsep=1pt,inner ysep=3pt},
column sep=6pt,row sep=9pt] {
\dfrac{Y_1}{\lambda X_1}P^{j+1}_{\le i} & \dfrac{Y_2-Y_1}{\lambda X_1}P^{j+1}_{\le i} & \dfrac{\lambda X_1-Y_2}{\lambda X_1}P^{j+1}_{\le i} & P^{j+2}_{\le i} & P^{j+3}_{\le i}\\
P^{j+1}_{\ge i+1} & P^{j+2}_{\ge i+1} & \dfrac{\lambda X_1-Y_2}{Y_3-Y_2}P^{j+3}_{\ge i+1} & \dfrac{\lambda(X_2-X_1)}{Y_3-Y_2}P^{j+3}_{\ge i+1} & \dfrac{Y_3-\lambda X_2}{Y_3-Y_2}P^{j+3}_{\ge i+1}\\
};
\node[font=\footnotesize,anchor=south] at ($(aftercollision.north -| aftercollision-1-1.center)+(0,19pt)$) {$j+1$};
\node[font=\footnotesize,anchor=south] at ($(aftercollision.north -| aftercollision-1-2.center)+(0,19pt)$) {$j+2$};
\node[font=\footnotesize,anchor=south] at ($(aftercollision.north -| aftercollision-1-3.center)+(0,19pt)$) {$j+3$};
\node[font=\footnotesize,anchor=south] at ($(aftercollision.north -| aftercollision-1-4.center)+(0,19pt)$) {$j+4$};
\node[font=\footnotesize,anchor=south] at ($(aftercollision.north -| aftercollision-1-5.center)+(0,19pt)$) {$j+5$};
\draw[-{Latex[length=1.6mm]},semithick] ($(aftercollision-1-1.north)+(2pt,1pt)$) to[out=60,in=120,looseness=1.1] ($(aftercollision-1-2.north)+(-2pt,1pt)$);
\draw[-{Latex[length=1.6mm]},semithick] ($(aftercollision-1-2.north)+(2pt,1pt)$) to[out=60,in=120,looseness=1.1] ($(aftercollision-1-3.north)+(-2pt,1pt)$);
\draw[-{Latex[length=1.6mm]},semithick] ($(aftercollision-2-4.south)+(-2pt,-1pt)$) to[out=240,in=300,looseness=1.1] ($(aftercollision-2-3.south)+(2pt,-1pt)$);
\draw[-{Latex[length=1.6mm]},semithick] ($(aftercollision-2-5.south)+(-2pt,-1pt)$) to[out=240,in=300,looseness=1.1] ($(aftercollision-2-4.south)+(2pt,-1pt)$);
\end{tikzpicture}
\end{minipage}
\end{center}

The two matrices agree at the collision. Afterward, the two horizontally moving points precede the two vertically moving points, with
\[
Y_1/\lambda<Y_2/\lambda<X_1<X_2.
\]
Accordingly, the lower-state transfers are $j+1\to j+2$ and $j+2\to j+3$, while the upper-state transfers are $j+4\to j+3$ and $j+5\to j+4$. Column $j+3$ now gains mass from its neighbors in both state groups. No further collision occurs before $\lambda_3$, so continuing these coordinated movements until $\lambda=\lambda_3$ gives the target chord from $(0,0)$ to $(X_3,Y_3)$.

Throughout the rotation, the segments on the rotating chord remain aligned with slope $\lambda$, while the final segment has slope
\[
\frac{Y_3-\lambda X_2}{X_3-X_2}\ge\lambda.
\]
The slope bounds established above also preserve convexity at the joins to the unchanged exterior portions. Thus the selected PP plot remains convex. Meanwhile, every other adjacent-state PP plot is unchanged. Consequently, the MLRP is preserved throughout the rotation. At a collision, the zero column corresponds to a duplicated cumulative point and does not affect the PP plot. 

Again write $\phi=\phi(\lambda)$ for the current experiment, now in $\mathcal E_{m+2}$ and given by the appropriate matrix together with the unchanged exterior columns. On each open subinterval of $(\lambda_2,\lambda_3)$ that avoids $\lambda_{1,2}$, the entries are positive and continuously differentiable. As in the two-column case, differentiating the displayed coefficients expresses the velocity as a nonnegative combination of reverse signal-replacement directions at the current experiment. For the first matrix,
\[
\begin{aligned}
\phi'(\lambda)
={}&\frac1\lambda\,\phi^{j+1\to j+2}_{\le i}+\frac{Y_2}{\lambda(Y_2-\lambda X_1)}\,\phi^{j+3\to j+4}_{\le i}\\
&+\frac{X_1}{Y_2-\lambda X_1}\,\phi^{j+3\to j+2}_{\ge i+1}+\frac{X_2}{Y_3-\lambda X_2}\,\phi^{j+5\to j+4}_{\ge i+1}.
\end{aligned}
\]
For the second matrix,
\[
\begin{aligned}
\phi'(\lambda)
={}&\frac1\lambda\,\phi^{j+1\to j+2}_{\le i}+\frac{Y_2}{\lambda(Y_2-Y_1)}\,\phi^{j+2\to j+3}_{\le i}\\
&+\frac{X_1}{\lambda(X_2-X_1)}\,\phi^{j+4\to j+3}_{\ge i+1}+\frac{X_2}{Y_3-\lambda X_2}\,\phi^{j+5\to j+4}_{\ge i+1}.
\end{aligned}
\]
The four terms in each expression correspond to the four arrows in its matrix, including simultaneous inflows and outflows after the collision. 

Each matrix is continuous on its closed interval of use, and the two matrices agree at $\lambda_{1,2}$ when that slope is reached. Hence $\phi(\lambda)$ is continuous throughout $[\lambda_2,\lambda_3]$ and satisfies condition~(b) of Proposition~\ref{prop:LM-generation} there. At $\lambda_2$, $\phi(\lambda_2)$ is $P^{(2)}$ with zero columns inserted, which is a sequence of splittings. Define $P^{(3)}$ by deleting any zero columns from $\phi(\lambda_3)$, which is a sequence of mergings; the remaining columns are positive. Hence $P^{(3)}$ is reachable from $P^{(2)}$, and therefore from $P$.

\paragraph{Extending the construction by induction.}
For general $k$, to pass from $P^{(k)}$ to $P^{(k+1)}$, the same construction rotates the chord from slope $\lambda_k$ to slope $\lambda_{k+1}$ and retains the ordered list
\[
X_0,\ldots,X_k,\quad Y_1/\lambda,\ldots,Y_k/\lambda,\quad X_{k+1},
\]
again with repetitions retained. Retaining these original cumulative coordinates yields proportional subdivisions that preserve every other adjacent PP plot, while between collisions the velocity is again a nonnegative combination of reverse signal-replacement directions. More pairs of marks may now collide, at slopes $\lambda_{r,s}=Y_s/X_r$ with $1\le r<s\le k$ that lie in the rotation interval. There are finitely many such slopes. Keeping the ordered marks, including repetitions, gives continuity through these collisions even when several occur simultaneously.

The formal induction in Online Appendix~\ref{OA:Lehmann-regularity} verifies these properties for every $k$ and connects the endpoints of successive stages by splittings and mergings using Lemma~\ref{lem:PP-reconstruction}. Repeating the construction finitely many times gives $\tilde P=P^{(j'-j)}$, with the required PP plots and reachable from $P$.
\end{proof}

\section{Proofs for Section~\ref{sec:posterior}}\label{appendix:posterior}

\begin{proof}[Proof of Lemma~\ref{lem:posterior-regularity}]
Since $P_{ij}=\tau^jq_i^j/\mu_i$ is continuous in $(\tau,q)$, continuity of $\mathcal I_m$ implies continuity of $\mathcal W_m$. Conversely, suppose $\mathcal W_m$ is continuous and let $F\subseteq\R$ be closed. Then $\mathcal W_m^{-1}(F)$ is closed in the compact set $\mathcal X_m$, hence compact. Its image under $(\tau,q)\mapsto P$ is exactly $\mathcal I_m^{-1}(F)$, because every experiment has a posterior representation and $\mathcal W_m(\tau,q)=\mathcal I_m(P)$. This image is compact and therefore closed in $\mathcal E_m$. Thus $\mathcal I_m$ is continuous.

On the positive domains, Bayes' formulas~\eqref{eq:bayes} and the inverse formulas $P_{ij}=\tau^jq_i^j/\mu_i$ are differentiable, since all denominators are positive. Applying the chain rule to $\mathcal I_m(P)=\mathcal W_m(\tau,q)$ in each direction therefore proves the equivalence of differentiability. These arguments hold for every $m$, proving the lemma.
\end{proof}

\begin{proof}[Proof of Theorem~\ref{thm:posterior-foc-direct}]
By Lemma~\ref{lem:posterior-regularity}, the induced measure $\mathcal I(P)=\mathcal W(X_P)$ is continuous and differentiable; it is also permutation and split invariant. For $P\in\mathcal E_{m,+}$ and its posterior representation $(\tau,q)$, the chain rule in~\eqref{eq:marginal-value} gives
\[
\begin{aligned}
\langle\nabla\mathcal I_m(P),P^{j\to k}\rangle =\sum_i\mu_iP_{ij}\bigl[L_k(\delta_{\theta_i})-L_j(\delta_{\theta_i})\bigr] =\tau^j\bigl[L_k(q^j)-L_j(q^j)\bigr].
\end{aligned}
\]
The second equality uses $\mu_iP_{ij}=\tau^jq_i^j$ and the fact that $L_j$ and $L_k$ are affine. Since $\tau^j>0$, the signal-replacement inequalities are equivalent to~\eqref{ineq:posterior-blackwell-direct}. Theorem~\ref{thm:BM-general} proves part~(a).

For reverse signal replacements, the same chain-rule calculation, restricted to the relevant states, gives
\[
\begin{aligned}
\langle\nabla\mathcal I_m(P),P^{j\to j+1}_{\le\ell}\rangle
&=\tau^j\Bigl(\sum_{i\le\ell}q_i^j\Bigr)
\bigl[L_{j+1}(\underline q_\ell^j)-L_j(\underline q_\ell^j)\bigr],\\
\langle\nabla\mathcal I_m(P),P^{j\to j-1}_{\ge\ell}\rangle
&=\tau^j\Bigl(\sum_{i\ge\ell}q_i^j\Bigr)
\bigl[L_{j-1}(\overline q_\ell^j)-L_j(\overline q_\ell^j)\bigr].
\end{aligned}
\]
By~\eqref{eq:cross-product} and the endpoint conventions, MLRP and the cutoff conditions agree in the two representations. Thus the reverse-signal-replacement inequalities are equivalent to~\eqref{ineq:posterior-lehmann-lower-direct}--\eqref{ineq:posterior-lehmann-upper-direct}, and Theorem~\ref{thm:LM-general} proves part~(b).
\end{proof}

\section{Proofs for Section~\ref{sec:applications}}\label{appendix:applications}

\subsection{Likelihood- and Posterior-Separable Measures}

\paragraph{Likelihood-separable measures.}

\begin{proof}[Proof of Lemma~\ref{lem:nMLR}]
\emph{(a)$\Rightarrow$(b).} Suppose $P\in\EMLRP_m$ contains adjacent columns $P^j=\alpha x$ and $P^{j+1}=\beta y$. Set $h=\sum_{r\le j}P^r$. Positivity of $x,y$ gives $h\in(0,1)^n$. Summing the MLRP cross-product inequalities over columns on either side of the pair gives
\begin{equation*}
h_{i+1}x_i\le h_ix_{i+1},\qquad
y_{i+1}(1-h_i)\le y_i(1-h_{i+1}),\qquad 1\le i<n.
\end{equation*}
Together with $y\ge_{\mathrm{MLR}}x$, these are~\eqref{eq:nMLR}.

\emph{(b)$\Rightarrow$(a).} Choose
\begin{equation*}
\alpha=\frac12\min_i\frac{h_i}{x_i},\qquad
\beta=\frac12\min_i\frac{1-h_i}{y_i}.
\end{equation*}
Then $P=[h-\alpha x,\ \alpha x,\ \beta y,\ \mathbf1-h-\beta y]$ is positive. The inequalities in~\eqref{eq:nMLR} give the MLR ordering of its four columns, so $P\in\EMLRP_{4,+}$.

\emph{(b)$\Leftrightarrow$(c).} Given $y\ge_{\mathrm{MLR}}x$, condition~\eqref{eq:nMLR} is equivalent to
\begin{equation}\label{eq:nMLR2}
h_{i+1}\le F_i(h_i),\qquad
F_i(t)=\min\left\{\frac{x_{i+1}}{x_i}t,\ 1-\frac{y_{i+1}}{y_i}(1-t)\right\}.
\end{equation}
Each $F_i$ is continuous and increasing. Since $h_1<1=\overline h_1$, induction gives $h_i\le\overline h_i$, proving (c). Conversely, set $h_1^\varepsilon=1-\varepsilon$ and $h_{i+1}^\varepsilon=F_i(h_i^\varepsilon)$. Continuity of the finite recursion gives $\lim_{\varepsilon\downarrow0}h_i^\varepsilon=\overline h_i>0$ for every $i$. Hence all coordinates are positive for sufficiently small $\varepsilon>0$. Moreover, $h_i^\varepsilon<1$ implies $h_{i+1}^\varepsilon\le1-\frac{y_{i+1}}{y_i}(1-h_i^\varepsilon)<1$. Thus $h^\varepsilon\in(0,1)^n$ satisfies~\eqref{eq:nMLR}.
\end{proof}

\begin{proof}[Proof of Proposition~\ref{prop:likelihood-separable}]
\emph{(a).} Sufficiency is proved in the main text. For necessity, suppose $\mathcal I$ is Blackwell monotone. Adding a zero column preserves Blackwell equivalence, so $\psi(\mathbf0)=0$. Splitting and merging proportional columns likewise preserve the measure. For any $x\in[0,1]^n$, equal splittings and merges therefore give $\psi(rx)=r\psi(x)$ for every rational $r\ge0$ with $rx\in[0,1]^n$; continuity extends it to every admissible real $r\ge0$. Finally, merging the first two columns of $[x,\ y,\ \mathbf1-x-y]$, for $x+y\le\mathbf1$, gives $\psi(x+y)\le\psi(x)+\psi(y)$. Thus $\psi$ is sublinear.

\emph{(b).} Suppose $\psi$ is sublinear and $\mathcal I$ is Lehmann monotone. Fix $y\ge_{\mathrm{nMLR}}x$. The proof of Lemma~\ref{lem:nMLR} supplies $P\in\EMLRP_{4,+}$ with $P^2=\alpha x$ and $P^3=\beta y$, where $\alpha,\beta>0$. Positive homogeneity and differentiability imply $\nabla\psi(\alpha x)=\nabla\psi(x)$ and $\nabla\psi(\beta y)=\nabla\psi(y)$. Theorem~\ref{thm:LM-general} gives decreasingness in reverse signal replacement, and Lemma~\ref{lem:binding-extension} includes tied cutoffs. Thus, for every $1\le\ell\le n$,
\begin{align*}
\langle\nabla\mathcal I_4(P),P^{2\to3}_{\le\ell}\rangle
&=\alpha\sum_{i=1}^{\ell}\left(-\frac{\partial\psi(x)}{\partial P_i}+\frac{\partial\psi(y)}{\partial P_i}\right)x_i\leq0,\\
\langle\nabla\mathcal I_4(P),P^{3\to2}_{\ge\ell}\rangle
&=\beta\sum_{i=\ell}^{n}\left(-\frac{\partial\psi(y)}{\partial P_i}+\frac{\partial\psi(x)}{\partial P_i}\right)y_i\leq0.
\end{align*}
Dividing by $\alpha$ and $\beta$ proves~\eqref{ineq:sublinear-LM}. Conversely, sublinearity gives split invariance. Every adjacent pair of columns of $P\in\EMLRP_{m,+}$ is nMLR-ordered, so~\eqref{ineq:sublinear-LM} implies decreasingness in reverse signal replacement. Theorem~\ref{thm:LM-general} proves Lehmann monotonicity on $\EMLRP$.
\end{proof}

\begin{proof}[Proof of Proposition~\ref{prop:likelihood-separable-p-norm}]
\emph{(a).} Let $n = 3$ and $\psi(x)=\sqrt{x^\intercal Ax}$, where
\begin{equation*}
A=\begin{bmatrix}10&10&10\\10&20&10\\10&10&20\end{bmatrix},
\qquad
P=\begin{bmatrix}.78&.1&.1&.02\\.2&.3&.4&.1\\.05&.1&.3&.55\end{bmatrix}
\in\EMLRP_{4,+}.
\end{equation*}
Since $x^\intercal Ax=10[(x_1+x_2+x_3)^2+x_2^2+x_3^2]$, $A$ is positive definite and $\psi$ is continuous, differentiable, and sublinear. The cutoff is admissible because $P_{12}P_{23}-P_{22}P_{13}=1/100>0$, but
\begin{equation*}
\langle\nabla\mathcal I_4(P),P^{2\to3}_{\le1}\rangle
=\frac1{10}\left(\frac8{\sqrt{89/10}}-\frac5{\sqrt{7/2}}\right)>0.
\end{equation*}
Hence $\mathcal I$ is not Lehmann monotone.

\emph{(b).} The weighted $p$-norm is continuous, differentiable, and sublinear. Fix $x,y\in(0,1)^n$ with $y\ge_{\mathrm{MLR}}x$ and a cutoff $1\le\ell\le n$. Define
\[
a_i=\frac{w_ix_i^p}{\psi(x)^p},\qquad
b_i=\frac{w_iy_i^p}{\psi(y)^p},\qquad
A=\sum_{i\le\ell}a_i,\qquad B=\sum_{i\le\ell}b_i.
\]
Differentiating $\psi$ gives
\[
\begin{aligned}
\frac{\partial\psi(x)}{\partial P_{i}}&=\frac{w_ix_i^{p-1}}{\psi(x)^{p-1}}=w_i^{1/p}a_i^{(p-1)/p}, \quad 
\frac{\partial\psi(y)}{\partial P_{i}}&=\frac{w_iy_i^{p-1}}{\psi(y)^{p-1}}=w_i^{1/p}b_i^{(p-1)/p}.
\end{aligned}
\]
Using $x_i=\psi(x)(a_i/w_i)^{1/p}$ and then H\"older's inequality with exponents $p$ and $p/(p-1)$, we obtain
\begin{align*}
\sum_{i\le\ell}\left[\frac{\partial\psi(y)}{\partial P_{i}}-\frac{\partial\psi(x)}{\partial P_{i}} \right]x_i=\psi(x)\left[\sum_{i\le\ell}a_i^{1/p}b_i^{(p-1)/p}-A\right] \le\psi(x)\left[A^{1/p}B^{(p-1)/p}-A\right].
\end{align*}
We next show $B\le A$. Since $y\ge_{\mathrm{MLR}}x$, $b_i a_j\le a_i b_j$ whenever $i\le\ell<j$. Summing these inequalities and using $\sum_i a_i=\sum_i b_i=1$ gives
\[
B(1-A)=\sum_{i\le\ell<j}b_i a_j
\le\sum_{i\le\ell<j}a_i b_j=A(1-B).
\]
This proves the lower-tail inequality in~\eqref{ineq:sublinear-LM}; the upper-tail inequality follows symmetrically. Proposition~\ref{prop:likelihood-separable} therefore proves Lehmann monotonicity.
\end{proof}

\paragraph{Posterior-separable measures.}

Use the posterior notation of Section~\ref{sec:posterior}.

\begin{proof}[Proof of Proposition~\ref{prop:ps-lehmann}]
Sufficiency follows directly from the preceding identities. For necessity, fix full-support beliefs $q'\ge_{\mathrm{MLR}}q$ and take the prior $\mu=(q+q')/2$. The representation $\tau^1=\tau^2=1/2$, $q^1=q$, $q^2=q'$ belongs to $\mathcal X_{2,+}$ and is MLR ordered. Lehmann monotonicity at this prior and Theorem~\ref{thm:posterior-foc-direct}(b) therefore give~\eqref{ineq:ps-lehmann}, with Lemma~\ref{lem:binding-extension} applied to $\mathcal I_v$ covering cutoffs with equal likelihood products.
\end{proof}

\begin{proof}[Proof of Corollary~\ref{cor:mutual-info}]
Let $q'\ge_{\mathrm{MLR}}q$ have full support. At cutoff $\ell$, write $Q=\sum_{i\le\ell}q_i$, $Q'=\sum_{i\le\ell}q'_i$, $R=\sum_{i\ge\ell}q_i$, and $R'=\sum_{i\ge\ell}q'_i$. First-order stochastic dominance gives $Q\ge Q'$ and $R'\ge R$. Conditioning on the respective tails yields
\begin{align*}
D_{\mathrm{KL}}(\underline q_\ell\,\|\,q')
-D_{\mathrm{KL}}(\underline q_\ell\,\|\,q)
&=D_{\mathrm{KL}}(\underline q_\ell\,\|\,\underline q'_\ell)
+\log\frac{Q}{Q'}\ge0,\\
D_{\mathrm{KL}}(\overline q'_\ell\,\|\,q)
-D_{\mathrm{KL}}(\overline q'_\ell\,\|\,q')
&=D_{\mathrm{KL}}(\overline q'_\ell\,\|\,\overline q_\ell)
+\log\frac{R'}{R}\ge0.
\end{align*}
Since $v=-H$ gives $D_v=D_{\mathrm{KL}}$, these are the two inequalities in~\eqref{ineq:ps-lehmann}. Proposition~\ref{prop:ps-lehmann} completes the proof.
\end{proof}

\subsection{A Convex-Order Characterization of the Lehmann Order}\label{appendix:lehmann-cx}

\begin{proof}[Proof of Theorem~\ref{thm:lehmann-cx}]
\emph{(a)$\Rightarrow$(b).} Fix nonzero $a,b\in\R^n_+$ with $a\ge_{\mathrm{MLR}}b$. First suppose that $f$ is convex, differentiable, and affine outside a compact subinterval of $(0,\infty)$. Then $v=v_{a,b,f}$ is continuous on the closed simplex and differentiable on its relative interior, as required by Proposition~\ref{prop:ps-lehmann}.\footnote{Extend $f$ continuously to zero. If $f(z)=cz+d$ for large $z$, set $v(q)=c\langle a,q\rangle$ when $\langle b,q\rangle=0$. The affine tails ensure continuity, including where both posterior means vanish.}

For any nonnegative vectors $s\ge_{\mathrm{MLR}}r$, \eqref{eq:MLR-scales} gives
\[
\langle a,s\rangle\langle b,r\rangle-\langle a,r\rangle\langle b,s\rangle
=\sum_{i<k}(a_kb_i-a_ib_k)(r_is_k-r_ks_i)\ge0.
\]
Fix full-support beliefs $q'\ge_{\mathrm{MLR}}q$ and a cutoff $\ell$. Applying this identity to $(q,q')$, $(\underline q_\ell,q)$, and $(q',\overline q'_\ell)$ yields
\[
t:=\rho_{a,b}(q)\le t':=\rho_{a,b}(q'),\qquad
\langle a-tb,\underline q_\ell\rangle\le0,\qquad
\langle a-t'b,\overline q'_\ell\rangle\ge0.
\]
The chain rule gives the affine approximation to $v$ at $q$:
\[
L_q(r):=v(q)+\nabla v(q)\cdot(r-q)
=f'(t)\langle a,r\rangle+[f(t)-tf'(t)]\langle b,r\rangle.
\]
For $r=\underline q_\ell$, subtracting the two affine approximations gives
\begin{align*}
L_{q'}(r)-L_q(r)
&=\langle b,r\rangle\bigl[f(t')-f(t)-f'(t')(t'-t)\bigr] +[f'(t')-f'(t)]\langle a-tb,r\rangle\le0.
\end{align*}
The bracket is nonpositive by convexity, and the second term is nonpositive because $f'(t')\ge f'(t)$. This is the lower-tail inequality in~\eqref{ineq:ps-lehmann}. Interchanging $q,q'$ and taking $r=\overline q'_\ell$ gives the upper-tail inequality. No division by a tail's benchmark weight is needed, so the argument also covers zero weights. Proposition~\ref{prop:ps-lehmann} gives Lehmann monotonicity of $\mathcal I_{a,b,f}$, so $\mathcal I_{a,b,f}(P)\ge\mathcal I_{a,b,f}(Q)$.

For a general finite convex $f:(0,\infty)\to\R$, choose a compact interval $K\subset(0,\infty)$ containing all ratios induced by $P,Q$ and the prior. Extending $f$ linearly outside a larger compact interval and smoothing gives differentiable convex functions $f_h$, affine outside compact intervals, with $\sup_{z\in K}|f_h(z)-f(z)|\to0$. Apply the preceding argument to each $f_h$ and pass to the limit in~\eqref{eq:pmr-def}. Thus no differentiability or boundary extension of the original $f$ is required. By~\eqref{eq:pmr-def}, these inequalities for all $f$ are equivalent to $Z^P_{a,b}\cx Z^Q_{a,b}$.

\emph{(b)$\Rightarrow$(c).} Take $a=e_k$ and $b=e_i$.

\emph{(c)$\Rightarrow$(a).} Fix adjacent states $i$ and $k=i+1$. Restrict both experiments to these states, with prior $(\mu_i,\mu_k)/(\mu_i+\mu_k)$. Writing $z^j=q_k^j/q_i^j$, the restricted posterior probability of $\theta_k$ after signal $j$ is $z^j/(1+z^j)$. By Bayes' rule, the expected value of any continuous convex $g:[0,1]\to\R$ at this posterior under $P$ is
\[
\frac{\mu_i}{\mu_i+\mu_k}
\sum_j\tau^j_{e_i}(1+z^j)\,g\!\left(\frac{z^j}{1+z^j}\right),
\]
and similarly for $Q$. Since $Z^P_{e_k,e_i}=\sum_j\tau^j_{e_i}\delta_{z^j}$ and $z\mapsto(1+z)g(z/(1+z))$ is convex,~(c) orders the restricted random posteriors in the convex order. Proposition~\ref{prop:blackwell-cx} gives Blackwell, hence Lehmann, dominance of the restrictions. Lemma~\ref{lem:Lehmann-geometry} then gives $\Ss_i(Q)\subseteq\Ss_i(P)$ for every adjacent pair, and hence $P\succeq_LQ$.
\end{proof}

\begin{proof}[Proof of Proposition~\ref{prop:chi-asseyer}]
Condition~(a) consists of ratio tests with vectors $u+\kappa\mathbf 1$ and $\mathbf 1$, where the constant $\kappa$ makes the first nonnegative, and condition~(b) of ratio tests with vectors $\sum_{k>i}e_k$ and $e_i$. A ratio test at another full-support prior $\nu$ is the test at $\mu$ with both vectors rescaled by $\nu_\theta/\mu_\theta$ in each state $\theta$, which preserves their MLR ordering; hence condition~(b) of Theorem~\ref{thm:lehmann-cx} implies~(a) and~(b). Conversely, fix $i<k$ and let the prior concentrate on $\{\theta_i,\theta_k\}$. In the limit, $h_i$ becomes the posterior odds $q_k/q_i$ conditional on $\theta_i$, so~(b) implies condition~(c) of Theorem~\ref{thm:lehmann-cx}. For $u$ with $u_k>u_i$, $\langle u,q\rangle$ becomes a strictly increasing affine function of the posterior probability of $\theta_k$ in the two-state restriction, so~(a) gives convex-order dominance of these posteriors and, by Proposition~\ref{prop:blackwell-cx}, Blackwell dominance on every pair of states; as in the proof of Theorem~\ref{thm:lehmann-cx}, this gives $P\succeq_LQ$.
\end{proof}

\subsection{Quasi-Convex Information Measures}\label{appendix:quasi-convex}

\begin{proof}[Proof of Theorem \ref{thm:blackwell-monotone-characterization}]
Necessity is established in the main text. For sufficiency, suppose that $\mathcal I_m$ is permutation invariant and decreasing in signal replacement. Let $\mathcal M_m$ denote the set of all $m\times m$ stochastic matrices. For any $P,Q\in\mathcal E_m$, $P\succeq_B Q$ if and only if $Q=PM$ for some $M\in\mathcal M_m$.

The set $\mathcal M_m$ is convex, and its extreme points are the matrices with exactly one nonzero entry in each row (see, e.g., \citet{Cao2022}). Write $\textbf{ext}(\mathcal M_m)$ for this set of extreme points, and let $\textbf{ext}_k(\mathcal M_m)$ denote those of rank $k$. For these matrices, the rank equals the number of nonzero columns. In particular, $\textbf{ext}_m(\mathcal M_m)$ consists of permutation matrices, so permutation invariance gives $\mathcal I_m(PE)=\mathcal I_m(P)$ for every $E\in\textbf{ext}_m(\mathcal M_m)$. The following lemma shows how to increase the rank by one, with a single signal replacement recovering the original extreme point.

\begin{lemma}\label{lem:extreme-garbling-replacement}
Suppose $\mathcal I_m:\mathcal E_m\to\R$ is continuous, differentiable, and decreasing in signal replacement. For every $2\le k\le m$ and $E\in\textbf{ext}_{k-1}(\mathcal M_m)$, there exists $E'\in\textbf{ext}_k(\mathcal M_m)$ such that, for every $P\in\mathcal E_m$ and $\lambda\in[0,1]$,
\begin{equation}\label{eq:extreme-garbling-path}
PE'\succeq_B(1-\lambda)PE'+\lambda PE\succeq_B PE.
\end{equation}
Moreover, $\mathcal I_m(PE')\ge\mathcal I_m(PE)$.
\end{lemma}

\begin{proof}[Proof of Lemma~\ref{lem:extreme-garbling-replacement}]
Since $E$ has rank $k-1<m$, it has a column $j$ containing at least two ones and a zero column $j'$. Move one of these ones, within its row, to column $j'$ to obtain $E'$. Column $j$ remains nonzero, so $E'$ has one more nonzero column and hence rank $k$. Replacing signal $j'$ with signal $j$ with probability one undoes this move, giving $E'M_1^{j'\rightarrow j}=E$. With replacement probability $\lambda$, the resulting experiment is therefore
\[
PE'M_\lambda^{j'\rightarrow j}=(1-\lambda)PE'+\lambda PE.
\]
We then garble this intermediate experiment by replacing signal $j'$ with signal $j$ with probability one, leaving all other signals unchanged. The resulting experiment is $PE$, since $M_\lambda^{j'\rightarrow j}M_1^{j'\rightarrow j}=M_1^{j'\rightarrow j}$. Both operations are garblings, establishing the two Blackwell comparisons in \eqref{eq:extreme-garbling-path}.

For the value comparison, first take $R\in\mathcal E_{m,+}$ and set $\phi(\lambda)=R M_\lambda^{j'\rightarrow j}$. This experiment is positive for $0\le\lambda<1$, and $\phi'(\lambda)=\phi(\lambda)^{j'\rightarrow j}/(1-\lambda)$. Hence, for $0<\lambda<1$,
\[
\frac{d}{d\lambda}\mathcal I_m(\phi(\lambda))
=\frac{1}{1-\lambda}\left\langle\nabla\mathcal I_m(\phi(\lambda)),\phi(\lambda)^{j'\rightarrow j}\right\rangle\le0.
\]
The mean value theorem makes $\mathcal I_m(\phi(\lambda))$ decreasing on $(0,1)$, and continuity at the endpoints gives $\mathcal I_m(R)\ge\mathcal I_m(RM_1^{j'\rightarrow j})$. Approximation by positive experiments and continuity extend this inequality to every $R\in\mathcal E_m$. Applying it to $R=PE'$, with $E'M_1^{j'\rightarrow j}=E$, yields $\mathcal I_m(PE')\ge\mathcal I_m(PE)$.
\end{proof}

For any $E\in\textbf{ext}(\mathcal M_m)$, repeated application of Lemma~\ref{lem:extreme-garbling-replacement} increases the rank by one at each step until a permutation matrix $\Pi$ is reached. The resulting comparisons and permutation invariance give $\mathcal I_m(PE)\le\mathcal I_m(P\Pi)=\mathcal I_m(P)$. 

Finally, every $M\in\mathcal M_m$ is a convex combination of extreme points. Thus $Q=PM$ is a convex combination of the experiments $PE$, and quasi-convexity implies
\[
\mathcal I_m(Q)\le\max_{E\in\textbf{ext}(\mathcal M_m)}\mathcal I_m(PE)\le\mathcal I_m(P).
\]
This establishes Blackwell monotonicity.
\end{proof}

\bibliographystyle{econ}
\bibliography{references.bib}

\newpage
\begin{center}
{ \LARGE \textbf{Online Appendix for} \\[10pt]
\textbf{``On the Monotonicity of Information Measures''}\\[10pt] }

\vspace*{20pt}

{ \large Xiaoyu Cheng \qquad Yonggyun Kim }

\vspace*{20pt}
\end{center}

\setcounter{page}{1}
\renewcommand{\thepage}{OA-\arabic{page}}
\setcounter{footnote}{0}

\setcounter{section}{0}
\setcounter{example}{0}
\setcounter{remark}{0}
\renewcommand{\thesection}{OA.\arabic{section}}

\section{Proofs of Technical Lemmas for Lehmann Monotonicity}\label{OA:Lehmann-regularity}

\subsection{Realization and identification from PP plots}\label{OA:PP-reconstruction}
We first prove the following lemma, which implies Lemma~\ref{lem:PP-reconstruction}.

\begin{lemma}[Realization and common refinement]\label{lem:PP-realization}
Any $n-1$ convex polygonal curves from $(0,0)$ to $(1,1)$ with positive finite slopes are the adjacent PP plots of a positive MLRP experiment. If the curves have $b$ bends in total, at most $\max\{2,b+1\}$ signals are needed. Moreover, any two MLRP experiments with the same adjacent PP plots admit a common refinement by adjacent column splittings.
\end{lemma}

\begin{proof}
Write the $i$-th curve as the graph of $c_i$, and define the cumulative probabilities by
\[
z_1(t)=t,\qquad z_{i+1}(t)=c_i(z_i(t)),\qquad 0\le t\le1.
\]
Every coordinate increases strictly from zero to one. Mark the values of $t$ at which any prescribed curve reaches a bend; between these marks, $z(t)$ is linear. Take its increments between successive marked values, including the endpoints, as the columns of $R$. These columns are positive, the rows sum to one, and the adjacent PP plots are the prescribed curves. The ratios $R_{i+1,j}/R_{ij}$ are their increasing slopes; multiplying these ratios gives the same ordering for every other state pair, so $R$ satisfies the MLRP. Each bend is encountered once, so there are at most $b+1$ columns. If there is only one, split it into two.

For the common-refinement claim, join the cumulative column sums of $P$ and $Q$ to form two polygonal paths from $\mathbf0$ to $\mathbf1$. Every point $z$ has $(z_i,z_{i+1})$ on the common PP plot and satisfies $z_1\ge\cdots\ge z_n$, since the plots lie below the diagonal. Any two such points $x,y$ are coordinatewise comparable: if their first difference is $x_r<y_r$, then $x_i\le x_r<1$ for $i\ge r$, and monotonicity of the PP plots, which have no vertical segment before horizontal coordinate one, gives $x_i\le y_i$ successively. Points with the same coordinate sum therefore coincide. This sum runs continuously from $0$ to $n$ on both paths, so the paths are identical. Subdivide the common path at both experiments' cumulative vertices, retaining enough repetitions for their null columns. Each subdivision splits a column proportionally, giving the required common refinement.
\end{proof}

\subsection{Reduction to positive experiments}\label{OA:positive-approximation}

\begin{proof}[Proof of Lemma~\ref{lem:positive-approximation}]
Suppose $\mathcal I$ satisfies the hypotheses of the lemma, and fix $P,Q\in\EMLRP$ with $P\succeq_L Q$.

\medskip\noindent\emph{An order-preserving deformation of the PP plots.} Add null columns, if necessary, to give the two experiments a common signal count $m$; this leaves PP plots unchanged and, under split invariance, also leaves their values under $\mathcal I$ unchanged. For $0<\delta<1/2$, apply the same linear map to every adjacent PP plot of both experiments:
\[
A_\delta(x,y)=((1-\delta)x+\delta y,\ \delta x+(1-\delta)y).
\]
It fixes $(0,0)$ and $(1,1)$ and transforms a slope $\sigma\in[0,\infty]$ into
\[
\sigma_\delta=\frac{\delta+(1-\delta)\sigma}{(1-\delta)+\delta\sigma}\in(0,\infty).
\]
This expression is strictly increasing in $\sigma$, with the endpoint values interpreted by limits. Thus every transformed plot is a convex, strictly increasing polygon with positive finite slopes and no more than $m-1$ bends. Lemma~\ref{lem:PP-realization} realizes these plots by positive experiments. Splitting positive columns further if necessary gives both realizations $P_\delta,Q_\delta$ in the common space $\EMLRP_{\bar m,+}$, where
\[
\bar m=\max\{2,\,1+(n-1)(m-1)\}.
\]

\begin{figure}[htbp]
\centering
\begin{subfigure}[b]{0.46\textwidth}
\centering
\begingroup\color{black}
\begin{tikzpicture}[scale=5]
\path[use as bounding box] (-.21,-.17) rectangle (1.3,1.1);
\draw[->] (0,0) -- (1.05,0) node[below] {\small $i$};
\draw[->] (0,0) -- (0,1.05) node[left] {\small $i+1$};
\fill[red!20,opacity=0.2] (0,1) -- (0,0) -- (.35,0) -- (.75,.2) -- (1,.65) -- (1,1) -- cycle;
\draw[thick,red] (0,0) -- (.35,0) -- (.75,.2) -- (1,.65) -- (1,1);
\fill[blue!20,opacity=0.2] (0,1) -- (0,0) -- (.15,0) -- (.6,.25) -- (1,.8) -- (1,1) -- cycle;
\draw[thick,blue] (0,0) -- (.15,0) -- (.6,.25) -- (1,.8) -- (1,1);
\fill (.35,0) circle[radius=0.01];
\node[anchor=north] at (.35,-.02) {\footnotesize $F_i^1$};
\fill (.75,.2) circle[radius=0.01];
\node[anchor=west] at (.76,.16) {\footnotesize $F_i^2$};
\fill (1,.65) circle[radius=0.01] node[anchor=west] {\footnotesize $F_i^3$};
\node[black,anchor=west] at (.83,.36) {$\Ss_i(P)$};
\node[black] at (.4,.65) {$\Ss_i(Q)$};
\end{tikzpicture}
\endgroup
\caption{Experiments with zero entries}
\end{subfigure}
\hspace*{15pt}
\begin{subfigure}[b]{0.46\textwidth}
\centering
\begingroup\color{black}
\begin{tikzpicture}[scale=5]
\path[use as bounding box] (-.21,-.17) rectangle (1.3,1.1);
\draw[->] (0,0) -- (1.05,0) node[below] {\small $i$};
\draw[->] (0,0) -- (0,1.05) node[left] {\small $i+1$};
\fill[red!20,opacity=0.2] (0,1) -- (0,0) -- (.28,.07) -- (.64,.31) -- (.93,.72) -- (1,1) -- cycle;
\draw[thick,red] (0,0) -- (.28,.07) -- (.64,.31) -- (.93,.72) -- (1,1);
\fill[blue!20,opacity=0.2] (0,1) -- (0,0) -- (.12,.03) -- (.53,.32) -- (.96,.84) -- (1,1) -- cycle;
\draw[thick,blue] (0,0) -- (.12,.03) -- (.53,.32) -- (.96,.84) -- (1,1);
\fill (.28,.07) circle[radius=0.01];
\node[anchor=north] at (.28,-.02) {\footnotesize $A_\delta F_i^1$};
\fill (.64,.31) circle[radius=0.01];
\node[anchor=west] at (.65,.27) {\footnotesize $A_\delta F_i^2$};
\fill (.93,.72) circle[radius=0.01] node[anchor=west] {\footnotesize $A_\delta F_i^3$};
\node[black,anchor=west] at (.8,.44) {$\Ss_i(P_\delta)$};
\node[black] at (.4,.65) {$\Ss_i(Q_\delta)$};
\end{tikzpicture}
\endgroup
\caption{Positive approximation}
\end{subfigure}
\caption{{Positive approximation preserving convexity and inclusion.}}
\label{figure:positive-approximation}
\end{figure}

For each adjacent state pair, the upper region of the transformed plot is
\[
\Ss_i(P_\delta)
=\operatorname{conv}\bigl(\{(0,1)\}\cup A_\delta(\Ss_i(P))\bigr).
\]
Indeed, the transformed cumulative vertices generate the transformed plot, and the extra point
\[
A_\delta(0,1)
=\delta(0,0)+(1-2\delta)(0,1)+\delta(1,1)
\]
already lies in the upper region generated with the fixed anchor $(0,1)$. Since $\Ss_i(Q)\subseteq\Ss_i(P)$, the same formula gives $\Ss_i(Q_\delta)\subseteq\Ss_i(P_\delta)$. Lemma~\ref{lem:Lehmann-geometry} therefore gives $P_\delta\succeq_L Q_\delta$. As $A_\delta$ converges to the identity, these upper regions converge in Hausdorff distance to the original ones.

\medskip\noindent\emph{Convergence of information values in a fixed dimension.}
Compactness of $\EMLRP_{\bar m}\times\EMLRP_{\bar m}$ gives a common subsequence $\delta_r\downarrow0$ with
\[
P_{\delta_r}\to P^*,\qquad
Q_{\delta_r}\to Q^*,\qquad P^*,Q^*\in\EMLRP_{\bar m}.
\]
The upper regions are convex hulls of finitely many cumulative vertices depending continuously on the matrix entries. Hence $\Ss_i(P^*)=\Ss_i(P)$ and $\Ss_i(Q^*)=\Ss_i(Q)$ for every $i$. Continuity and split invariance of $\mathcal I$ and Lemma~\ref{lem:PP-reconstruction} imply
\[
\mathcal I(P_{\delta_r})\to \mathcal I(P^*)=\mathcal I(P),\qquad
\mathcal I(Q_{\delta_r})\to \mathcal I(Q^*)=\mathcal I(Q).
\]
Notice that the limit matrices need not equal the original matrices: equality of PP plots and split invariance are what identify their values under $\mathcal I$. Each pair $P_{\delta_r},Q_{\delta_r}$ is positive and satisfies $P_{\delta_r}\succeq_L Q_{\delta_r}$, so the assumed comparison gives $\mathcal I(P_{\delta_r})\ge \mathcal I(Q_{\delta_r})$. Passing to the limit yields
\[
\mathcal I(P)=\lim_{r\to\infty}\mathcal I(P_{\delta_r})\ge\lim_{r\to\infty}\mathcal I(Q_{\delta_r})=\mathcal I(Q).
\]
This proves the comparison for all $P,Q\in\EMLRP$ with $P\succeq_L Q$.
\end{proof}

\subsection{Extension to equal likelihood products}\label{OA:binding-extension}

\begin{proof}[Proof of Lemma~\ref{lem:binding-extension}]
Fix $P\in\EMLRP_{m,+}$, $1\le j<m$, and $1\le\ell<n$ with
\[
P_{\ell j}P_{\ell+1,j+1}=P_{\ell,j+1}P_{\ell+1,j},
\qquad D=P^{j\to j+1}_{\le\ell}.
\]
Since $P$ is interior relative to the row-sum affine space, $\langle\nabla\mathcal I(P),D\rangle$ is defined even though the replacement direction leaves the MLRP set at the tied cutoff.

\medskip\noindent\emph{Step 1: a strict perturbation.}
For $\delta\ge0$, define
\[
(Q_\delta)_{sr}=\frac{P_{sr}e^{\delta sr}}
{\sum_{b=1}^m P_{sb}e^{\delta sb}},
\qquad 1\le s\le n,\quad 1\le r\le m.
\]
This curve is differentiable at zero, with $Q_0=P$. Write
\[
A:=\lim_{\delta\downarrow0}\frac{Q_\delta-P}{\delta}.
\]
For $s<s'$ and $r<r'$,
\[
\frac{(Q_\delta)_{sr}(Q_\delta)_{s'r'}}
{(Q_\delta)_{sr'}(Q_\delta)_{s'r}}
=\frac{P_{sr}P_{s'r'}}{P_{sr'}P_{s'r}}
e^{\delta(s'-s)(r'-r)}>1\qquad(\delta>0).
\]
Thus all MLR inequalities of $Q_\delta$ are strict.

\medskip\noindent\emph{Step 2: a replacement of length $\kappa\delta$.}
Use the feasible cutoff length $\epsilon_*$ from the proof of Lemma~\ref{lem:Lehmann-marginal-general}. At the tied cutoff of $P$, the tilted cross-product ratio is $e^\delta$, so
\[
\epsilon_*(Q_\delta)
=\frac{(Q_\delta)_{\ell+1,j+1}}
{(Q_\delta)_{\ell+1,j}+(Q_\delta)_{\ell+1,j+1}}
(1-e^{-\delta}),\qquad
\frac{\epsilon_*(Q_\delta)}{\delta}
\longrightarrow
\eta:=\frac{P_{\ell+1,j+1}}{P_{\ell+1,j}+P_{\ell+1,j+1}}>0.
\]
Choose a constant $0<\kappa<\eta$. For all sufficiently small $\delta>0$, the segment
\[
\phi_\delta(t)=Q_\delta+t(Q_\delta)^{j\to j+1}_{\le\ell},
\qquad 0\le t\le\kappa\delta,
\]
is positive and MLRP, with a strict cutoff. Let $R_\delta=\phi_\delta(\kappa\delta)$.

\medskip\noindent\emph{Step 3: comparison of the two values.}
The source-column entries in the lower tail of $\phi_\delta(t)$ are $(1-t)$ times those of $Q_\delta$. Consequently, the chain rule and the strict-cutoff inequality~\eqref{ineq:Lehmann-finite-1} imply
\[
\frac{d}{dt}\mathcal I(\phi_\delta(t))
=\frac{1}{1-t}
\bigl\langle\nabla\mathcal I(\phi_\delta(t)),
(\phi_\delta(t))^{j\to j+1}_{\le\ell}\bigr\rangle\le0.
\]
The mean value theorem gives $\mathcal I(R_\delta)\le\mathcal I(Q_\delta)$.

\medskip\noindent\emph{Step 4: cancellation at the common limit.}
Continuity of the replacement direction gives
\[
\lim_{\delta\downarrow0}\frac{R_\delta-Q_\delta}{\kappa\delta}=D,
\qquad
\lim_{\delta\downarrow0}\frac{R_\delta-P}{\delta}=A+\kappa D.
\]
Differentiability of $\mathcal I$ at $P$, together with these limits and the definition of $A$, yields
\[
\langle\nabla\mathcal I(P),D\rangle
=\lim_{\delta\downarrow0}
\frac{\mathcal I(R_\delta)-\mathcal I(Q_\delta)}{\kappa\delta}\le0.
\]

The upper-tail inequality follows by the same argument, using the direction $P^{j\to j-1}_{\ge\ell}$ and inequality~\eqref{ineq:Lehmann-finite-2}.
\end{proof}

\subsection{General induction for chord replacement}\label{OA:general-removal}\label{OA:rotation-verification}

This subsection completes the proof of Lemma~\ref{lem:removal} by extending the two-column and three-column constructions in Appendix~\ref{sec:chord-replacement-proof}. The same rotation applies to any number of columns, retaining the original cumulative coordinates to preserve the other PP plots.

\begin{proof}[Completion of the proof of Lemma~\ref{lem:removal}]
\noindent\emph{Coordinates and induction.}
Fix the original experiment $P\in\EMLRP_{m,+}$, the state pair $i,i+1$, and the block of $j'-j$ columns $j+1,\ldots,j'$. As in Appendix~\ref{sec:chord-replacement-proof}, set
\[
X_r=\sum_{s=1}^r P_{i,j+s},\qquad
Y_r=\sum_{s=1}^r P_{i+1,j+s},\qquad
\lambda_r=\frac{Y_r}{X_r}\quad(1\le r\le j'-j),
\]
with $X_0=Y_0=0$. Write $F^j=F_i^j$ for the fixed state pair, as in the main appendix. Coordinates are measured relative to $F^j$, so $(X_r,Y_r)$ are the original vertices and $\lambda_r$ is the chord slope from $(0,0)$ to $(X_r,Y_r)$. These quantities remain fixed. The original segment slopes $d_r=P_{i+1,j+r}/P_{i,j+r}$ are positive and increasing, and their weighted averages satisfy $0<\lambda_1\le\cdots\le\lambda_{j'-j}$.

The cases $j'-j\le3$ are established in Appendix~\ref{sec:chord-replacement-proof}. Suppose $j'-j\ge4$ and $3\le k<j'-j$, and an experiment $P^{(k)}\in\EMLRP_+$ has been constructed whose selected PP plot replaces the original portion from $(0,0)$ to $(X_k,Y_k)$ by its chord, leaves the rest of that plot and every other adjacent PP plot unchanged, and is reachable from $P$. We construct $P^{(k+1)}$ with the corresponding properties, reachable from $P^{(k)}$.

\medskip\noindent\emph{Rotating the chord.}
If $\lambda_k=\lambda_{k+1}$, the current chord and the next original segment are already aligned, so take $P^{(k+1)}=P^{(k)}$. Otherwise, increase $\lambda$ from $\lambda_k$ to $\lambda_{k+1}$. The selected PP plot follows the line of slope $\lambda$ from $(0,0)$ to $(X_k,\lambda X_k)$, followed by the segment joining $(X_k,\lambda X_k)$ to the fixed endpoint $(X_{k+1},Y_{k+1})$. The exterior portions remain unchanged. As in the three-column case, the last slope satisfies
\[
\lambda\le\frac{Y_{k+1}-\lambda X_k}{X_{k+1}-X_k}\le d_{k+1}.
\]
Together with $\lambda\ge\lambda_k\ge d_1$, these bounds ensure convexity, including at the joins to the exterior portions. This portion of the PP plot increases pointwise with $\lambda$, from the plot at stage $k$ to the chord through $(X_{k+1},Y_{k+1})$.

Retain the horizontally moving points $(Y_s/\lambda,Y_s)$ and the vertically moving points $(X_r,\lambda X_r)$ for $1\le r,s\le k$. All lie on the rotating chord because $Y_s/\lambda\le X_k$. Together with the two endpoints, they subdivide the selected plot into $2k+1$ segments, counting a repeated point as a segment of zero length. Write $x_0(\lambda)\le\cdots\le x_{2k+1}(\lambda)$ for the ordered list
\begin{equation}\label{eq:rotation-induction-grid}
X_0,\ldots,X_k,\quad Y_1/\lambda,\ldots,Y_k/\lambda,\quad X_{k+1},
\end{equation}
with all repetitions retained. Their vertical coordinates are
\[
y_a(\lambda)=\lambda x_a(\lambda)\quad(0\le a\le2k),
\qquad y_{2k+1}(\lambda)=Y_{k+1}.
\]

\medskip\noindent\emph{The experiment along the rotation.}
The retained points include every original horizontal and vertical cumulative coordinate of the first $k+1$ columns in the active block. Hence, for each positive gap, there are original column indices $r,s$ such that
\[
[x_{a-1},x_a]\subseteq[X_{r-1},X_r],\qquad
[y_{a-1},y_a]\subseteq[Y_{s-1},Y_s].
\]
Define $\phi=\phi(\lambda)$ by replacing the first $k+1$ columns of the active block in the original $P$ with the following $2k+1$ columns, retaining all exterior columns:
\begin{equation}\label{eq:rotation-proportional-columns}
\phi^{j+a}=
\begin{bmatrix}
\dfrac{x_a-x_{a-1}}{X_r-X_{r-1}}\,P^{j+r}_{\le i}\\[7pt]
\dfrac{y_a-y_{a-1}}{Y_s-Y_{s-1}}\,P^{j+s}_{\ge i+1}
\end{bmatrix},\qquad a=1,\ldots,2k+1.
\end{equation}
As in the main appendix, the row restrictions in this block matrix omit the zero entries outside the indicated state group. For a repeated point, set the entire column to zero. The indices $r,s$ in \eqref{eq:rotation-proportional-columns} are determined separately for each positive gap. This is the general form of the matrices in the two-column and three-column cases: each row group subdivides original columns into proportional parts, preserving their order.

The gaps traverse each original interval completely, so these parts sum to the original active columns in each row group. Every column is either positive or a zero column, and $\phi(\lambda)\in\mathcal E_{m+k}$. Its selected PP plot is exactly the rotating plot described above, since its entries in rows $i$ and $i+1$ are $x_a-x_{a-1}$ and $y_a-y_{a-1}$. Every other adjacent-state PP plot is unchanged. Convexity of the selected plot therefore gives MLRP: after deleting zero columns, all adjacent-state likelihood ratios are increasing, which implies the same ordering for any pair of states by multiplication. Reinserting zero columns preserves the MLRP inequalities. Thus $\phi(\lambda)\in\EMLRP_{m+k}$ throughout the rotation.

As in Figure~\ref{figure:lem-removal}, a collision occurs when a vertically moving point meets a horizontally moving point, at $C_{rs}=(X_r,Y_s)$ and slope $\lambda_{r,s}=Y_s/X_r$. The possible collision slopes after the starting slope satisfy
\[
\lambda_k<\frac{Y_s}{X_r}\le\lambda_{k+1},\qquad 1\le r<s\le k.
\]
Indeed, if $s\le r$, then $Y_s/X_r\le\lambda_r\le\lambda_k$. There are finitely many such slopes. At each collision the gap between the meeting points vanishes, so the corresponding column in \eqref{eq:rotation-proportional-columns} becomes a zero column. At an interior collision, the two points exchange order and the same formula continues the rotation with the new ordering.

\medskip\noindent\emph{Reverse signal replacements.}
On each open interval between successive collision or endpoint slopes, the moving points are distinct, their order and the column indices in \eqref{eq:rotation-proportional-columns} are fixed, and every entry of $\phi$ is positive and continuously differentiable. Fix such an interval, and let $\mathcal H$ and $\mathcal V$ denote the positions of its horizontally and vertically moving interior points.

For $a\in\mathcal H$, the point $(x_a,y_a)=(Y_s/\lambda,Y_s)$ moves left, with $x_a'=-Y_s/\lambda^2<0$. Its upper-state cumulative probabilities are fixed. Its lower-state cumulative probabilities vary proportionally to the original column containing the gap immediately before it, so
\[
\frac{d}{d\lambda}\sum_{b=1}^a\phi^{j+b}_{\le i}
=\frac{x_a'}{x_a-x_{a-1}}\,\phi^{j+a}_{\le i}.
\]
Thus this movement transfers lower-state mass from column $j+a$ to column $j+a+1$, with positive rate $-x_a'/(x_a-x_{a-1})$.

For $a\in\mathcal V$, the point $(x_a,y_a)=(X_r,\lambda X_r)$ moves up, with $y_a'=X_r>0$. Its lower-state cumulative probabilities are fixed. The corresponding calculation uses the gap immediately after it and gives
\[
\frac{d}{d\lambda}\sum_{b=1}^a\phi^{j+b}_{\ge i+1}
=\frac{y_a'}{y_{a+1}-y_a}\,\phi^{j+a+1}_{\ge i+1}.
\]
This movement transfers upper-state mass from column $j+a+1$ to column $j+a$, with positive rate $y_a'/(y_{a+1}-y_a)$.

The exterior cumulative vectors are fixed. Differencing these cumulative derivatives therefore gives
\begin{equation}\label{eq:rotation-velocity}
\phi'(\lambda)
=\sum_{a\in\mathcal H}\frac{-x_a'}{x_a-x_{a-1}}\,\phi^{j+a\to j+a+1}_{\le i}
+\sum_{a\in\mathcal V}\frac{y_a'}{y_{a+1}-y_a}\,\phi^{j+a+1\to j+a}_{\ge i+1},
\end{equation}
a linear combination with positive coefficients of reverse signal-replacement directions at the current experiment, including simultaneous inflows and outflows.

\medskip\noindent\emph{Continuity and the induction endpoints.}
The ordered coordinates $x_a(\lambda)$ are continuous: changing every entry of a finite list by at most a given amount changes each ordered entry by at most that amount. Hence the $y_a(\lambda)$ are also continuous. At a collision, both gaps in the corresponding column tend to zero, so the full column tends to zero from either side. Every nonvanishing gap retains its original column indices locally, and its proportional coefficients have matching limits. Thus \eqref{eq:rotation-proportional-columns} is continuous through collisions, including simultaneous collisions and coincidences at either endpoint. Keeping the zero columns gives a continuous path $\phi(\lambda)$ in the fixed space $\EMLRP_{m+k}$ throughout $[\lambda_k,\lambda_{k+1}]$. It is positive except at the finitely many collision and endpoint slopes and, between them, differentiable with velocity~\eqref{eq:rotation-velocity}. Hence $\phi$ satisfies condition~(b) of Proposition~\ref{prop:LM-generation} on $[\lambda_k,\lambda_{k+1}]$.

At $\lambda_k$, the experiment $\phi(\lambda_k)$ has the same adjacent PP plots as $P^{(k)}$. Their signal counts and subdivisions may differ, but by Lemma~\ref{lem:PP-reconstruction} they admit a common refinement by adjacent column splittings, so splittings and mergings lead from $P^{(k)}$ to $\phi(\lambda_k)$. At $\lambda_{k+1}$, the selected plot is the target chord. Delete any zero columns from $\phi(\lambda_{k+1})$, which is a sequence of mergings, and call the result $P^{(k+1)}$. The remaining columns are positive; the retained distinct coordinates $X_0,\ldots,X_{k+1}$ ensure at least $k+1$ such columns. The splittings and mergings leading to $\phi(\lambda_k)$, the path $\phi$ on $[\lambda_k,\lambda_{k+1}]$, and the final mergings show that $P^{(k+1)}$ is reachable from $P^{(k)}$, and hence, by transitivity, from $P$. Thus $P^{(k+1)}$ has the required properties, completing the induction step. Starting from the three-column case and repeating finitely many times gives $\tilde P=P^{(j'-j)}$, as required.
\end{proof}

Figure~\ref{figure:lem-removal-collisions} illustrates two rotations with several interior collisions. Panel~(a) shows the step from $k=3$ to $k+1=4$; panel~(b) shows a further example with $k=4$. Filled dots are reached collision points, while hollow dots are other candidate intersections. The displayed slope inequalities specify their order. The number of possible pairs grows with $k$, but the number actually reached depends on the original coordinates and the interval $[\lambda_k,\lambda_{k+1}]$.

\begin{figure}[h]
\centering
\begingroup
\color{black}
\captionsetup{labelfont={color=black},textfont={color=black}}
\captionsetup[subfigure]{labelfont={color=black},textfont={color=black}}
\begin{subfigure}[b]{0.48\textwidth}
\centering
\begin{tikzpicture}[scale=5.4,every node/.style={text=black}]
\tikzset{collision motion/.style={red,thin,
postaction={decorate},decoration={markings,
mark=at position .35 with {\arrow[scale=.55]{Latex}},
mark=at position .80 with {\arrow[scale=.55]{Latex}}}}}
\draw[black,->] (0,0)--(1.02,0) node[below] {\small $i$};
\draw[black,->] (0,0)--(0,1.02) node[left] {\small $i+1$};
\node[below left,align=right,inner sep=2pt] at (0,0)
{\scriptsize $F^j$\\[-2pt]\scriptsize $(X_0,Y_0)$};
\foreach \x in {.25,.5,.75}
{\draw[black!45,dotted,very thin] (\x,0)--(\x,.375);}
\foreach \y in {.05,.15,.375}
{\draw[black!45,dotted,very thin] (0,\y)--(.75,\y);}
\foreach \x/\r in {.25/1,.5/2}
{\node[below,inner sep=2pt] at (\x,0) {\scriptsize $X_{\r}$};}
\foreach \y/\r in {.05/1,.15/2}
{\node[left,inner sep=2pt] at (0,\y) {\scriptsize $Y_{\r}$};}
\node[below,inner sep=2pt] at (.75,0) {\scriptsize $X_3=X_k$};
\node[left,inner sep=2pt] at (0,.375) {\scriptsize $Y_3=Y_k$};
\draw[black,thin] (0,0)--(.25,.05)--(.5,.15)--(.75,.375)--(.875,.875)--(.89,1);
\foreach \x/\y in {.25/.05,.5/.15,.75/.375,.875/.875}
{\fill[black] (\x,\y) circle[radius=.007];}
\node[below right] at (.75,.375) {\scriptsize $F^{j+3}$};
\node[above left] at (.875,.875) {\scriptsize $F^{j+4}$};
\draw[black,dotted] (0,0)--(.75,.375);
\draw[blue,thin] (0,0)--(.875,.875);
\foreach \y in {.05,.15,.375}
{\draw[collision motion] ({2*\y},\y)--(\y,\y);
\fill[black] ({2*\y},\y) circle[radius=.006];
\fill[blue] (\y,\y) circle[radius=.006];}
\foreach \x in {.25,.5,.75}
{\draw[collision motion] (\x,{.5*\x})--(\x,\x);
\fill[black] (\x,{.5*\x}) circle[radius=.006];
\fill[blue] (\x,\x) circle[radius=.006];}
\fill[red] (.25,.15) circle[radius=.010];
\fill[red] (.5,.375) circle[radius=.010];
\node[above right,yshift=3pt,inner sep=1pt,fill=white] at (.25,.15) {\footnotesize $C_{12}$};
\node[above right,inner sep=2pt] at (.5,.375) {\footnotesize $C_{23}$};
\draw[red,fill=white] (.25,.375) circle[radius=.009];
\node[above left,inner sep=2pt] at (.25,.375) {\scriptsize $C_{13}$};
\node at (.24,.73) {$\Ss_i(P)$};
\end{tikzpicture}
\vspace{-7pt}
{\small
\[
\lambda_3=\tfrac12
<\underbrace{\tfrac35}_{C_{12}}
<\underbrace{\tfrac34}_{C_{23}}
<\lambda_4=1.
\]}
\vspace{-14pt}
\caption{{Two successive collisions ($k=3$).}}
\end{subfigure}
\hfill
\begin{subfigure}[b]{0.48\textwidth}
\centering
\begin{tikzpicture}[scale=5.4,every node/.style={text=black}]
\tikzset{collision motion/.style={red,thin,
postaction={decorate},decoration={markings,
mark=at position .35 with {\arrow[scale=.55]{Latex}},
mark=at position .80 with {\arrow[scale=.55]{Latex}}}}}
\draw[black,->] (0,0)--(1.02,0) node[below] {\small $i$};
\draw[black,->] (0,0)--(0,1.02) node[left] {\small $i+1$};
\node[below left,align=right,inner sep=2pt] at (0,0)
{\scriptsize $F^j$\\[-2pt]\scriptsize $(X_0,Y_0)$};
\foreach \x in {.25,.45,.6,.8}
{\draw[black!45,dotted,very thin] (\x,0)--(\x,.4);}
\foreach \y in {.025,.1,.2,.4}
{\draw[black!45,dotted,very thin] (0,\y)--(.8,\y);}
\foreach \x/\r in {.25/1,.45/2,.6/3}
{\node[below,inner sep=2pt] at (\x,0) {\scriptsize $X_{\r}$};}
\foreach \y/\r in {.025/1,.1/2,.2/3}
{\node[left,inner sep=2pt] at (0,\y) {\scriptsize $Y_{\r}$};}
\node[below,inner sep=2pt] at (.8,0) {\scriptsize $X_4=X_k$};
\node[left,inner sep=2pt] at (0,.4) {\scriptsize $Y_4=Y_k$};
\draw[black,thin] (0,0)--(.25,.025)--(.45,.1)--(.6,.2)--(.8,.4)--(.9,.9)--(.915,1);
\foreach \x/\y in {.25/.025,.45/.1,.6/.2,.8/.4,.9/.9}
{\fill[black] (\x,\y) circle[radius=.007];}
\node[below right] at (.8,.4) {\scriptsize $F^{j+4}$};
\node[above left] at (.9,.9) {\scriptsize $F^{j+5}$};
\draw[black,dotted] (0,0)--(.8,.4);
\draw[blue,thin] (0,0)--(.9,.9);
\foreach \y in {.025,.1,.2,.4}
{\draw[collision motion] ({2*\y},\y)--(\y,\y);
\fill[black] ({2*\y},\y) circle[radius=.006];
\fill[blue] (\y,\y) circle[radius=.006];}
\foreach \x in {.25,.45,.6,.8}
{\draw[collision motion] (\x,{.5*\x})--(\x,\x);
\fill[black] (\x,{.5*\x}) circle[radius=.006];
\fill[blue] (\x,\x) circle[radius=.006];}
\fill[red] (.6,.4) circle[radius=.010];
\fill[red] (.25,.2) circle[radius=.010];
\fill[red] (.45,.4) circle[radius=.010];
\node[above right,inner sep=2pt] at (.6,.4) {\footnotesize $C_{34}$};
\node[above left,yshift=4pt,inner sep=1pt,fill=white] at (.25,.2) {\footnotesize $C_{13}$};
\node[above right,inner sep=2pt] at (.45,.4) {\footnotesize $C_{24}$};
\foreach \x/\y in {.25/.1,.45/.2,.25/.4}
{\draw[red,fill=white] (\x,\y) circle[radius=.009];}
\node[below left,inner sep=1pt,fill=white] at (.25,.1) {\scriptsize $C_{12}$};
\node[below left,inner sep=1pt,fill=white] at (.45,.2) {\scriptsize $C_{23}$};
\node[above left,inner sep=2pt] at (.25,.4) {\scriptsize $C_{14}$};
\node at (.24,.73) {$\Ss_i(P)$};
\end{tikzpicture}
\vspace{-7pt}
{\small
\[
\lambda_4=\tfrac12
<\underbrace{\tfrac23}_{C_{34}}
<\underbrace{\tfrac45}_{C_{13}}
<\underbrace{\tfrac89}_{C_{24}}
<\lambda_5=1.
\]}
\vspace{-14pt}
\caption{{Three successive collisions ($k=4$).}}
\end{subfigure}
\caption{{Successive collisions}}
\label{figure:lem-removal-collisions}
\endgroup
\end{figure}

\section{Local Utilities}\label{OA:posterior-local-utility}

Fix a full-support prior $\mu$. Let $\mathcal X_+$ denote the laws induced by positive representations $(\tau,q)\in\mathcal X_{m,+}$, with the topology of weak convergence.

\subsection{Regularity}\label{OA:local-utility-regularity}
\begin{enumerate}
\item[\hypertarget{cond:posterior-local-utility}{(R)}] The measure $\mathcal W$ admits local utilities $v(\cdot;X)$ for all $X\in\mathcal X_+$. They can be chosen differentiable on $\operatorname{relint}\Delta(\Theta)$, with posterior gradient $(q,X)\mapsto\nabla_q v(q;X)$ jointly continuous on $\operatorname{relint}\Delta(\Theta)\times\mathcal X_+$.
\end{enumerate}

We take posterior gradients relative to the affine hull of the simplex, with coordinates summing to zero. Normalize each local utility by subtracting $v(\mu;X)$. Since $Y-X$ has total mass zero, the normalized utility represents the same outer-mixture derivative. With this normalization,
\[
v(q;X)=\int_0^1\nabla_q v\bigl(\mu+s(q-\mu);X\bigr)\cdot(q-\mu)\,ds
\qquad\text{for }q\in\operatorname{relint}\Delta(\Theta).
\]
Joint continuity of the posterior gradient thus gives joint continuity of $v(q;X)$ on the same domain.

\begin{lemma}\label{lem:local-utility-derivatives}
Let $\mathcal W:\mathcal X\to\R$ be continuous and satisfy \hyperlink{cond:posterior-local-utility}{(R)}. Then $\mathcal W$ is differentiable. At every $(\tau,q)\in\mathcal X_{m,+}$ inducing $X=\sum_j\tau^j\delta_{q^j}$, a coordinate gradient of $\mathcal W_m$ is given by
\begin{equation*}
\frac{\partial\mathcal W_m}{\partial\tau^j}=v(q^j;X),
\qquad
\frac{\partial\mathcal W_m}{\partial q_i^j}=\tau^j\frac{\partial v(q^j;X)}{\partial q_i}.
\end{equation*}
Consequently, for all signals $s_j,s_k$ and all $r\in\Delta(\Theta)$,
\begin{equation*}
L_k(r)-L_j(r)=D_{v(\cdot;X)}(r\,\|\,q^j)-D_{v(\cdot;X)}(r\,\|\,q^k),
\end{equation*}
where $D_{v(\cdot;X)}$ is defined by~\eqref{eq:bregman} with $v(\cdot;X)$ in place of $v$.
\end{lemma}

\begin{proof}
For $X,Y\in\mathcal X_+$, set $X_t=(1-t)X+tY$. Applying the local-utility representation toward $Y$ and toward $X$ gives the same derivative $\int v(q;X_t)(Y-X)(dq)$ for $0<t<1$. This derivative is continuous in $t$ by joint continuity of the normalized local utility on the combined finite support of $X$ and $Y$. Using that support as a fixed coordinate representation gives continuity of $\mathcal W(X_t)$ at the endpoints. The fundamental theorem of calculus therefore yields
\begin{equation*}
\mathcal W(Y)-\mathcal W(X)
=\int_0^1\!\int v(q;X_t)\,(Y-X)(dq)\,dt.
\end{equation*}
Fix $z=(\tau,q)\in\mathcal X_{m,+}$, with $X=\sum_j\tau^j\delta_{q^j}$, and let $z'=(\tau',q')\to z$ within this domain, inducing $Y=\sum_j\tau'^j\delta_{q'^j}$. Let $g\in\R^{m+mn}$ have components $g_{\tau^j}=v(q^j;X)$ and $g_{q_i^j}=\tau^j\frac{\partial v(q^j;X)}{\partial q_i}$. Applying the first-order approximation of $v$ to each term $\tau'^jv(q'^j;X_t)-\tau^jv(q^j;X_t)$ in the integral identity gives
\begin{equation*}
\lim_{\substack{z'\to z\\z'\in\mathcal X_{m,+},\ z'\ne z}}
\frac{\left|\mathcal W_m(z')-\mathcal W_m(z)-\langle g,z'-z\rangle\right|}
{\|z'-z\|}=0.
\end{equation*}
Indeed, the finitely many posteriors stay in a compact subset of the relative interior and $X_t\to X$ weakly uniformly in $t\in[0,1]$. Joint continuity of $v$ and its posterior gradient therefore makes the normalized remainders tend to zero uniformly in $t$. This proves differentiability in the sense of Section~\ref{sec:posterior-foc}.

By~\eqref{eq:marginal-plane}, these coefficients give
\begin{align*}
 L_k(r)-L_j(r) = &v(q^k;X)-v(q^j;X)
        +\nabla v(q^k;X)\cdot(r-q^k)
        -\nabla v(q^j;X)\cdot(r-q^j) \\
 = & D_{v(\cdot;X)}(r\,\|\,q^j)-D_{v(\cdot;X)}(r\,\|\,q^k).
\end{align*}
\end{proof}

\subsection{Blackwell Monotonicity and \texorpdfstring{\citet{grant1998intrinsic}}{Grant, Kajii, and Polak (1998)}}\label{appendix:gkp}

Let $\mathcal W$ be continuous and satisfy \hyperlink{cond:posterior-local-utility}{(R)}. Taking $r=q^j$ in Lemma~\ref{lem:local-utility-derivatives} gives
\[
L_k(q^j)-L_j(q^j)=-D_{v(\cdot;X)}(q^j\,\|\,q^k).
\]
Theorem~\ref{thm:posterior-foc-direct}(a) therefore gives Blackwell monotonicity if and only if $v(q^j;X)$ lies weakly above the tangent to $v(\cdot;X)$ at $q^k$, for every pair of posteriors in every positive representation.

This yields the usual first-order condition for convexity of each $v(\cdot;X)$ on the relative interior. To see why the support test suffices at a fixed $X$, take any two interior beliefs $p,r$. Choose a Bayes-plausible law $Y$ assigning each a small positive weight and the remaining weight to an interior belief that restores the mean $\mu$. Apply the support test at $(1-t)X+tY$ and let $t\downarrow0$. Joint continuity gives $D_{v(\cdot;X)}(p\,\|\,r)\ge0$. Conversely, convexity gives the support test directly.

This is the fixed-prior counterpart of the convexity characterization in \citet[Proposition~1(i) and (iii)]{grant1998intrinsic}. Their mean-preserving split gives a chord inequality involving the values of the local utility. Our signal replacement moves the destination posterior toward the source posterior, which gives the tangent inequality above.

\subsection{Lehmann Monotonicity}\label{OA:local-utility-lehmann}

\begin{proposition}[Local utilities and Lehmann monotonicity]\label{prop:posterior-local-utility}
Let $\mathcal W:\mathcal X\to\R$ be continuous and satisfy \hyperlink{cond:posterior-local-utility}{(R)}. Then $\mathcal W$ is Lehmann monotone if and only if, for every $m\ge2$ and every MLR-ordered representation $(\tau,q)\in\mathcal X_{m,+}$, writing $X=\sum_j\tau^j\delta_{q^j}$, the following conditions hold at every cutoff $1\le\ell\le n$. If $1\le j\le m-1$ and $q_\ell^jq_{\ell+1}^{j+1}>q_{\ell+1}^jq_\ell^{j+1}$, then
\begin{equation}\label{ineq:local-utility-lehmann-lower}
D_{v(\cdot;X)}\bigl(\underline q^j_\ell\,\big\|\,q^{j+1}\bigr)\ \ge\ D_{v(\cdot;X)}\bigl(\underline q^j_\ell\,\big\|\,q^{j}\bigr).
\end{equation}
If $2\le j\le m$ and $q_{\ell-1}^{j-1}q_\ell^j>q_\ell^{j-1}q_{\ell-1}^j$, then
\begin{equation}\label{ineq:local-utility-lehmann-upper}
D_{v(\cdot;X)}\bigl(\overline q^j_\ell\,\big\|\,q^{j-1}\bigr)\ \ge\ D_{v(\cdot;X)}\bigl(\overline q^j_\ell\,\big\|\,q^{j}\bigr).
\end{equation}
The truncations and endpoint conventions are those of Theorem~\ref{thm:posterior-foc-direct}. All local utilities and posterior gradients are evaluated at the original random posterior $X$.
\end{proposition}

\begin{proof}
By Lemma~\ref{lem:local-utility-derivatives}, $\mathcal W$ is differentiable and the sign of $L_k(r)-L_j(r)$ is given by the difference of the two tangent gaps at $r$. For a lower-tail replacement, take $k=j+1$ and $r=\underline q^j_\ell$ in the lemma. Then $L_{j+1}(r)\le L_j(r)$ is exactly~\eqref{ineq:local-utility-lehmann-lower}. For an upper-tail replacement, taking $k=j-1$ and $r=\overline q^j_\ell$ gives~\eqref{ineq:local-utility-lehmann-upper}. These are the two conditions of Theorem~\ref{thm:posterior-foc-direct}(b) at their respective cutoffs, proving the equivalence.
\end{proof}

This extension has no counterpart in the Blackwell characterization of \citet{grant1998intrinsic}. Even for posterior-separable measures, its conditions are not implied by convexity (Section~\ref{sec:posterior-separable}).

\subsection{Scope}\label{OA:local-utility-scope}

Condition~\hyperlink{cond:posterior-local-utility}{(R)} suffices to identify the coordinate derivatives in Theorem~\ref{thm:posterior-foc-direct} with the local-utility coefficients in Lemma~\ref{lem:local-utility-derivatives}. Coordinate differentiability and the existence of local utilities are distinct requirements on a given measure. For example, $\mathcal W(X)=\int |q_1-\mu_1|\,X(dq)$ admits the convex local utility $v(q;X)=|q_1-\mu_1|$. Moving two equally weighted posteriors in opposite directions away from $\mu$ produces an absolute-value kink in its coordinate representation at the uninformative law.

Conversely, Theorem~\ref{thm:posterior-foc-direct} does not require a G\^ateaux representation. With $\operatorname{Var}(X)=\int\|q-\mu\|^2\,X(dq)$ and $\lambda>0$, consider
\begin{equation*}
\mathcal W(X)=\operatorname{Var}(X)+\lambda\operatorname{Var}(X)^{3/4}.
\end{equation*}
This measure is continuous and Blackwell monotone, since variance is Blackwell monotone and the transformation is increasing. Its coordinate representation is differentiable wherever variance is positive. At any positive representation $z=(\tau;\mu,\ldots,\mu)$ of $\delta_\mu$, a nearby representation $z'=(\tau',q')$ induces $X'=\sum_j\tau'^j\delta_{q'^j}$ with
\begin{equation*}
0\le\operatorname{Var}(X')=\sum_j\tau'^j\|q'^j-\mu\|^2\le\|z'-z\|^2.
\end{equation*}
The quotient below is therefore bounded by $\|z'-z\|+\lambda\|z'-z\|^{1/2}$, and hence
\begin{equation*}
\lim_{\substack{z'\to z\\z'\in\mathcal X_{m,+},\ z'\ne z}}
\frac{|\mathcal W_m(z')-\mathcal W_m(z)|}{\|z'-z\|}=0.
\end{equation*}
Thus the zero vector gives a first-order approximation at every such representation. Nevertheless, for any $Y\in\mathcal X_+$ with $\operatorname{Var}(Y)>0$,
\begin{equation*}
\begin{aligned}
\lim_{\alpha\downarrow0}
\frac{\mathcal W((1-\alpha)\delta_\mu+\alpha Y)-\mathcal W(\delta_\mu)}{\alpha}
&=\lim_{\alpha\downarrow0}\left[\operatorname{Var}(Y)+\lambda\alpha^{-1/4}\operatorname{Var}(Y)^{3/4}\right]=+\infty.
\end{aligned}
\end{equation*}
There is therefore no finite local utility at $\delta_\mu$. This comparison concerns the regularity of the specified measure: it represents the same preferences as variance, which does admit local utilities.

\section{Local Non-monotonicity of Nested-Logit Information Measures}\label{OA:nested-logit}

\citet{fosgerau2020discrete} introduce Bregman information measures\footnote{Despite the name, these measures differ from the posterior-separable measures of Section~\ref{sec:posterior-separable}, whose first-order conditions involve the Bregman divergence $D_v$. The nested-logit measure below depends on the labels of the signals through the nests, so it is neither posterior separable nor permutation invariant.} that support additive random utility models within a rational inattention framework. They show that the nested-logit specification need not be Blackwell monotone, using a permutation of alternatives. Since a permutation is a global operation, this observation does not by itself establish a failure of local monotonicity. A measure that remains monotone under small information-reducing perturbations could still be useful for local analysis. We apply our signal-replacement condition to show numerically that the nested-logit measure can fail this local requirement as well. The example also preserves the MLRP, establishing local failures of both Blackwell and Lehmann monotonicity.

Fix a full-support prior $\mu$ and a partition of the signals into nests, with $g_j$ denoting the nest containing $s_j$. The signal labels represent alternatives, and their nest membership is held fixed. For $P\in\mathcal E_{m,+}$, use the posterior notation of Section~\ref{sec:posterior} and let
\[
\tau^g=\sum_{j\in g}\tau^j,
\qquad
q^g=\frac{\sum_{j\in g}\tau^jq^j}{\tau^g}
\]
denote the probability of nest $g$ and the posterior conditional on that nest. With a common nesting parameter $\xi\in(0,1]$, the nested-logit information measure can be written as
\begin{equation}\label{eq:nested-logit-measure}
\mathcal I(P)=\sum_{j=1}^m\tau^j\left[
\xi D_{\mathrm{KL}}(q^j\,\|\,\mu)
+(1-\xi)D_{\mathrm{KL}}(q^{g_j}\,\|\,\mu)
\right].
\end{equation}
Thus, it is a weighted average of the mutual information between the state and the full signal, and that between the state and the nest. 

Take $\mu=(1/2,1/2)$, $\xi=1/2$, nests $\{1,2\}$ and $\{3\}$, and
\[
P=\begin{bmatrix}0.60&0.21&0.19\\0.20&0.39&0.41\end{bmatrix}.
\]
Replace signal $s_2$ by $s_3$ with probability $\epsilon$, obtaining
\[
Q_\epsilon=P+\epsilon P^{2\to3}
=\begin{bmatrix}
0.60&0.21(1-\epsilon)&0.19+0.21\epsilon\\
0.20&0.39(1-\epsilon)&0.41+0.39\epsilon
\end{bmatrix}.
\]
For every $0\le\epsilon<1$, this experiment is positive and its likelihood ratios are strictly increasing:
\[
\frac13<\frac{13}{7}<\frac{41+39\epsilon}{19+21\epsilon}.
\]
Hence the replacement stays within the MLRP domain, with $P\succeq_BQ_\epsilon$ and $P\succeq_LQ_\epsilon$. For $\epsilon>0$, it mixes signals with distinct posteriors and therefore strictly reduces Blackwell information. With two states, the Blackwell and Lehmann orders coincide on MLRP experiments, so the loss is strict in both orders.

At $P$, the relevant posteriors are $q^2=(7/20,13/20)$, $q^3=(19/60,41/60)$, and $q^{\{1,2\}}=(81/140,59/140)$. Differentiating~\eqref{eq:nested-logit-measure} along the replacement gives
\begin{equation*}
\begin{aligned}
\langle\nabla\mathcal I(P),P^{2\to3}\rangle
&=\frac3{10}\left[
\frac12D_{\mathrm{KL}}(q^2\,\|\,q^{\{1,2\}})
-D_{\mathrm{KL}}(q^2\,\|\,q^3)\right]\\
&\approx0.01510415>0.
\end{aligned}
\end{equation*}
Consequently, $\mathcal I(Q_\epsilon)>\mathcal I(P)$ for all sufficiently small positive $\epsilon$. The replacement reduces information in the full signal but strictly increases information in the nest, and the latter effect dominates.

The positive directional derivative violates the local signal-replacement condition in Theorem~\ref{thm:BM-general}: the measure strictly increases in response to a marginal loss of information. Thus, restricting attention to nearby experiments does not restore monotonicity at $P$, even within the MLRP domain. 

\end{document}